\documentclass[11pt,leqno,textwidth=10cm]{amsart}
\usepackage{amssymb}
\allowdisplaybreaks
\usepackage{mathrsfs}
\usepackage{mathtools}
\usepackage{abstract}
\usepackage[hidelinks]{hyperref}
\usepackage{marginnote}

\usepackage{etoolbox}
\makeatletter
\patchcmd{\@maketitle}{\newpage}{}{}{} 
\makeatother
\numberwithin{equation}{section}

\newtheorem{thm}{Theorem}[section]
\newtheorem{lem}[thm]{Lemma}
\newtheorem{prop}[thm]{Proposition}
\newtheorem{corol}[thm]{Corollary}

\theoremstyle{definition}
\newtheorem{defn}[thm]{Definition}

\theoremstyle{remark}
\newtheorem{rem}[thm]{Remark}

\newcommand{\tP}{\tilde{P}}
\newcommand{\tu}{\tilde{u}}
\newcommand{\hu}{\hat{u}^0}
\newcommand{\trho}{\tilde{\rho}}
\newcommand{\tEMT}{\tilde{T}}
\newcommand{\phFmet}{\bar{g}}
\newcommand{\phFnab}{\bar\nabla}
\newcommand{\bnab}{\hat{\nabla}}
\newcommand{\tlapse}{\tilde{N}}
\newcommand{\tshift}{\tilde{X}}
\newcommand{\phTmet}{\tilde{g}}

\newcommand{\delu}{\partial_{\mathbf{u}}}
\newcommand{\mL}{\mathcal{L}_{g, \gamma}}
\newcommand{\mD}{\hat{\Delta}_{g, \gamma}}
\newcommand{\EnergyDensity}{E}
\newcommand{\Ric}{\text{Ric}}
\newcommand{\Riem}{\text{Riem}}
\DeclareMathOperator{\ri}{Riem}
\newcommand{\Lgg}{L^2_{g,\gamma}}
\newcommand{\define}{:=}
\newcommand{\Edelu}{{E^g_{\delu,N-1}}}
\newcommand{\Etot}{E^g_{tot}}

\newcommand{\p}{\partial}
\newcommand{\Si}{\Sigma}
\newcommand{\hdel}{\hat{\partial}_0}
\newcommand{\di}{\textrm{d}}

\newcommand{\Eg}{{E^g}}
\newcommand{\hN}{\widehat{N}}
\newcommand{\alg}[1]{\begin{aligned}#1\end{aligned}}
\newcommand{\al}{\alpha}

\newcommand{\Chrtn}[4]{\/^{(#4)}\tilde{\Gamma}^{#1}_{#2 #3}}
\newcommand{\Chr}[3]{\Gamma^{#1}_{#2 #3}}

\newcommand{\na}{\nabla}

\newcommand{\mcr}[1]{\mathscr{#1}}
\newcommand{\eq}[1]{\begin{equation} #1 \end{equation}}

\begin{document}


\title[]{Slowly expanding  stable dust spacetimes}
\author[D.~Fajman, M.~Ofner, Z.~Wyatt]{David Fajman, Maximilian Ofner, Zoe Wyatt}

\date{\today}


\address{
\begin{tabular}[h]{l@{\extracolsep{8em}}l} 
David Fajman \& Maximilian Ofner,  & Zoe Wyatt, \\
Faculty of Physics, & Department of Pure Mathematics and  \\ 
University of Vienna, & Mathematical Statistics, \\
Boltzmanngasse 5, & Wilberforce Road,  \\
1090 Vienna, Austria. & Cambridge, CB3 0WB, U.K. \\
David.Fajman@univie.ac.at, \& & zoe.wyatt@maths.cam.ac.uk       \\
maximilian.ofner@univie.ac.at & 
\end{tabular}
}

\maketitle
\begin{abstract}

We establish the future nonlinear stability of a large class of FLRW models as solutions to the Einstein-Dust system. We consider the case of a vanishing cosmological constant, which in particular implies that the expansion rate of the respective models is linear i.e.~has zero acceleration. The resulting spacetimes are future globally regular. These solutions constitute the first generic class of future regular Einstein-Dust spacetimes not undergoing accelerated expansion and are thereby the slowest expanding generic family of future complete Einstein-Dust spacetimes currently known. 
\end{abstract}

\section{Introduction}
\subsection{General Relativistic Hydrodynamics}
The Einstein-relativistic Euler system (EES)
\eq{\label{EES}\alg{
R_{\mu\nu}-\frac12 R g_{\mu\nu} &=T_{\mu\nu}\\
\nabla^{\mu}T_{\mu\nu}&=0\\
T_{\mu\nu}&=(\rho+p)u_{\mu}u_{\nu}+p g_{\mu\nu}
}}
describes the dynamical evolution of a four-dimensional spacetime $(\mathcal{M}, g)$ containing a relativistic perfect fluid with pressure $p$, energy density $\rho$ and 4-velocity vector $u^\mu$. 
Perfect fluids compatible with relativity were one of the earliest matter models  considered in general relativity \cite{Ch31,OS39} and have been extensively studied in the context of general relativistic hydrodynamics with numerous applications ranging from astrophysics to cosmological evolution (see e.g. \cite{EllisCosmology, WeinbergCosmology}). 

The equations \eqref{EES} are supplemented by specifying an equation of state which relates the energy density and pressure $p=f(\rho)$. Different choices of the function $f$ encode different behaviour of the fluid. 
We focus in the following on the class of linear, barotropic equations of state, $p=c_S^2 \rho$, where the constant $c_S$ denotes the speed of sound of the fluid with $0\leq c_S\leq 1$. This equation of state contains well-known fluid models: for $c_S=0$, i.e.~$p=0$, \eqref{EES} reduces to the \emph{Einstein-Dust} system, the case $c_S=1/\sqrt 3$ is the \emph{Einstein-radiation fluid} system and $c_S=1$ is the \emph{Einstein-stiff fluid} system.
The main result of the present paper can be roughly stated as:

\begin{thm}\label{thm:rough}
All four-dimensional FLRW spacetime models with compact spatial slices and negative spatial Einstein geometry are future stable solutions of the Einstein-Dust system.
\end{thm}

To date, all known future stability results establishing the global existence, regularity and completeness of solutions to \eqref{EES} concern the regime of accelerated expansion. In such a setting, the fast decay rates of perturbations induced by the expansion have a strong regularization effect on the fluid. For slower expansion rates this effect becomes weaker and global regularity of solutions is less likely to hold.

Theorem \ref{thm:rough} establishes the \textbf{first nonlinear future stability result for a coupled Einstein-relativistic Euler system in the absence of accelerated expansion} and thereby initiates the study of the EES in the regime of non-accelerated expansion. Such a regime is also relevant in cosmology.  The epoch in the early universe, shortly after a hypothetical inflationary phase, is expected to not initially have exhibited accelerated expansion. It is this epoch, which is not covered by previous results on the EES, which we intend to make accessible by the research initiated in the present paper.

\subsection{Background and Previous Results}
For the sake of the following presentation we consider four-dimensional FLRW spacetimes of the form
\eq{\label{eq:IntroMetric}
\left((0,\infty)\times M,-\di t_c^2+a(t_c)^2\cdot  \gamma\right),
}
where $(M,\gamma)$ is a complete Riemannian manifold. We remind the reader that in the standard FLRW-models the spatial slices appearing in \eqref{eq:IntroMetric} have constant sectional curvature $k_M\in \{ -1,0,1\}$ and so are taken to be one of $\mathbb{R}^3, \mathbb{S}^3, \mathbb{H}^3$ or quotients thereof (eg, $\mathbb{T}^3$). 
We distinguish three classes of scale factors: $\ddot a(t_c)>0$ are referred to as \emph{accelerated expansion}, $\ddot{a}(t_c)<0$ as \emph{deccelerated expansion} and $\ddot a(t_c)=0$ as \emph{linear expansion}. We introduce the notion of \emph{power law inflation}, where $a(t_c)=(t_c)^{p}$ for $p>0$. Finally, we recall that a cosmological constant can be included by adding $+ \Lambda g_{\mu\nu} $ to the LHS of \eqref{EES}. In the following discussion only $\Lambda\geq 0$ is relevant. 
 
\subsubsection{Shock Formation}
Relavistic and non-relativistic fluids are well-known to form shocks in finite time. This was first observed in the general relativistic context by Oppenheimer and Snyder when they investigated the collapse of spherically symmetric clouds of dust \cite{OS39}. In the terminology of current stability analysis this constitutes the instability of Minkowski spacetime as a particular solution to the Einstein-Dust system. 
Note that a part of Minkowski spacetime corresponds to $M = \mathbb{R}^3, \gamma=\delta, a(t_c) = 1$ in \eqref{eq:IntroMetric}. 

More recently, Christodoulou's monograph \cite{Ch07} demonstrated that under a very general equation of state, the constant solutions to the relativistic Euler equations on a fixed background Minkowski space are unstable, i.e.~even without gravitational backreaction, fluids form shocks from arbitrarily small initial inhomogeneities in finite time. This suggests that Minkowski spacetime is unstable as a solution to the EES for a large class of equations of state.
Note also that Christodoulou's monograph  gave a detailed description of the nature of the fluid shock formation, thus providing a major extension beyond work of Sideris \cite{Sideris} on the non-relativistic Euler equations.  

\subsubsection{$\Lambda$-induced Accelerated Expansion and Stabilisation of Fluids}\label{intro:lambda-induced}
As is clear from the previous paragraphs, a powerful dispersive mechanism is required to regularise fluids and to prevent  finite-time shock formation. The prime example of such a mechanism comes from cosmological models exhibiting exponential expansion. Heuristically speaking,  a cosmological constant $\Lambda>0$  generates expansion of the form $a(t_c) \sim e^{Ht_c}$ where $H = \sqrt{\Lambda/3}$, for \emph{all} cases of the sectional curvature $k_M$. The cosmological constant creates damping terms in the equations of motion for the fluid, which dilutes the fluid and causes fluid lines to `stretch apart', thus preventing shock formation. 

This effect was first observed by Brauer, Rendall and Reula \cite{BRR94} who studied Newtonian cosmological models with  $\Lambda>0$ and a perfect fluid (albeit for a slightly different equation of state). They found that the regularising effect from the exponential expansion was strong enough to prevent shock formation for small inhomogeneities of initially uniformly quiet fluid states. See also late-time asymptotics work by Reula \cite{Reula99} and Rendall \cite{Rendall04}.

Moving to the fully coupled Einstein-relativistic Euler system, there has been much research concerning spacetimes undergoing exponential expansion. 
The first result is by Rodnianski and Speck \cite{RoSp13}, who proved future stability to irrotational perturbations of uniformly quiet fluids with $0<c_S<1/\sqrt 3$ on FLRW-spacetimes with underlying spatial manifold $M=\mathbb{T}^3$. The irrotational restriction was later removed by Speck in \cite{Sp12}. 
An alternative proof of future stability for these FLRW-background solutions was later given by Oliynyk \cite{Oliynyk16}, whose Fuchsian techniques were able to uniformly cover the cases $0<c_S\leq 1/\sqrt 3$.

Moving to the case of dust $c_S = 0$, stability for FLRW-spacetimes with underlying spatial manifold $M=\mathbb{T}^3$ was given by Had\v{z}i\'{c} and Speck \cite{HaSp15}, while more general spatial manifolds were considered by Friedrich \cite{Fr17} using his conformal method. Indeed the work by L\"ubbe and Valiente-Kroon \cite{LVK13} treated the radiation case with $c_S = 1/3$ using an extension of Friedrich's conformal method. Finally, we note that very recent work of Oliynyk \cite{Ol20} has established stability for ultra-radiation fluids $1/\sqrt 3<c_S\leq 1/\sqrt 2$  on a fixed, exponentially expanding spacetime.

\subsubsection{Alternative Mechanisms for Accelerated Expansion}
A cosmological constant is not the only known mechanism for generating solutions to Einstein’s equations with accelerated expansion. In work that predates the references of Subsection \ref{intro:lambda-induced}, Ringstr\"om  \cite{Ri08} considered the future global stability of a large class of solutions to the Einstein-nonlinear-scalar field system with a scalar field potential $V(\Phi)$ that satisfied $V(0)>0, V'(0)=0, V''(0)>0$. Roughly speaking $V(0)$ emulates the cosmological constant $\Lambda$ and so these spacetimes undergo accelerated expansion. Ringstr\"om \cite{Ri09} later considered alternative potentials $V(\Phi)$ which relaxed the rate of spacetime expansion to the class of accelerated power law inflation, which in our terminology corresponds to $p>1$. 

Note that in Ringstr\"om's papers the global spatial topology becomes irrelevant for the long time behaviour of cosmological spacetimes in the small data regime. This is in sharp contrast to the Einstein \emph{vacuum} equations where the spatial topology does affect the long-time behaviour. In this case, only the Milne geometry with a negative spatial curvature yields future eternally expanding cosmological models with precisely linear expansion rate.

Finally we note that the Chaplygin equation of state, which describes a fluid with negative pressure, can also generate sufficient spatial expansion to ensure future stability results for the coupled EES system \cite{LeFloch-Wei}. 

\subsubsection{Critical Expansion Rates}\label{intro:CritExp}
Interpolating between Minkowski space (which can be considered as a cosmological spacetime with non-compact slices and no expansion) and exponentially expanding spacetimes, it is clear that there must be a transition between shock formation and stability. 
To investigate the expansion rate for which this transition occurs, and how it depends on the equation of state, it is useful to study the stabilisation of fluids on fixed Lorentzian geometries obeying power-law inflation. We consider $M = \mathbb{T}^3$ with $a(t_c) = (t_c)^p$ for $p>0$. The following table summarises some of the main results concerning linear equation of states $p= c_S^2 \rho$ from \cite{Sp13, FOW21}, see also \cite{Wei}:

\begin{center}
\begin{tabular}{c|  c |c |c  | c}
Case & Power-law rate &  Range of $c_S$ & Behaviour & Reference\\ 
\hline
No. 1 & $p>1$ & $0<c_S<1/\sqrt 3$ & Stable & \cite{Sp13} \\  
No. 2 & $p=1$ & $c_S = 1/\sqrt 3$ & Shocks  & \cite{Sp13} \\
No. 3 & $p=1$ & $0<c_S<1/\sqrt 3$ & Stable (irrot.)  & \cite{FOW21} \\
No. 4 & $p>\frac12$ & $c_S = 0$ & Stable & \cite{Sp13}
\end{tabular}
\end{center}  
 
In combination, these results indicate that in spacetimes undergoing power-law inflation whether shocks form from small data depends on the equation of state and, in particular in the linear case, on the speed of sound. Indeed the literature suggests that \emph{slower speeds of sound reduce the tendency of shock formation}. 
For the particular case of dust ($c_S=0$), case No. 4 shows that shocks are avoided even in deccelerating spacetimes with scale factors $a(t)=t^{1/2+\delta}$ for $\delta>0$.

\subsection{Main Results}
In the present paper we consider the Einstein-Dust system in the regime of linear expansion. For the linearly expanding case, Cases No. 2 and 3 above show that even in the absence of backreaction the speed of sound determines whether shocks form or not. We prove that for the case of dust ($c_S=0$) shock formation does not occur under the full gravity-fluid dynamics. 

Our background geometry is that of the Milne model, which generalises the $k_M = -1$ FLRW vacuum spacetimes.
Let $(M, \gamma)$ be a closed, connected, orientable three-dimensional manifold  admitting a Riemannian Einstein metric $\gamma$ with negative Einstein constant. After rescaling, we suppose that 
$$ \Ric[\gamma] = -\frac29 \gamma .$$
The generalised Milne spacetime is the Lorentz cone spacetime $\mathcal{M} = (0,\infty) \times M$ with metric
\eq{\notag
g_M \define -\di t_c^2 + \frac{t_c^2}{9}\gamma_{ab} \di x^a \di x^b.
}
The spacetime $(\mathcal{M}, g_M)$ is globally hyperbolic and a solution to the four-dimensional vacuum Einstein equations.
We formulate the main theorem using terminology introduced in Section \ref{sec:FunctionsSpaces}. 
We let $\mathscr B_\varepsilon^{j,k,l,m}\big(\tfrac{t_0^2}{9}\gamma,-\tfrac{t_0}9 \gamma,0,0\big)$ denote the ball of radius $\varepsilon$ in the space $H^j\times H^k\times H^l\times H^m$ centred at $\big(\tfrac{t_0^2}{9}\gamma,-\tfrac{t_0}9 \gamma,0,0\big)$. 
Our main theorem reads:

\begin{thm}\label{thm-1}
Let $(\mathcal{M},g_M)$ be as above.
Let $\varepsilon>0$ and $(g_0,k_0,\rho_0,u_0)$ be initial data for the Einstein-Dust system at $t_c=t_0$ such that 
\eq{\notag
(g_0,k_0,\rho_0,u_0)\in \mathscr B_\varepsilon^{6,5,4,5}\left(\frac{t_0^2}{9}\gamma,-\frac{t_0}9 \gamma,0,0\right).
}
Then, for $\varepsilon$ sufficiently small the corresponding future development under the Einstein-Dust system is future complete and admits a CMC foliation labelled by $\tau\in[\tau_0, 0)$ such that the induced metric and second fundamental form on constant CMC slices converge as
\eq{\notag
(\tau^2 g,\tau k)\rightarrow \left(  \gamma,\frac{1}{3} \gamma\right) \mbox{ as } \tau \nearrow 0 \quad \text{i.e. as } t_c\nearrow \infty.
}
\end{thm}

\begin{rem}
The Milne model is known to be a stable solution solution to the Einstein vacuum
equations \cite{AM11}, the Einstein massive-Vlasov equations \cite{AF20}, the coupled Einstein-Maxwell-scalar field system arising from a Kaluza-Klein reduction \cite{BFK19}, and the
Einstein Klein-Gordon equations \cite{Wang-KG, FW21}. 
\end{rem}

\begin{rem}
Negative spatial curvature is crucial as spherical or toroidal spatial topologies would lead to recollapsing or slowly expanding matter dominated solutions, respectively. The asymptotic behaviour of the solutions in the theorem coincide with the corresponding vacuum solutions.
\end{rem}

\subsubsection{Structure and Key Novelties in the Proof}
\label{intro:Novelties}
The proof of Theorem \ref{thm-1} consists of three major parts: $(i)$ energy estimates for the perturbation of the spacetime geometry with sources given by the dust variables, $(ii)$ energy estimates for the dust variables in the perturbed spacetime geometry and $(iii)$ a bootstrap argument based on both sets of energy estimates establishing global existence and asymptotic behaviour. This rough approach is standard in the literature on the Milne stability problem (see e.g. \cite{AM11, AF20, BFK19}), however, for the Einstein dust system there are crucial difficulties caused by a regularity problem inherent to the dust equations, which turns out to affect all parts of the argument. We outline the difficulties and how these are overcome in the following.

To control the perturbed spacetime geometry throughout the evolution we use a CMC time-foliation in combination with a spatial-harmonic gauge \cite{AM03}. The existence of such a foliation for small perturbations of negative Einstein spaces is non-trivial but standard \cite{FK20}. The Einstein equations then take the form  of an elliptic-hyperbolic system (see \eqref{EoM})  where the lapse and shift are determined by elliptic PDEs with sources given in terms of metric, second fundamental form and the dust variables. This elliptic system  provides Sobolev estimates for the lapse and shift. 

The core idea to establish decay for the geometric variables in previous works on  Milne stability is a corrected energy ($\Eg_k$ in Definition \ref{defn:GeometricEnergy}) based on the modified Einstein-operator ($\mL$ in Definition \ref{defn:mL}) of the spatial Einstein geometry \cite{AM11, AF20}. For the Einstein-Dust system we must deviate  from this standard approach due to a regularity issue from the dust model, which in turn affects all parts of the proof. 

When expanded, the equations of motion for the dust variables take a form where the source term of the evolution equation for the energy density contains the spatial divergence of the fluid velocity (see \eqref{eq:EoM-fluid}). Consequently, the fluid energy density can be controlled only in one order of regularity below the order of regularity of the fluid velocity.
From the perspective of the Einstein equations this is very problematic as both components of the dust, energy density and fluid velocity, appear at the same order of regularity as source terms of the Einstein equations. As such, they are required to be controlled in suitable Sobolev spaces at the same order as the second fundamental form. Due to the required high regularity of the fluid velocity discussed previously, the velocity then needs to be controlled one order above the second fundamental form. However, the equation of motion for the fluid velocity requires the second fundamental form at the same order of regularity as the velocity itself  (see \eqref{eq:EoM-fluid}). This apparent inconsistency prevents one from establishing a standard and straightforward regularity hierarchy to analyse the fully coupled nonlinear system.

An approach to circumvent this issue has been introduced by Had\v{z}i\'{c} and Speck in \cite{HaSp15} and is modified in the present paper. The central idea is to use a fluid derivative $\delu \sim u^\alpha\nabla_\alpha [g_M]$ as a differential operator in the energies for the perturbations of the metric and second fundemental form ($\Edelu$ in Definition \ref{defn:delu-geom-energy}). At highest order of regularity, say $N$, where the loss of derivatives prevents the closure of the system of estimates, the Einstein equations are commuted with $N-1$ spatial derivatives and one fluid derivative.  When this derivative acts on the dust source terms in the Einstein equations,  in the subsequent calculations for the energy estimates, the equations of motion of the dust variables are used as constraint equations replacing $\delu \rho$ and $\delu  u^j$. In this way, no derivatives are lost and the corresponding auxiliary energies $\Edelu$ for the geometric variables can be estimated in terms of dust variables of one order of regularity below the expected one. 

In a follow-up step, we need to show that the auxiliary geometric energies $\Edelu$ in fact control the actual top-order regularity norms of the geometric variables. This is achieved by rewriting the wave-type evolution equation for the metric and second fundamental in terms of an elliptic part and certain mixed spatial and fluid derivative operators (see Proposition \ref{prop:Elliptic-Est}). 
Consequently, the auxiliary energies of the first step provide top-order estimates on the geometric variables (see Corollary \ref{corol:Coercive-top-order-geom})  and an overall strategy to close the estimates.

Two final major regularity issues arise when proving energy estimates for the auxiliary energies $\Edelu$ however. We end up needing to estimate one fluid derivative and a critical number of spatial derivatives on certain terms involving the lapse and shift, and we cannot commute the $\delu$ operator past the spatial derivatives without exceeding the assumed regularity of the fluid spatial velocity. 

To circumvent this problem, we only commute past some of the derivatives and instead derive two auxiliary estimates using the elliptic equations \eqref{eq:EoM-lapse-shift} for the lapse and shift. In the estimate on the lapse term (see Proposition \ref{prop:delu-commuted-lapse}) we crucially use the equations of motion of the dust variables  to replace a certain matter term $\delu \eta$ as a constraint, thus avoiding derivative loss. The estimate for the shift term (see Proposition \ref{prop:Lie-deriv-Shift}) proceeds differently, relying on a \emph{remarkable combination} of commutator estimates, the Bianchi identity and the Einstein equations in the CMCSH gauge.

\subsubsection{Final Remarks}
In the regime of non-accelerated expansion, the work \cite{BRR94} indicates that the backreaction between the fluid and the geometry cannot be ignored. The authors consider the case of dust with a Newtonian backreaction, finding that shocks form for arbitrarily small initial data in the regime where the homogeneous background spacetime, which is perturbed, expands like $a(t)=t^{2/3}$. 
This contrasts \emph{noteably} with case No. 5 above which does not include backreaction. While \cite{BRR94} concerns only Newtonian dynamics, it is nevertheless a fair indication that the fully coupled dynamics under the Einstein-fluid system will likely lead to the formation of shocks.
Continuing this line of reasoning, we note that although the work \cite{FOW21} also treated linear expansion, the full coupling between gravity and fluid makes our present work highly nontrivial.  Indeed the issues highlighted on the previous Section \ref{intro:Novelties} are indicative of the substantial technical difficulties that arise in the fully coupled EES.

Finally, it is interesting to recall that Sachs and Wolfe derived a linear \emph{in}stability result for the Einstein-Dust equations with $\Lambda = 0$, however their metric had underlying spatial manifold $M = \mathbb{R}^3$  \cite{SachsWolfe}. The fluid plays a major dynamical role in these flat FLRW models. Nevertheless one gleans the importance of the negatively curved spatial slices appearing in our nonlinear stability result.

\subsection{Paper Outline}  
In Section \ref{sec:EoM} we introduce the system of equations and perform a natural rescaling of the variables. In Section \ref{sec:Prelims} we introduce function spaces and energy functionals controlling the metric perturbation and shear tensor. 

The main theorem is proved using continuous induction. In Section \ref{sec:LocalandBootstrap} we discuss the local existence theory and initiate the bootstrap argument. 
The remainder of the paper, beginning with Section \ref{sec:FurtherPrelims}, treats the individual estimates necessary to close the bootstrap argument. Section \ref{sec:FurtherPrelims} gathers various auxiliary estimates, which are used in later sections. Among those are estimates on the source terms of the evolution equations, estimates on the dust-derivative acting on various quantities, commutators of the dust derivative and other operators and estimates on high derivatives combining the dust-derivative and other operators.

Section \ref{sec:EllipticEstimate} derives the elliptic estimate for the Einstein operator and the evolution equations, which is crucial to turn estimates in terms of the dust derivatives into those in terms of standard energies. These estimates are then given subsequently. In Section \ref{sec:LapseShift} we provide the estimates on lapse function and shift vector field. A crucial set of lapse and shift estimates on highest order of regularity, involving also the dust derivative, are given here too.
In Section \ref{sec:GeomEnergyEst} we derive the central top-order energy estimate for the auxiliary energy controlling the geometric perturbations. In Section \ref{sec:EnergyEstMatter} we derive the estimates for the dust variables and in Section \ref{sec:EndBoot} we close the bootstrap. 

\subsection*{Acknowledgements}
D.~F.~ and M.~O.~ have been supported by the Austrian Science Fund (FWF): [P 34313-N].

\section{Equations of Motion} \label{sec:EoM}
\subsection{The Einstein--Dust system}
The Einstein--relativistic Euler system reads
\eq{\alg{\label{eq:Eulerform1}
R_{\mu\nu}[\phFmet]-\frac12R[\phFmet]\phFmet_{\mu\nu}&=2\tEMT_{\mu\nu},\\
\phFnab_{\mu}\tEMT^{\mu\nu}&=0,\\
\tEMT^{\mu\nu}&=(\trho+\tP) \tu^{\mu}\tu^\nu+\tP\phFmet^{\mu\nu} ,
}}
where we set $c=1$ and $4\pi G = 1$. 
We use $\phFnab$ to denote the Levi-Civita connection of the physical metric $\phFmet$. The four-velocity of the fluid $\tu^\mu$ is a future-directed timelike vectorfield normalised by 
\eq{\label{eq:FluidNormalisation}
\phFmet_{\mu\nu} \tu^\mu \tu^\nu=-1 \,.
}
We assume a linear, barytropic fluid equation of state $\tP= c_S^2 \trho$ where $c_S \geq 0$ is a constant, and $\tP\geq 0$ and $\trho \geq 0$ denote the pressure and energy density respectively. In the present paper, we \textbf{restrict ourselves to dust}, which means we set
\eq{\notag
c_S^2 \define 0 \,.
}

The fluid equations in \eqref{eq:Eulerform1} can equivalently (for $\trho>0$) be written as
\eq{\alg{\label{eq:Eulerform2}
\tu^{\alpha}\phFnab_{\alpha}\ln \trho+\phFnab_\alpha \tu^{\al}&=0\,,\qquad
\tu^{\alpha}\phFnab_{\alpha} \tu^{\mu}=0\,.
}}
The system \eqref{eq:Eulerform2} is overdetermined in the sense that $\tu^0$ can be determined from the other fluid velocity components via \eqref{eq:FluidNormalisation}. 

We will study the Einstein-Dust equations using the following ADM ansatz for the metric
\eq{\label{eq:ADM}
\phFmet = -\tlapse^2 \di t^2 +\phTmet_{ab}(\di x^a+\tshift^a\di t)(\di x^b + \tshift^b \di t).
}
Note that $\phFmet_{ab}=\phTmet_{ab}$ but in general $\phFmet^{ab}\neq\phTmet^{ab}$. 
On $t$=constant slices, we let $\tau$ be the trace of the second fundamental form $\tilde k$ with respect to $\tilde g$ and define $\Sigma$ to be the trace-free part of $\tilde k$. That is,
\eq{\alg{\notag
\tau&\define\text{tr}_{\tilde g} \tilde k = \phTmet^{ab} \tilde{k}_{ab},\qquad
\tilde k\define\tilde\Sigma+\tfrac13 \tau \tilde g.
}}

We use Roman letters ($a,b,i,j...$) to denote spatial indices. 
Let $\nabla$ denote the Levi-Civita connection of the spatial metric $\phTmet$. Using \eqref{eq:ADM} the Christoffel symbols of the 4-metric $\phFmet$ become (see e.g. \cite{Re08})
\eq{\alg{\notag
 \Chrtn 0004&=\tlapse^{-1}( \p_t \tlapse + \tshift^a\nabla_a\tlapse - \tilde{k}_{ab}\tshift^a\tshift^b ),\qquad
 &\Chrtn 0ab4&=-\tlapse^{-1}\tilde{k}_{ab},
 \\
\Chrtn 0a04&=\tlapse^{-1}(\nabla_a\tlapse - \tilde{k}_{ab}\tshift^b),
 \quad
&\Chrtn abc4 &=\Gamma^a_{bc}[\phTmet]+\tlapse^{-1}\tilde k_{bc}\tshift^a,\\
 \Chrtn a0b4&=-\tlapse\tilde k^a_b+\nabla_b\tshift^a-\tlapse^{-1}\tshift^a\nabla_b\tlapse+\tlapse^{-1}\tilde k_{bc}\tshift^c\tshift^a,\\
 \Chrtn a004&=\p_t\tshift^a+\tshift^b\na_b\tshift^a-2\tlapse\tilde k_c^a\tshift^c+\tilde{N}\nabla^a\tilde{N}\\
&\quad-\tlapse^{-1}(\p_t\tlapse+\tshift^b\nabla_b\tlapse-\tilde k_{bc}\tshift^b\tshift^c)\tshift^a.
}}
Noting the above, the fluid equations \eqref{eq:Eulerform2} reduce to
\begin{align*}
\tu^{\alpha}\p_\alpha \ln \trho+\p_\alpha \tu^{\alpha}+\Chrtn \alpha\alpha\nu4 \tu^\nu&=0,\qquad
\tu^{\alpha}\p_{\alpha} \tu^{\mu}+\tu^\alpha {\Chrtn \mu\alpha\nu4 \tu^\nu} =0.
\end{align*}

\subsection{The Rescaled Einstein--Dust system in CMCSH gauge}
Following the work of Andersson and Moncrief \cite{AM03, AM11}, we hereon impose the CMCSH gauge which foliates by surfaces of constant mean curvature, taking advantage of the fact that on the Milne background  $\tau = \phTmet^{ab}\tilde{k}_{ab} = -3/t_c$. 
\begin{defn}[CMCSH gauge]\label{gauges}
\eq{\notag
\begin{split}
t=\tau,\qquad
H^a\define \tilde{g}^{cb}(\Gamma[\tilde{g}]^{a}_{cb}-\Gamma[\gamma]^a_{cb})&=0.
\end{split}
}
\end{defn}

We next rescale our variables with respect to the mean curvature $\tau$.

\begin{defn}[Rescaled variables $(g_{ab}, N, X^a, \Sigma_{ab}, u^a, u^0, \rho, \hN, \hu)$ and logarithmic time $T$]\label{defn:rescaling}

The rescaled geometric variables are defined as
\begin{subequations}
\eq{\label{GeometricRescaling}
\begin{array}{rlrl}
g_{ab}&\define\tau^2\tilde{g}_{ab},&g^{ab}\define&(\tau^2)^{-1}\tilde{g}^{ab},\\
N&\define\tau^2\tlapse, &X^a\define&\tau \tshift^a,\\
\Sigma_{ab}&\define\tau\tilde \Sigma_{ab}.
\end{array}
}
Note $\tshift_a = \phTmet_{ab}\tshift^b = \tau^{-3} X_a$. Let $\tu^0\define \tu^\tau$. The rescaled matter variables are defined as
\eq{\label{FluidRescaling}
u^a\define\tau^{-2}\tu^a , \qquad
\rho\define|\tau|^{-3}\trho, \qquad u^0 \define \tau^{-2} \tu^0 .
}
\end{subequations}
Denote $\widehat N\define\frac N3-1$ and $\hu \define u^0-1/3$.
Finally we define the logarithmic time 
\eq{\notag
T\define-\ln(\tau/(e\tau_0)),
}
which satisfies  $\p_T=-\tau \p_\tau$. 
\end{defn}
The above definition means we have the following ranges $\tau_0 \leq \tau\nearrow 0$ and $1 \leq T \nearrow \infty$ where $\tau\nearrow0$ corresponds to the direction of cosmological expansion (i.e. $t_c \nearrow \infty$). 

\begin{lem}\label{lem:Renorm-u0}
The normalisation condition \eqref{eq:FluidNormalisation} implies
\eq{\alg{\notag
u^0 
&= \frac{1}{(N^2 - X_a X^a)}\Big(\tau X_a u^a+ \Big[\tau^{2}(X_a u^a)^2 + (N^2 - X_aX^a)(\tau^{2}g_{ab} u^a u^b+1)\Big]^{1/2}\Big) .
}}
\end{lem}
\begin{proof}
Using \eqref{eq:FluidNormalisation} we have
\eq{\notag
0 = (-\tlapse^2 + \tshift_a\tshift^a)(\tu^0)^2 + 2\tshift_a \tu^a \tu^0 + \phTmet_{ab}\tu^a\tu^b+1.
}
This is a quadratic equation in $\tu^0$. The roots are
\eq{\notag
\tu^0 = \frac{1}{2(-\tlapse^2 + \tshift_a\tshift^a)}\Big( -2\tshift_a \tu^a\pm \Big[4(\tshift_a \tu^a)^2 - 4(-\tlapse^2 + \tshift_a\tshift^a)(\phTmet_{ab}\tu^a\tu^b+1)\Big]^{1/2}\Big) .
}
Applying the rescalings from Definition \ref{defn:rescaling} we find
\eq{\notag
\tu^0 = \frac{\tau^4}{(N^2 - X_a X^a)}\Big( \tau^{-1}X_a u^a\mp \Big[\tau^{-2}(X_a u^a)^2 + (N^2 - X_aX^a)(\tau^{-2}g_{ab} u^a u^b+\tau^{-4})\Big]^{1/2}\Big) .
}
Hence, we introduce the rescaled quantity $u^0 = \tau^{-2} \tu^0 $ for the larger root.
\end{proof}

 Let $\nabla, \bnab$ denote the Levi-Civita connection of the Riemannian metrics $g, \gamma$ respectively.
The Christoffel symbols of $\phFmet$ now become (see e.g. \cite{AF20})
\eq{\alg{\notag
 \Chrtn abc4&=\Chr abc[g]+\Gamma_{bc}X^a,\quad
 &\Chrtn a004 &=\tau^{-2}\Gamma^a,\\
 \Chrtn a0b4&=\tau^{-1}\left(-\delta_b^a+\Gamma_b^a\right),\quad
 &\Chrtn 0004&=\tau^{-1}(-2+\Gamma_R),\\
  \Chrtn 0ab4&=\tau\Gamma_{ab},\quad
 &\Chrtn 00a4&=\Gamma_a,\\
}}
where we have introduced the following rescaled geometric components:
 
\begin{defn}[Rescaled Christoffel components $\Gamma^a, \Gamma^a_b, \Gamma_R, \Gamma_{ab}, \overset{\circ}{\Gamma}{}^a$]\label{Def-Gammas-rescaled}
\eq{\alg{\notag
\Gamma^a&\define-\p_T X^a-X^a-2\widehat{N} X^a+X^b\na_b X^a-2 N \Si_c^a X^c+{N}{\na}^a{N}\\
 &\quad\,+ \Big(N^{-1}\p_T N-N^{-1} X^b\nabla_b N+N^{-1}\left(\Si_{bc}+\tfrac13g_{bc}\right) X^b X^c\Big)X^a, \\
\Gamma^a_b&\define-N\Si^a_b-\delta_b^a\widehat{N}+\nabla_b X^a- N^{-1} X^a\nabla_b N+ N^{-1}\left( \Si_{bc}+\tfrac13g_{bc}\right) X^c X^a,
\\
\Gamma_R&\define N^{-1}\big(-\p_T N+X^a\nabla_a N-(\Si_{ab}+\tfrac13g_{ab})X^aX^b\big),
\quad
\overset{\circ}{\Gamma}{}^a \define \Gamma^a - N \nabla^a N, \\
\Gamma_{ab}&=-N^{-1}(\Sigma_{ab}+\frac13g_{ab}),\qquad
\Gamma_a\define N^{-1}(\nabla_a N-(\Sigma_{ab}+\tfrac13g_{ab})X^b).
}}
\end{defn}

\begin{rem}[Background solutions]\label{rem:BackgroundSolution}
The rescaled background (B) Milne geometry written in CMCSH gauge is
\eq{\notag
(g_{ab}, \Sigma_{ab}, N, X^a)|_B\equiv (\gamma, 0, 3, 0).
}
Furthermore, 
\eq{\notag
(\Gamma^a, \Gamma^a_b, \Gamma_R, \Gamma_a) |_B \equiv 0, \quad \Gamma_{ab} |_B \equiv -\tfrac19 \gamma_{ab}.
}
Let $\rho_0'>0$ be a constant.
The background, uniformly quiet fluid solution in CMCSH gauge is
\eq{\notag
(u^0, u^i, \rho)|_B = (\tfrac13,0,\rho_0').
}
See also Appendix \ref{appendix:Background}.
\end{rem}

We next evaluate certain energy momentum and matter source terms arising from the dust. 
\begin{defn}[Matter source terms $\EnergyDensity, \jmath^a, \eta, S_{ab}, \underline{T}^{ab}$]\label{Def-resc-matter}
\eq{\alg{\notag
\EnergyDensity &\define \rho (u^0)^2 N^2,\qquad
&\jmath^a &\define\rho N u^0 u^a,  \\
\eta &\define E+ \rho g_{ab}(u^0 X^a+ \tau u^a)(u^0 X^b + \tau u^b), \\
S_{ab} &\define \rho( u^0 X_a+ \tau u_a )(u^0 X_b+ \tau u_b) + \tfrac{1}{2}\rho g_{ab},   \qquad
&\underline{T}^{ab} &\define \rho u^a u^b.
}}
For further details on these definitions see Appendix \ref{appendix:Matter}.
\end{defn}

\begin{defn}[Matter source terms $F_{u^j}, F_{u^0}, F_\rho$]\label{Defs-fluid-source}
\eq{\alg{\label{eq:Defs-fluid-source}
F_{u^j} &\define \tau^{-1}(u^0)^2\overset{\circ}{\Gamma}{}^j + ({\Chr jkl}[g]-{\Chr jkl}[\gamma])u^ku^l+2u^0u^i\Gamma^j_i+\tau u^k u^i \Gamma_{ki}X^j,
\\
F_{u^0} &\define (u^0)^2\Gamma_R+2\tau u^ju^0\Gamma_j + \tau^2 u^k u^j \Gamma_{kj},
\\
F_\rho &\define \tau \nabla_i u^i + \Gamma^i_i u^0 - \tau \Gamma_j u^j + \tau \Gamma_{ik}X^i u^k  - \tau \frac{u^j}{u^0} \nabla_j u^0 -\tau^2\frac{ u^k u^j}{u^0} \Gamma_{kj}.
}}
\end{defn}

\subsection{Equations of Motion}
Bringing together all the previous notation, as well as using the general equations presented in \cite{AF20}, the equations of motion for the Einstein-Dust system in CMCSH gauge are the following. We have two constraint equations: 
\begin{subequations}\label{EoM}
\eq{\alg{\label{EoM-constraints}
\text{R}(g)-|\Sigma|_g^2+\tfrac{2}{3}&= 
	4 \tau \EnergyDensity ,
\\
\nabla^a \Sigma_{ab} &=
	2 \tau^2 \jmath_b ,
}}
and two elliptic equations for the lapse and shift variables
\eq{\alg{\label{eq:EoM-lapse-shift}
(\Delta - \tfrac{1}{3})N &= 
	N \left( |\Sigma|_g^2 - \tau \eta \right)-1, 
\\
\Delta X^a + \Ric[g]^a{}_b X^b &=
	2 \nabla_b N \Sigma^{ba} - \nabla^a \hN+ 2 N \tau^2 \jmath^a 
	- (2N \Sigma^{bc} - \nabla^b X^c)(\Gamma[g]^a_{bc} -  \Gamma[\gamma]^a_{bc}).
}}
We also have evolution equations for the induced metric and trace-free part of the second fundamental form 
\eq{\alg{\label{eq:EoM-pT-g-Sigma}
\p_T g_{ab} &=
	2N \Sigma_{ab} + 2\hN g_{ab} - \mcr{L}_X g_{ab}, 
\\
\p_T \Sigma_{ab} &=
	-2\Sigma_{ab} - N(\Ric[g]_{ab} +\tfrac{2}{9}g_{ab} ) + \nabla_a \nabla_b N + 2N \Sigma_{ac} \Sigma^c_b 
\\
& \quad 
	-\tfrac{1}{3} \hN  g_{ab} - \hN \Sigma_{ab} - \mcr{L}_X \Sigma_{ab} + N \tau S_{ab},
}}
and, finally, evolution equations for the fluid components:
\eq{\alg{\label{eq:EoM-fluid}
u^0\p_Tu^j
&= \tau u^a \nabla_a u^j + 
\tau^{-1}(u^0)^2 N\nabla^j N+F_{u^j},
\\ 
u^0\p_T u^0
& =\tau u^a \nabla_a u^0 + F_{u^0} ,
\\
u^0 \p_T \rho
&= \tau u^a \nabla_a \rho + \rho F_\rho .
}}
\end{subequations}

Using notation from \cite{AM03, AM11} we introduce new variables which allow us to rewrite \eqref{eq:EoM-pT-g-Sigma}.

\begin{defn}[Perturbation variables $h, v, w$ and geometric source terms $F_h, F_v$]\label{defn:Fh-Fv}
Define the variables
\eq{\notag
h_{ab}\define g_{ab}-\gamma_{ab}, \quad v_{ab}\define6\Sigma_{ab}, \quad w \define N/3, }
and the geometric source terms
\eq{\alg{\notag
(F_h )_{ab}&\define 2\hN g_{ab} + h_{ac}\bnab_b X^c + h_{cb}\bnab_a X^c ,
\\
(F_v )_{ab}&\define \nabla_a \nabla_b N + 2N \Sigma_{ac} \Sigma^c_b-\tfrac{1}{3} \hN  g_{ab} - \hN \Sigma_{ab}+ N \tau S_{ab} - v_{ac}\bnab_b X^c - v_{cb}\bnab_a X^c.
}}
\end{defn}

We start with  the following identity from \cite{AM03}:
\eq{\notag
\mcr{L}_X g_{ab} = X^c \bnab_c g_{ab} +g_{ac}\bnab_b X^c + g_{cb}\bnab_a X^c.
}
Due to rigidity properties of negative Einstein manifolds in three spatial dimensions (see e.g. \cite[\textsection 1.1]{AM11}), we have $\p_T \gamma = 0$. Thus the equations \eqref{eq:EoM-pT-g-Sigma} reduce to
\eq{\alg{\label{eq:EoM-hv}
\p_T h_{ab} &=
	wv_{ab}  - X^m \bnab_m h_{ab}  + F_h,
\\
\p_T v_{ab} &=
	-2v_{ab}  - 9 w\mL h_{ab}   
	 - X^c\bnab_c v_{ab}  +6F_v. }}
In Section \ref{sec:GeomEnergyEst} we will also write the first equation in \eqref{eq:EoM-hv} as 
\eq{\label{eq:EoM-h2}
\p_T h_{ab} =	wv_{ab} + 2\hN g_{ab} - (\mcr{L}_X g)_{ab} = wv_{ab} + 2\hN g_{ab} - g_{am}\nabla_b X^m - g_{bm} \nabla_a X^m.
}

Hereon we use the differential equations \eqref{eq:EoM-lapse-shift}, \eqref{eq:EoM-fluid} and \eqref{eq:EoM-hv} to analyse the solutions to our Einstein-Dust system. 

\section{Preliminary Definitions}\label{sec:Prelims}
In this section we present several preliminary definitions concerning Sobolev spaces, norms, elliptic estimates and energy functionals. All of this is standard except for Definitions \ref{defn:delu} and \ref{defn:delu-geom-energy} where we introduce the fluid derivative $\delu$ and then the energy functionals for the geometric variables involving this fluid derivative. 
 
\subsection{Function Spaces and Norms}
\label{sec:FunctionsSpaces}
\begin{defn}[$\Lgg$-inner product]
Let $\mu_g = \sqrt{\det g}$ denote the volume element on $(M, g)$, similarly for $\mu_\gamma$.  Let $V, P$ be $(0,2)$-tensors on $M$.  Define an inner product  by
\eq{\notag
\langle V, P\rangle_\gamma \define V_{ij} P_{kl}\gamma^{ik}\gamma^{jl},
}
and define a mixed $L^2$-scalar product
\eq{\notag
(V, P)_{L^2(g, \gamma)} \define \int_M \langle  V, P\rangle_\gamma \, \mu_g ,
}
with corresponding norm
$
\| V\|_{\Lgg}^2 \define (V, V)_{L^2(g, \gamma)}.
$
\end{defn}

The following definition follows notation first introduced in \cite{AM11}.
\begin{defn}[$\ri{[}\gamma{]}\circ$ and $\mL$]\label{defn:mL}
Let $V$ be a symmetric $(0,2)$-tensor on $M$. Define the tensorial contraction
\eq{\notag
(\ri[\gamma]\circ V)_{ij} \define \ri[\gamma]_{iajb}\gamma^{aa'} \gamma^{bb'}V_{a'b'} \,,
}
where, following the convention of \cite{AM03}, the Riemann tensor is defined by $[\bnab_a,\bnab_i]V_b =(\bnab_a\bnab_i - \bnab_i\bnab_a) V_b\define -\Riem[\gamma]^c{}_{bai}V_c$. 
Define the following differential operators
\eq{\alg{\notag
\mD V_{ij} &\define (\sqrt{\det g})^{-1}\bnab_a\big(\sqrt{\det g}\cdot  g^{ab} \bnab_b V_{ij}\big), \\ 
\mL V_{ij} &\define -\hat{\Delta}_{g, \gamma} V_{ij} - 2 (\ri[\gamma]\circ V)_{ij} .
}}
\end{defn}

By using the gauge condition \eqref{gauges}, the operator $\mD$ can be rewritten as
\begin{align*}
\mD V_{cd}= g^{ab}\bnab_a \bnab_b V_{cd}- H^a \bnab_a V_{cd}.
\end{align*}
The operator $\mL$ is self-adjoint with respect to the mixed $L^2$-scalar product (see e.g. \cite{AM03})
\eq{\label{eq:mL-self-adjoint}
(\mL V, P)_{L^2(g, \gamma)} = (V, \mL  P)_{L^2(g, \gamma)}.
}
A self-adjoint elliptic operator on a compact manifold has a discrete spectrum of eigenvalues. 
Using eigenvalue estimates from \cite{KK-15}, we are led to the following result.

\begin{prop}[Estimates on $\lambda_0$]\label{prop:MilneGeometricRigidity}
Let $(M,\gamma)$ be a negative Einstein three-manifold with Einstein constant $k=-2/9$. Then the smallest eigenvalue of the operator $\mL$ satisfies $\lambda_0\geq 1/9$ and the operator also has trivial kernel $\ker(\mL)=\{0\}$. 
\end{prop}

\begin{defn}[Sobolev norms]
Let $k \in \mathbb{Z}_{\geq 0}$. 
For $\kappa$ a Riemannian metric on $M$, $f$ a function and $V$ a $(1,1)$-tensor, define 
\eq{ \alg{\notag
|\nabla^k f|_\kappa^2 &\define \kappa^{a_1 b_1} \cdots \kappa^{a_k b_k} (\nabla_{a_1}\cdots \nabla_{a_k} f )\cdot (\nabla_{b_1}\cdots \nabla_{b_k} f ),
\\
|\nabla^k V|_\kappa^2 &\define \kappa_{ij} \kappa^{kl} \kappa^{a_1 b_1} \cdots \kappa^{a_k b_k} (\nabla_{a_1}\cdots \nabla_{a_k} V^i{}_k)\cdot (\nabla_{b_1}\cdots \nabla_{b_k} V^j{}_l).
}}
The obvious extension to $(p,q)$-tensors holds. 
We write
\eq{\alg{\notag
\|V \|_{H^k} &= \Big(\sum_{0\leq \ell \leq k}\int_M |\nabla^\ell V|_g^2 \, \mu_g\Big)^{1/2}.
}}
Note that under a global smallness assumption on $g-\gamma$ guaranteed by the bootstrap assumptions, we have the norm equivalence $\| \cdot\|_{L^2} \cong \|\cdot \|_{\Lgg}$. 
\end{defn}

\begin{rem}\label{rem:u-norm}
When we write $\|u \|_{H^{k}}$ we are denoting a sum \emph{only} over the spatial components of the velocity vector-field $u^\mu$.  
\end{rem}

We frequently, and without comment, use the following product estimate:
\begin{lem}[Sobolev product estimates]\label{lem:SobProdEst}
If $s>n/p = 3/2$ then 
\eq{\notag
\| uv \|_{H^s} \lesssim \| u \|_{H^s}\| v \|_{H^s}.
}
\end{lem}

We conclude this subsection with a result concerning elliptic regularity, see e.g. \cite[App. H]{Besse}.
\begin{lem}[Elliptic regularity using $\mL$]
\label{lem:elliptic-regularity-mL}
Let $V$ be a symmetric $(0,2)$-tensor on $M$.
There exist constants $C_1, C_2>0$ such that for all $s \in \mathbb{Z}_{\geq 0}$
\eq{\notag
C_1\| V \|_{H^{k+2s}} \leq \| \mL^s V\|_{H^k} \leq C_2\|V\|_{H^{k+2s}}\,,
}
where $\mL^s$ denotes $s$-copies of $\mL$. 
\end{lem}

\subsection{Energy for the Geometric Perturbations} \label{subsec:energy-geom-def}
As noted in \cite{AM03}, in the spatially harmonic gauge \eqref{gauges} we have
\begin{align*}
\Ric[g]_{ab}+\frac29 g_{ab} = \frac12 \mL(g-\gamma)_{ab} + J_{ab},
\end{align*}
where $J_{ab}$ are higher-order terms (writen as $S_{ab}$ in \cite[pg. 22]{AM03}) satisfying,  for $k\geq1$,
\begin{align*}
\| J\|_{H^{k-1}}\leq C \| g-\gamma \|_{H^k}.
\end{align*}
Following \cite{AM11, AF20}, we define an energy for the geometric perturbation of the first and second fundamental forms by using $\mL$. This energy will fulfill a strong decay estimate enabled by the inclusion of certain correction terms $\Gamma_{(m)}$.

\begin{defn}[Geometric energy $\Eg_k$]\label{defn:GeometricEnergy}
Let $\lambda_0$ be the lowest eigenvalue of the operator $\mL$, with lower bounds given in Proposition \ref{prop:MilneGeometricRigidity}. 
We define the correction parameter $\alpha=\alpha(\lambda_0,\delta_\alpha)$ by
\eq{\notag
\alpha\define
\begin{cases}
1& \lambda_0>1/9\\
1-\delta_\alpha& \lambda_0=1/9,
\end{cases}
}
where $\delta_\alpha=\sqrt{1-9(\lambda_0-\varepsilon')}$ with $1\gg\varepsilon'>0$ remains a variable to be determined in the course of the argument to follow. By fixing $\varepsilon'$ once and for all, $\delta_\alpha$ can be made suitably small when necessary.
The corresponding correction constant, relevant for defining the corrected energies, is defined by
\eq{\notag
c_E\define\begin{cases}
1& \lambda_0>1/9\\
9(\lambda_0-\varepsilon')& \lambda_0=1/9.
\end{cases}
}
We are now ready to define the energy for the geometric perturbation.
For $m,k \in \mathbb{Z}_{\geq 1}$ let
\begin{align*}
\mathcal{E}_{(m)}&\define\frac12\left( v,\mathcal{L}_{g,\gamma}^{m-1}( v) \right)_{\Lgg}+\frac92\left( h,\mathcal{L}_{g,\gamma}^{m}(h)\right)_{\Lgg},\qquad
\Gamma_{(m)}\define\left( v,\mathcal{L}_{g,\gamma}^{m-1}(h)\right)_{\Lgg}.
\end{align*}
The energy measuring the geometric perturbation is then defined by
\eq{\notag
\Eg_k\define \sum_{1\leq m\leq k} \big(\mathcal{E}_{(m)}+c_E\Gamma_{(m)}\big).
}
\end{defn}

We will see later that the corrected geometric energy $\Eg_k$ is in fact coercive over the standard Sobolev norms of the geometric variables $g, \Sigma$.

\begin{defn}[Operators $\hdel, \delu$]\label{defn:delu}
We define the operators on $M$
\eq{\notag
\hdel \define \p_T + \mathcal{L}_X, \qquad \delu \define u^0 \p_T - \tau u^a \bnab_a .
}
\end{defn}
The following identity, taken from \cite{CM01}, holds for some function $f$ on $M$:
\eq{\label{eq:hdelIntegral}
\p_T \int_M f \mu_g = 3 \int_M \hN f \mu_g + \int_M \hdel(f) \mu_g \,.
}

\begin{rem}[Regularity parameters $\ell, N$]
\label{rem:ell-N}
At top-order our bootstrap assumptions will involve Sobolev norms $H^N$ where $N$ is a large integer. It is convenient to require $N$ to be odd so that we can introduce $\ell\in\mathbb{Z}$ satisfying $\ell=\frac{N-1}{2}$.
\end{rem}

We end this section with the top-order geometric energy for the geometric variables $g, \Sigma$, which crucially involves the fluid operator $\delu$. This energy will also fulfill a strong decay estimate enabled by the inclusion of the correction terms. 

\begin{defn}[$E^g_{\delu, 2\ell}$ and $E_{tot}$]\label{defn:delu-geom-energy}
For $s \in \mathbb{Z}_{\geq 1}$, define 
\eq{\alg{\notag
\mathcal{E}^g_{\delu,2s}&\define \frac92 \left(\delu\mL^{s}(h),\delu\mL^{s}(h)\right)_{\Lgg}+ \frac12\left( \delu \mL^{s} v ,\delu\mL^{s-1} v \right)_{\Lgg} ,
\\
\Gamma^g_{\delu,2s}&\define \left(\delu\mL^{s-1}v,\delu\mL^{s}(h)\right)_{\Lgg}.
}}
The corrected \emph{$\delu-$boosted} geometric energy is then given by
\eq{\notag
\Edelu\define \sum_{s=1}^\ell \Big(\mathcal{E}^g_{\delu,2s}+ c_E  \Gamma^g_{\delu,2s}\Big).
}
Finally we define
\eq{\notag
\Etot \define \Edelu + E^g_{N-1}.
}
\end{defn}

\section{The Bootstrap Argument}
In this section we first state the local-existence theory for the Einstein-Dust system in CMCSH gauge. Then we introduce the bootstrap assumptions on our solution and  give some immediate consequences of these estimates. 
 
\label{sec:LocalandBootstrap}
\subsection{Local Existence}
\begin{thm}
Let $N\geq 6$. Consider CMC initial data $(g_0,k_0,N_0,X_0,\rho_0,u_0)\in H^{N}\times H^{N-1}\times H^N\times H^N\times H^{N-2}\times H^{N-1}$ at $T=T_0$ such that the constraints \eqref{EoM-constraints} hold. Then there exists a unique classical solution $(g,k,N,X,\rho,u)$ on $[T_0,T_+)$ for $T_+>T_0$ to the system \eqref{EoM}, which is consequently also a solution to the Einstein-Dust equations.
The components have the following regularity features
\eq{\alg{\notag
g,N,X \in C^0([T_0,T_+),H^N)\cap C^1([T_0,T_+),H^{N-1}),\\
k\in C^0([T_0,T_+),H^{N-1})\cap C^1([T_0,T_+),H^{N-2}),\\
u\in C^{0}([T_0,T_+],H^{N-1}),\\
\rho,\partial_ \bold{u}\rho,\partial_{\bold u}u \in C^{0}([T_0,T_+],H^{N-2}).
}}
Furthermore, the time of existence and the norms of the solution depend continuously on the initial data. 
For the maximal time of existence $T_\infty$ we have either $T_\infty=+\infty$ or
\eq{\alg{\notag
\lim_{T\nearrow T_\infty} \sup_{[T_0,T_\infty]}& \|g-\gamma\|_{H^N}+ \|\Si\|_{H^{N-1}}+\|N-3\|_{H^N}+\|X\|_{H^N}\\
&+\|\partial_T N\|_{H^{N-1}}+\|\partial_T X\|_{H^{N-1}}+\|\rho\|_{H^{N-2}}+|\tau|\|u\|_{H^{N-1}}> \delta(\gamma),
}}
where $\delta(\gamma)$ is a positive fixed constant depending only on the background metric.
\end{thm}
\begin{proof}
The proof follows analogous to \cite[Theorem 3.5]{HaSp15} where the control on the lapse and shift are replaced by the elliptic techniques applied in \cite{F16-2}. The adaption of the regularity scheme to avoid the loss of derivatives from \cite{HaSp15} to the present case is executed in detail in the global analysis discussed in the remainder of the paper. The smallness condition in the continuation criterion stems from a smallness requirement in applying the corresponding elliptic equation for the lapse.
\end{proof}

\begin{rem}
Since we apply the local-existence theorem only for data close to the background solution the smallness condition required to extend the solution for arbitrarily large times is automatically fulfilled when our smallness conditions hold, which we prove by a bootstrap argument. Since the smallness parameter of the continuation criterion depends only on the background geometry we can choose the smallness of the initial data, which we do accordingly without mentioning it explicitly again.
\end{rem}
\subsection{Bootstrap Assumptions}
Let $\mu, \lambda$ be fixed positive constants with $ \mu \ll 1$ and $\lambda <1$. 
We assume that, for all $T_0 \leq T\leq T'$, the following bootstrap assumptions hold \eqref{EoM}:
\eq{\alg{\label{Bootstraps}
\|{g-\gamma}\|_{H^N} +\|\Si\|_{H^{N-1}} &\leq C \varepsilon e^{-\lambda T},\\
\|N-3\|_{H^{N}}+\|X\|_{H^{N}}&\leq C\varepsilon  e^{-T},\\
\|\p_T N\|_{H^{N-1}}+\|\p_T X\|_{H^{N-1}}&\leq C\varepsilon  e^{-T},\\
\|\rho \|_{H^{N-2}}&\leq C \varepsilon, \\
 \|u \|_{H^{N-1}}&\leq C \varepsilon e^{\mu T},
}}
where $T'<T_\infty$ is fixed. 
We hereon assume that \eqref{Bootstraps} hold and \emph{do not repeat this fact}. Recall also Remark \ref{rem:u-norm} regarding the norm on $u$. 

\begin{defn}[$\Lambda(T)$]
It is convenient to introduce the following notation
\eq{\notag
\Lambda(T) \define \| \hN \|_{H^N} + \| X\|_{H^N} + \| \p_T N \|_{H^{N-1}} + \| \p_T X \|_{H^{N-1}} + |\tau| \| \rho\|_{H^{N-2}}
+ \tau^2 \| u \|_{H^{N-1}}^2.
}
Note that, under the bootstrap assumptions \eqref{Bootstraps}, $\Lambda(T) \lesssim \varepsilon e^{-T}$. 
\end{defn} 

We state some immediate consequences of the bootstrap assumptions regarding $u^0$ which we  use without further comment. 
Using Lemma \ref{lem:Renorm-u0}, we have
\eq{\alg{\label{eq:est-u0}
\|\hu\|_{H^{N-1}}
&\lesssim \varepsilon , \\ 
\| \nabla \hu \|_{H^{N-2}} &\lesssim \|\nabla N \|_{H^{N-2}} + \tau^2 \| u^a \nabla u^a \|_{H^{N-2}} + \text{h.o.t} \lesssim \Lambda(T). 
}}
Note the first estimate does not pick up any $\mu-$loss. By the Sobolev embedding $H^2 \hookrightarrow L^\infty$, $\|\hu\|_{L^\infty}= \|u^0-1/3\|_{L^\infty}\leq 1/10$ and thus 
\eq{\notag
\|u^0\|_{L^\infty}\lesssim 1, \quad \|u^0\|_{L^\infty}^{-1} \lesssim 1.
}

The following Lemma concerning the dust matter components is indicative of the good behaviour that, as we discussed in Section \ref{intro:CritExp}, we roughly expect as the speed of sound is reduced.

\begin{lem}[Estimates on matter components]
\label{lem:matter-first-estimates}
We have
\eq{\alg{\notag
|\tau|\big(\| \EnergyDensity\|_{H^{N-2}} + \| \eta\|_{H^{N-2}}+ \| S\|_{H^{N-2}} \big) 
&\lesssim \Lambda(T),\\
|\tau|^2 \| \jmath \|_{H^{N-2}}\lesssim \varepsilon e^{(-1+\mu)T} \Lambda(T), \quad
 |\tau|^3 \| \underline{T}\|_{H^{N-2}}
&\lesssim \Lambda(T)^2.
}}
\end{lem}
\begin{proof}
The estimates are immediate by distributing derivatives across the terms written in Definition \ref{Def-resc-matter} and using Lemma \ref{lem:SobProdEst}.
\end{proof}

\begin{lem}[Geometric coercivity estimate] \label{lem:Geom-coercive-est}
Let $s\in \mathbb{Z}$ such that $1\leq s \leq N-1$. 
There is a $\delta>0$ and a constant $C>0$ such that for $(g,\Sigma)\in B_{\delta}((\gamma,0))$ the following inequality holds
\eq{\notag
\| g-\gamma\|_{H^{s}}^2 + \| \Sigma \|_{H^{s-1}}^2 \leq C  E^g_{s}.
}
\end{lem}

\begin{proof}
The proof of this lemma follows verbatim from \cite[Lemma 19]{AF20}, which itself follows from \cite[Lemma 7.2]{AM11}, and the eigenvalue estimates referred to in Proposition \ref{prop:MilneGeometricRigidity}. 
\end{proof}

The following Lemma is actually only used once in our entire argument. Its significance lies in the fact that it allows us to convert the quadratic derivative term $(\nabla V)^2$ appearing in the $H^1$ Sobolev norm into just one second-order derivative as $V \cdot \mL V$, which will more naturally be controlled by $E^g_{\delu, N-1} $.

\begin{lem}\label{lem:H1-mL}
Let $V$ be a symmetric $(0,2)$-tensor on $M$. Then,
\eq{\notag
\| V\|_{H^1}  \lesssim \| V\|_{L^2} + (V, \mL V)_{\Lgg}^{1/2}.
}
\end{lem}
\begin{proof}
We integrate by parts, use the closeness between the $g$ and $\gamma$ metrics and the boundedness of the $\Riem[\gamma]$ components: 
\eq{\alg{\notag
\| V\|_{H^1}^2 
&= \| V\|_{L^2}^2 + \int_M g^{ab} g^{ij} g^{kl} \nabla_a V_{ik} \nabla_b V_{jl} \, \mu_g
\leq  \| V\|_{L^2}^2 + \Big| - \int_M g^{ab} g^{ij} g^{kl} V_{ik} \nabla_a \nabla_b V_{jl} \, \mu_g \Big|
\\
&\leq  \| V\|_{L^2}^2 + \Big| \int_M \langle V,\mL V \rangle_\gamma \, \mu_g + 2 \int_M \langle V, \Riem[\gamma]\circ V \rangle_\gamma \, \mu_g \Big|
\\
&\lesssim
\| V\|_{L^2}^2 + (V, \mL V)_{\Lgg}.
}}
\end{proof}

\section{Preliminary Estimates}\label{sec:FurtherPrelims}
This section is concerned with deriving several preliminary estimates that are required for our later energy inequalities. There are estimates on commutator terms (Sections \ref{sec:Comm1}, \ref{sec:Comm2}), estimates on matter terms (Section \ref{sec:matt1}) and also estimates on geometric variables (Sections \ref{sec:Geom1}, \ref{sec:Geom2}). 
We also present an integration by parts Lemma \ref{lem:IBP}, in particular \eqref{eq:symIBP}, which later plays an important role  in removing various critical terms that arise during the energy estimates. 

\subsection{First Commutator Estimates}\label{sec:Comm1}
In this subsection we let $V$ be an arbitrary symmetric $(0,2)$-tensor and $\phi$ a scalar, unless otherwise specified. 
We begin with the identity
\eq{\label{eq:pt-nabla}
[\p_T, \nabla_a]V_{i j} = (-\p_T \Gamma^c_{ai})V_{c j} + (-\p_T \Gamma^c_{aj})V_{i c}.
}
The terms $\p_T \Gamma[g]$ can be estimated (see \cite[(10.12)]{AF20}), for $0 \leq k \leq N-2$, by:
\eq{\label{eq:est-pt-Gamma}
\| \p_T \Gamma(g)\|_{H^k} \lesssim 
\| \Sigma\|_{H^{k+1}} + \| X\|_{H^{k+2}}
+ \| \hN \|_{H^{k+1}}.
}
We also have the following commutator identities: 
\begin{align}\label{eq:Commutator-pT-nablas}
[\p_T,\nabla_{a_1}\cdots \nabla_{a_k}]\phi 
&= [\p_T,\nabla_{a_1}]\nabla_{a_2}\cdots \nabla_{a_k} \phi 
+ \nabla_{a_1} [\p_T,\nabla_{a_2}]\nabla_{a_3}\cdots \nabla_{a_k} \phi 
\\ \notag 
&\quad
+\ldots +  \nabla_{a_1}\cdots \nabla_{a_{k-2}} [\p_T,\nabla_{a_{k-1}}] \nabla_{a_k} \phi 
+ \nabla_{a_1}\cdots \nabla_{a_{k-1}} [\p_T,\nabla_{a_{k}}] \phi ,
\\ \label{eq:Commutator-nabla-nablas}
[\nabla_i,\nabla_{a_1}\cdots \nabla_{a_k}]\phi 
&= [\nabla_i,\nabla_{a_1}]\nabla_{a_2}\cdots \nabla_{a_k} \phi 
+ \nabla_{a_1} [\nabla_i,\nabla_{a_2}]\nabla_{a_3}\cdots \nabla_{a_k} \phi 
\\\notag 
&\quad
+\ldots +  \nabla_{a_1}\cdots \nabla_{a_{k-2}} [\nabla_i,\nabla_{a_{k-1}}] \nabla_{a_k} \phi 
+ \nabla_{a_1}\cdots \nabla_{a_{k-1}} [\nabla_i,\nabla_{a_{k}}] \phi .
\end{align}
Note the last terms in each of \eqref{eq:Commutator-pT-nablas} and \eqref{eq:Commutator-nabla-nablas} will in fact vanish since the metric $g$ is torsion free.

In the next part of this subsection we state an important estimate, given in \eqref{eq:Estimate-Upsilon}, that allows us to turn background $\bnab$ derivatives into dynamical $\nabla$ ones. 

\begin{defn}[Difference tensor $\Upsilon$]\label{defn:Upsilon}
Recalling that $\nabla$ and $\bnab$ are the Levi-Civita symbols of $g$ and $\gamma$ respectively, we define
$\Upsilon$  a (1,2)-tensor by
\eq{\notag
\Upsilon^a_{bc} \define \Gamma^a_{bc}[g]-\Gamma^a_{bc}[\gamma].
}
\end{defn}
Let $V$ be a vector and $P$ a one-form. Then we have
\eq{\notag
\nabla_a V^{i}= \bnab_a V^i + \Upsilon^i_{ja}V^{j},
\qquad
\nabla_a P_i= \bnab_a P_i - \Upsilon^j_{ia}P_j.
}
We will often schematically write tensorial contractions using $\ast$. For example, 
\eq{\notag
\nabla_a V^{i}= \bnab_a V^i + \Upsilon^i_{ja}V^{j}, \quad \text{ becomes } \quad \nabla V = \bnab V + \Upsilon \ast V.
}
In local coordinates the components of the $\Upsilon$ tensor are given by
\eq{\alg{\notag
\Upsilon^a_{bc} &=-\frac12 \gamma^{ai}\left( \nabla_b \gamma_{ci} + \nabla_c \gamma_{ai} - \nabla_i \gamma_{bc} \right) 
=\frac12 \gamma^{ai}\left( \nabla_b h_{ci} + \nabla_c h_{ai} - \nabla_i h_{bc} \right) .
}}

\begin{lem} \label{lem:Commutator-nabla-bnab}
For  $0\leq k \leq N-2$, we have
\eq{\alg{\notag
\| [\bnab, \nabla] V\|_{H^k}
&\lesssim \| V\|_{H^k}  + \varepsilon e^{-\lambda T}\| V\|_{H^{k+1}}.
}}
\end{lem}

\begin{proof}
First we see that from Definition \ref{defn:Upsilon}, for all $0 \leq k \leq N-1$,
\eq{\label{eq:Estimate-Upsilon}
\| \Upsilon \|_{H^k} \lesssim \| g-\gamma\|_{H^{k+1}}.
}
Thus we compute,
\eq{\alg{\notag
[\bnab_c, \nabla_a] V_{i j}
&= [\bnab_c, \bnab_a] V_{i j}+ \Upsilon^b_{ca}\bnab_b V_{i j} - \big((\bnab_c \Upsilon^b_{ai})V_{b j} + (\bnab_c \Upsilon^b_{aj}) V_{i b}\big).
}}
Using the boundedness of the $\Riem[\gamma]$ components, the required estimate then follows.
\end{proof}

We end this subsection with some very useful estimates that come from integration by parts.

\begin{lem}[Integration by parts]\label{lem:IBP}
Let $T$ be an arbitrary vectorfield and $V, P$ arbitrary symmetric $(0,2)$-tensors, on $M$. Then 
\begin{subequations}\label{eq:AllIBPS}
\eq{\label{eq:IBP}
(\mL V, P)_{L^2(g, \gamma)}  = \int_M \langle g^{ab} \bnab_a V, \bnab_b P \rangle_\gamma \mu_g 
-2 \int_M \langle\Riem[\gamma]_{\bullet a\bullet b} V^{ab}, P\rangle_\gamma \mu_g .
}
and 
\begin{align}
\label{eq:thirdIBP-2} 
\big| (T^a \bnab_a V, P)_{\Lgg} \big|
&\lesssim
\| T \|_{H^3} \| V \|_{L^2} \| P\|_{H^1},
\\ \label{eq:symIBP} 
\big| (T^a \bnab_a V, V)_{\Lgg} \big|&\lesssim
\| T \|_{H^3} \| V \|_{L^2}^2.
\end{align}
\end{subequations}
\end{lem}
\begin{proof}
The proof of \eqref{eq:IBP} follows by using the gauge condition $H^a = 0$. To show \eqref{eq:thirdIBP-2}, recall the Jacobi identity implies 
\eq{\notag 
\bnab_a \sqrt{\det g} = \frac12 \sqrt{\det g} \cdot g^{ij}(\bnab_a g_{ij}).
}
 Using this we find, 
\eq{\alg{\notag
(T^a \bnab_a V, P)_{\Lgg} 
&= \int_M \gamma^{ij}\gamma^{kl} T^a (\bnab_a V_{ik}) P_{jl} \, \sqrt{\det g}
= -\int_M \gamma^{ij}\gamma^{kl} \bnab_a \left(T^a P_{jl} \frac{\sqrt{\det g}}{\sqrt{\det \gamma}}\right) V_{ik} \sqrt{\det \gamma} 
\\&=
- (V, T^a \bnab_a  P)_{\Lgg} 
- ((\bnab_a T^a) V, P)_{\Lgg} 
- \tfrac12 (V, (T^a g^{bc}\bnab_a g_{bc}) P)_{\Lgg} .
}}
Thus,
\eq{\alg{\notag
\big| (T^a \bnab_a V, P)_{\Lgg} \big|
&= \big|
- (V, T^a \bnab_a  P)_{\Lgg} 
- ((\bnab_a T^a) V, P)_{\Lgg} 
- \tfrac12 (V, (T^a g^{bc}\bnab_a g_{bc}) P)_{\Lgg} \big|
\\&\lesssim
\big( \| T^a \|_{H^3} + \| T^a \|_{H^2} \| g-\gamma\|_{H^3} \big) \| V \|_{L^2} \| P\|_{H^1}.
}}
Crucially, in the symmetric case, we can bring one term over to the left hand side to show
\eq{\alg{\notag
\big| (T^a \bnab_a V, V)_{\Lgg} \big|
&= \big|
- \tfrac12 ((\bnab_a T^a) V, V)_{\Lgg} 
- \tfrac14 (V, (T^a g^{bc}\bnab_a g_{bc}) V)_{\Lgg} \big|
\\&\lesssim
\big( \| T^a \|_{H^3} + \| T^a \|_{H^2} \| g-\gamma\|_{H^3} \big) \| V \|_{L^2}^2.
}}
\end{proof}

\subsection{First Estimates on Geometric Components} \label{sec:Geom1}
We now establish some of our first estimates on the ADM variables $N, X$ and the geometric variables $g, \Sigma$. 

\begin{lem}[Dust derivatives of lapse and shift]
\label{lem:delu-N-X}
For $2\leq k \leq N-1 $,
\eq{\alg{\notag
\| \delu N \|_{H^k} + \| \delu X \|_{H^k} \lesssim 
 \Lambda(T).
}}
\end{lem}

\begin{proof}
Since the lapse is a scalar, $\delu N = u^0 \p_T N - \tau u^c \nabla_c N$. For $\delu X^a$ we use \eqref{eq:Estimate-Upsilon}. We find
\eq{\alg{\notag
\| \delu N \|_{H^k} + \| \delu X \|_{H^k} \lesssim 
\| \p_T N \|_{H^k} + \| \p_T X \|_{H^k} + \varepsilon e^{(-1+\mu) T} \Lambda(T).
}} 
\end{proof}

\begin{lem}[Estimates on geometric source terms]
We have,
\label{lem:Est-Fh-Fv}
\eq{\alg{\notag
\| F_h \|_{H^k} 
 &\lesssim 
\| \hN \|_{H^k} 
+ \| X\|_{H^{k+1}}^2 
 + \| g-\gamma \|_{H^{k+1}}^2, & \quad 2 \leq k \leq N-1,
 \\
\|F_v\|_{H^k}&\lesssim 
\| \hN \|_{H^{k+2}} 
+ |\tau|  \|\rho\|_{H^k} 
+ \| X\|_{H^{k+1}}^2 
 + \| g-\gamma \|_{H^{k+1}}^2 + \| \Sigma\|_{H^{k}}^2,  & \quad 2 \leq  k \leq N-2,
}}
and thus
\eq{\notag
\| F_h \|_{H^{N-1}} + \| F_v \|_{H^{N-2}}  \lesssim \Lambda(T) + \| g-\gamma\|_{H^N}^2 + \| \Sigma\|_{H^{N-1}}^2. 
}
\end{lem}

\begin{proof}
Using the product estimate of Lemma \ref{lem:SobProdEst} and \eqref{eq:Estimate-Upsilon} we obtain, for $2 \leq k \leq N-1$,
\eq{\alg{\notag
\| F_h\|_{H^k} &\lesssim \| \hN \|_{H^k} 
+ \| h \|_{H^k} \| \nabla X + \Upsilon \ast X\|_{H^{k}} 
\\
&\lesssim 
\| \hN \|_{H^k} 
+ \| g-\gamma\|_{H^k} \| X\|_{H^{k+1}} 
 + \| g-\gamma \|_{H^{k+1}}^2 \| X\|_{H^{k}}.
}}
Using in addition Lemma \ref{lem:matter-first-estimates} we see, for $2 \leq k \leq N-2$,
\eq{\alg{\notag
\| F_v \|_{H^k}
&\lesssim \| \hN \|_{H^{k+2}} 
+  \| \Sigma\|_{H^k}^2 
+ |\tau|  \|S_{ij} \|_{H^k} 
+ \| \Sigma \|_{H^k} \| \bnab X\|_{H^{k}} .
}}
\end{proof}

We can now combine the geometric source term estimates from the previous lemma with the equations of motion given in \eqref{eq:EoM-hv}.
\begin{lem}[Time derivatives of $g$ and $\Sigma$]
\label{lem:Estimates-pT-h-v}
The following estimates hold
\eq{\alg{\notag
\| \p_T g\|_{H^k} 
 &\lesssim 
\| \Sigma\|_{H^{k}}
+ \| g-\gamma \|_{H^{k+1}}
+\| \hN \|_{H^k} 
+ \| X\|_{H^{k+1}}
, &2\leq k \leq N-1,
\\
\| \p_T \Sigma\|_{H^k}& \lesssim 
\| \Sigma\|_{H^{k+1}}
+  \| g-\gamma \|_{H^{k+2}} 
+\| \hN \|_{H^{k+2}} 
+ \| X\|_{H^{k+1}}
+ |\tau| \|\rho\|_{H^k},
&2\leq k \leq N-2.
}}
\end{lem}

\begin{proof}
Using Lemma \ref{lem:Est-Fh-Fv}, \eqref{eq:EoM-hv} and \eqref{eq:Estimate-Upsilon} we find
\eq{\alg{\notag
\| \p_T h \|_{H^k} 
&\lesssim 
\| \Sigma\|_{H^{k}} + \| X(\nabla h + \Upsilon\ast h)\|_{H^{k}}  + \| F_h \|_{H^k}
\lesssim
\| \Sigma\|_{H^{k}}
+\| \hN \|_{H^k} 
+ \| X\|_{H^{k+1}}^2 
+ \| g-\gamma \|_{H^{k+1}}^2\,,
}}  
and, using additionally Lemma \ref{lem:elliptic-regularity-mL},
\eq{\alg{\notag
\| \p_T v\|_{H^k}
&\lesssim 
\| \Sigma\|_{H^{k}} +\| \mL h\|_{H^{k}}+ \| X(\nabla v + \Upsilon\ast v)\|_{H^{k}}  + \| F_v \|_{H^k}
\\&\lesssim 
\| \Sigma\|_{H^{k}}
+ \| g-\gamma \|_{H^{k+2}}
+\| \hN \|_{H^{k+2}} + \| X\|_{H^{k+1}}^2
+ \| \Sigma\|_{H^{k+1}}^2 + |\tau| \|\rho\|_{H^k}.
}}
\end{proof}

\begin{corol}[Dust derivatives of $g$ and $\Sigma$]
\label{corol:Estimates-delu-h-v}
The following estimates hold
\eq{\alg{\notag
\| \delu g\|_{H^k} 
 &\lesssim 
\| \Sigma\|_{H^{k}}
+ \| g-\gamma \|_{H^{k+1}}
+\Lambda(T)
, &2\leq k \leq N-1,
\\
\| \delu \Sigma\|_{H^k}& \lesssim 
\| \Sigma\|_{H^{k+1}}
+  \| g-\gamma \|_{H^{k+2}} 
+\Lambda(T),
 &2\leq k \leq N-2.
}}
\end{corol}

\begin{proof}
Writing 
\eq{\notag \delu V_{ab} = u^0 \p_T V_{ab}- \tau u^c \nabla_c V_{ab}- \tau u^c\Upsilon^d_{ca} V_{db} - \tau u^c\Upsilon^d_{cb} V_{ad}}
and using Lemma \ref{lem:Estimates-pT-h-v} gives the estimates.
\end{proof}

\begin{rem}
It is unsurprising that for the geometric variables $g, \Sigma$, the fluid derivative estimates in Corollary \ref{corol:Estimates-delu-h-v} do not gain us any improved information compared to the time derivative estimates in Lemma \ref{lem:Estimates-pT-h-v}. This will of course change when we consider instead estimates on certain fluid matter variables which are naturally more compatible with the fluid derivative operator $\delu$. 
\end{rem}

\subsection{First Estimates on Fluid Components}\label{sec:matt1}
In this subsection we establish further estimates on the various matter variables. Note that we need to estimate both matter components coming from contractions with the stress energy tensor (see Definition \ref{Def-resc-matter}) and the fluid source terms terms appearing in the equations of motion (see Definition \ref{Defs-fluid-source}).

\begin{lem}[Fluid source term estimates]
\label{lem:Fluid-Source-Estimates}
We have, for some $\nu >0$,
\eq{\alg{\notag
\|F_\rho\|_{H^{N-2}} 
&\lesssim |\tau| \| u\|_{H^{N-1}} + \| \Sigma \|_{H^{N-2}}+ \Lambda(T),
\\
\|F_{u^0} \|_{H^{N-1}} 
&\lesssim \| \Sigma\|_{H^{N-1}}^2 + \Lambda(T),
\\
\|F_{u^j} \|_{H^{N-1}} 
&\lesssim |\tau|^{-1} \Lambda (T) + \varepsilon^2 e^{-\nu T}.
}}
\end{lem}
\begin{proof}
For $2\leq k \leq N-2$ we distribute derivatives across the terms given in Definitions \ref{Def-Gammas-rescaled} and \ref{Defs-fluid-source}. This yields
\eq{\alg{\notag
\| F_\rho \|_{H^k}
&\lesssim 
|\tau| \| \nabla u^j \|_{H^{k}} 
+ \|u^0\|_{L^\infty} \| \Gamma^i_i\|_{H^k}
+ |\tau| \| \Gamma_j\|_{H^k}\| u \|_{H^k}
+  |\tau|\| \Gamma_{ik} \|_{H^k}\| X\|_{H^k}\| u \|_{H^k}
\\& \quad  
+ \tau^2 \| u\|_{H^k}^2 \| \Gamma_{kj}\|_{H^{k}} 
\\
&\lesssim |\tau| \| u\|_{H^{k+1}} + \| \Sigma \|_{H^k}+ \|X\|_{H^{k+1}}.
}}
Similarly, for $2\leq k \leq N-1$,
\eq{\alg{\notag
\| F_{u^0} \|_{H^k}
&\lesssim 
\| \Gamma_R\|_{H^k}+ |\tau|\| u\|_{H^k} \| \Gamma_j \|_{H^k} + \tau^2 \| u\|_{H^k}^2 \| \Gamma_{jk} \|_{H^k}
\\
&\lesssim \| \p_T N \|_{H^k} 
+ \| \Sigma \|_{H^k}^2+ \| \hN\|_{H^{k+1}}^2+ \|X\|_{H^{k}}^2 
+ \tau^2 \| u\|_{H^k}^2.
}}
Finally, for $2\leq k \leq N-1$,
\eq{\alg{\notag
\| F_{u^j} \|_{H^k}
&\lesssim 
|\tau|^{-1} \| \overset{\circ}{\Gamma}{}^j\|_{H^k}+ \| u \|_{H^k}^2 \big( \| \Gamma[g]-\Gamma[\gamma]\|_{H^k} + |\tau|\| X\|_{H^k} \| \Gamma_{ki}\|_{H^k}\big) + \| u\|_{H^k} \| \Gamma^i_j \|_{H^k}
\\
&\lesssim |\tau|^{-1} \big(\| \p_T X \|_{H^k} + \| X\|_{H^{k+1}} + \| \hN \|_{H^{k+1}}^2 + \| \Sigma\|_{H^k}^2 + \| X\|_{H^k} \| \p_T N \|_{H^k} \big) 
\\
& \quad 
+ \| u \|_{H^k}^2 \big( \|g-\gamma\|_{H^{k+1}} + |\tau| \| X\|_{H^k}\big) + \| u\|_{H^k} \big( \| \Sigma\|_{H^k} + \| \hN \|_{H^{k}} + \| X\|_{H^{k+1}} \big)
\\
&\lesssim |\tau|^{-1} \Lambda (T) + \varepsilon^2 e^{(1-2\lambda) T} + \varepsilon^2 e^{(-\lambda+2\mu)T}.
}}
It is convenient also to note that at lower order we can apply Lemma \ref{lem:Geom-coercive-est} to find
\eq{\alg{\label{eq:F-uj-est-2}
|\tau| \| F_{u^j} \|_{H^{N-2}}
&\lesssim  E^g_{N-2} + \Lambda(T)
+ |\tau|\| u\|_{H^{N-2}}^2  (E^g_{N-1})^{1/2}.
}}
\end{proof}

In the next lemma we provide estimates for fluid derivatives of certain matter components appearing in Definition \ref{Def-resc-matter}. The weights in $\tau$ are included for convenience since these expressions appear later on in the energy estimates. 

\begin{lem}[Dust derivatives of matter components]  \label{lem:Est-delu-eta-S}
We have,
\eq{\alg{\notag
|\tau| \| \delu \eta \|_{H^{N-2}}   + 
|\tau|\| \delu S \|_{H^{N-2}}  
&\lesssim 
\varepsilon e^{\max\{-1+\mu,-\lambda\} T}\Lambda(T) + \Lambda(T)^2,\\
|\tau|^2 \| \delu \jmath \|_{H^{N-2}}  
&\lesssim \Lambda(T)^2 + \varepsilon^2 \Lambda(T) e^{(-\lambda+\mu)T}.
}}
\end{lem}
\begin{proof}
Note that \eqref{eq:EoM-fluid} can be rewritten as:
\eq{\alg{\notag
\delu u^j
&= \tau u^a u^c\Upsilon^j_{ac}  + 
\tau^{-1}(u^0)^2 N\nabla^j N+F_{u^j}, \qquad 
\delu u^0 =F_{u^0} , \qquad
\delu \rho
= \rho F_\rho .
}}
Using Definition \ref{Def-resc-matter} we compute:
\eq{\alg{\notag
\delu \eta 
&= 
\rho F_\rho \left( (u^0)^2 N^2 + g_{ab}(X^a u^0 + \tau u^a)(X^b u^0 + \tau u^b)\right)
+ 2 \rho (u^0)^2 N \delu N
\\&\quad 
+ \rho (\delu g_{ab}) (X^a u^0 + \tau u^a)(X^b u^0 + \tau u^b) 
+ 2 \rho g_{ab}  u^0(X^b u^0 + \tau u^b) (\delu X^a)
\\&\quad 
+F_{u^0} \left(  2\rho u^0 N^2  + 2\rho g_{ab} X^a(X^b u^0 + \tau u^b) \right)
+ 2\rho g_{ab} u^0 u^a (X^b u^0 + \tau u^b) (\p_T \tau)
\\&\quad
+ 2\rho g_{ab} \tau (X^b u^0 + \tau u^b)\left(\tau u^a \Upsilon^j_{ac} u^c + 
\tau^{-1}(u^0)^2 N\nabla^j N+F_{u^j}\right).
}}
So, using  Lemma \ref{lem:delu-N-X}, Corollary \ref{corol:Estimates-delu-h-v} and the matter estimates from Lemma \ref{lem:Fluid-Source-Estimates}, we obtain
\eq{\alg{\notag
|\tau|\| \delu \eta \|_{H^{N-2}}
&\lesssim |\tau|\|\rho \|_{H^{N-2}}\| F_\rho \|_{H^{N-2}}
+ |\tau|\|\rho \|_{H^{N-2}}\| \delu N\|_{H^{N-2}}
+|\tau|\|\rho \|_{H^{N-2}} \| F_{u^0}\|_{H^{N-2}}
+ \text{h.o.t. }
\\
&\lesssim
\varepsilon e^{\max\{-1+\mu,-\lambda\}T}\Lambda(T) + \Lambda(T)^2.
}}

Again using \eqref{Def-resc-matter} we compute
\eq{\alg{\notag
\delu \jmath^a 
&= (\rho F_\rho) N u^0 u^a + (\delu N) \rho u^0 u^a + (F_{u^0}) \rho N u^a 
+ (\tau \Upsilon^a_{bc}u^b  u^c + 
\tau^{-1}(u^0)^2 N\nabla^a N+F_{u^a}) \rho N u^0.
}}
So, by the geometry estimates in Lemma \ref{lem:delu-N-X} and \ref{lem:Estimates-pT-h-v}, together with the fluid source term estimates in Lemma \ref{lem:Fluid-Source-Estimates}, we obtain
\eq{\alg{\notag
\tau^2 \| \delu \jmath^a\|_{H^{N-2}}
&\lesssim |\tau| \|\rho \|_{H^{N-2}} \|\hN \|_{H^{N-1}}
+ |\tau|^2 \|\rho \|_{H^{N-2}}\| F_{u^a} \|_{H^{N-2}} + \text{h.o.t. }
\\&\lesssim
\varepsilon^2 \Lambda(T) e^{(-\lambda+\mu)T} + \Lambda(T)^2 .
}}

Finally, from \eqref{Def-resc-matter}  we find
\eq{\alg{\notag
\delu S_{ab} &= 
(\delu \rho)\big( ( u^0 X_a+ \tau u_a )(u^0 X_b+ \tau u_b) +\tfrac12 g_{ab}\big) + \tfrac{1}{2}\rho (\delu g_{ab}) +\text{h.o.t.} 
}}
So, 
\eq{\alg{\notag
|\tau|\| \delu S\|_{H^{N-2}}
&\lesssim |\tau| \|\rho \|_{H^{N-2}}\| F_\rho \|_{H^{N-2}}
\lesssim
\varepsilon e^{\max\{-1+\mu,-\lambda\} T}\Lambda(T) + \Lambda(T)^2.
}}
\end{proof}

\begin{rem}
The significance of using dust derivatives is made clear by look at the higher regularity appearing in Lemma \ref{lem:Est-delu-eta-S} compared to the following Lemma \ref{lem:Est-pT-eta-S} which only concerns time derivatives. In Lemma \ref{lem:Est-delu-eta-S}, we directly computed the $\delu$ derivatives using the equations of motion, instead of doing a rough estimate by expanding $\delu \sim \p_T + \tau u^c \ast \nabla$. 
\end{rem}

\begin{lem}[Time derivatives of matter components] 
\label{lem:Est-pT-eta-S}
We have,
\eq{\alg{\notag
|\tau|\| \p_T\eta \|_{H^{N-3}}&\lesssim  
\varepsilon e^{(-1+\mu)T} \Lambda(T) + \Lambda(T)^2,
\\
\Lambda(T) \| \p_T \jmath \|_{H^{N-3}} &\lesssim  
\varepsilon^2  \Lambda(T)
+\varepsilon^4 e^{-(1+\nu)T},
\\
|\tau| \| \p_TS \|_{H^{N-3}} &\lesssim  
\Lambda(T).
}}
\end{lem}
\begin{proof}
We calculate $\p_T \eta$ explicitly from \eqref{Def-resc-matter} as
\eq{\alg{\notag
\p_T \eta &= \p_T \EnergyDensity +  \big(X_a X^a (u^0)^2 + 2\tau X_b u^b u^0 + \tau^2 g_{ab} u^a u^b \big)\p_T \rho
+  \big( 2\rho X^a (u^0)^2 + 2\tau \rho u^a u^0\big)\p_T X_a 
\\
&+ \big( 2u^0 \rho X_a X^a + 2 \tau \rho X_b u^b \big)\p_T u^0 +  \big( 2\tau \rho X_b u^0 + 2 \tau^2 \rho u_b\big)\p_T u^b + \big( 2\rho X_b u^b u^0+2 \tau \rho u_a u^a\big)\p_T \tau.
}}
Using $\p_T\tau = -\tau$, we obtain
\eq{\alg{\label{eq:pT-eta-101}
|\tau| \|\p_T \eta \|_{H^{N-3}}
&\lesssim
|\tau| \|\p_T \EnergyDensity \|_{H^{N-3}} + \Lambda(T) |\tau|  \| \p_T\rho\|_{H^{N-3}} + \varepsilon e^{(-1+\mu)T} \Lambda(T)\|\p_T X \|_{H^{N-3}}
\\&\quad
+ \Lambda(T)^2 \|\p_T u^0 \|_{H^{N-3}}
+ \varepsilon e^{(-1+\mu)T} \Lambda(T)\|\tau \p_T u^b \|_{H^{N-3}}
+ \Lambda(T)^2.
}}

To estimate $\p_T \EnergyDensity$ we use the rescaled continuity equation \cite[Eq 10.16]{AF20}:
\eq{\notag
\p_T \EnergyDensity
=
	(3-N)\EnergyDensity - X^a \nabla_a \EnergyDensity 
	+ \tau N^{-1} \nabla_a (N^2 j^a) 
	- \tau^2 {\tfrac{N}{3}} g_{ab} \underline{T}^{ab}
	- \tau^2 N \Sigma_{ab} \underline{T}^{ab}.
}
So by Lemma \ref{lem:matter-first-estimates}, we obtain
\eq{\alg{\notag
|\tau|\| \p_T \EnergyDensity\|_{H^{N-3}} 
&\lesssim |\tau|\| \EnergyDensity\|_{H^{N-2}}\big( \| \hN \|_{H^{N-3}} + \| X\|_{H^{N-3}}\big)
 + |\tau|^2 \| \jmath\|_{H^{N-2}}+|\tau|^3 \|\underline{T}\|_{H^{N-3}}
\\
&\lesssim \varepsilon e^{(-1+\mu)T}\Lambda(T) + \Lambda(T)^2. 
}}

To estimate $\p_T \rho, \p_T u^0, \p_T u^a$ we use the equations of motion  \eqref{eq:EoM-fluid} together with Lemma \ref{lem:Fluid-Source-Estimates}. We find, 
\begin{subequations} \label{eq:est-pT-fluid-compts}
\eq{\alg{
|\tau|\| \p_T \rho\|_{H^{N-3}}
&\lesssim |\tau|\| \rho\|_{H^{N-3}} \| F_\rho\|_{H^{N-3}}+ |\tau|^2 \| u\|_{H^{N-3}} \| \rho\|_{H^{N-2}}
\\&
\lesssim 
\varepsilon e^{\max\{-1+\mu,-\lambda\}T} \Lambda(T) + \Lambda(T)^2,
}}
and
\eq{\alg{
\| \p_T u^0\|_{H^{N-2}}
&\lesssim \| F_{u^0} \|_{H^{N-2}}+ |\tau| \| u\|_{H^{N-2}} \| \nabla u^0\|_{H^{N-2}}
\lesssim \Lambda(T) + \| \Sigma\|_{H^{N-2}}^2,
\\
|\tau|\| \p_T u^a\|_{H^{N-2}}
&\lesssim |\tau|\| F_{u^a} \|_{H^{N-2}}+  \| \hN \|_{H^{N-1}} + \tau^2 \| u\|_{H^{N-1}}^2
\lesssim \Lambda(T)  + \varepsilon^2 e^{-(1+\nu)T}.
}}
\end{subequations}
Putting all these estimates into \eqref{eq:pT-eta-101}  gives, for $2\leq k \leq N-3$:
\eq{\alg{\notag
|\tau|\|\p_T \eta \|_{H^{k}}
&\lesssim 
\varepsilon e^{(-1+\mu)T} \Lambda(T) + \Lambda(T)^2.
}}

Next, and again using \eqref{Def-resc-matter}, we compute
\eq{\alg{\notag
\p_T \jmath^a &= (\p_T \rho) N u^0 u^a + (\p_T N) \rho u^0 u^a + (\p_T u^0) \rho N u^a + (\p_T u^a) \rho N u^0.
}}
So, 
\eq{\alg{\notag
\Lambda(T)\| \p_T \jmath^a\|_{H^{N-3}} 
&\lesssim \varepsilon |\tau|\| u \|_{H^{N-3}}\| \p_T \rho\|_{H^{N-3}} 
\\&\quad
+ \varepsilon |\tau|\|\rho\|_{H^{N-3}} 
\Big(\| \p_T N\|_{H^{N-3}} \| u \|_{H^{N-3}} 
+ \| \p_T u^0\|_{H^{N-3}} \| u \|_{H^{N-3}} + \|\p_T u^a \|_{H^{N-3}} \Big)
\\&\lesssim 
\varepsilon^2  \Lambda(T)
+\varepsilon^4 e^{-(1+\nu)T}.
}}

Finally, from \eqref{Def-resc-matter} we calculate (written schematically)
\eq{\alg{\notag
\p_T S &= 
 (\p_T\rho)\big( ( u^0 X^a+ \tau u^a )^2 + \tfrac12 g\big)
+ \tfrac{1}{2}\rho (\p_Tg)
+ \rho \p_T g \cdot ( u^0 X^c+ \tau u^c )^2
 \\&\quad
+ \rho\p_T( u^0 X^c+ \tau u^c)\cdot (u^0 X^d+ \tau u^d).
}}
So using Lemma \ref{lem:Estimates-pT-h-v} and \eqref{eq:est-pT-fluid-compts} we have,
\eq{\alg{\notag
|\tau| \| \p_T S\|_{H^{N-3}}
&\lesssim 
|\tau|(\| \rho\|_{H^{N-3}} \| F_\rho\|_{H^{N-3}}+ |\tau|\| u\|_{H^{N-3}} \| \rho\|_{H^{N-2}})
+ |\tau|\|\rho \|_{H^{N-3}} \times \big( \text{h.o.t.}\big)
\\&
\lesssim 
|\tau|\|\rho \|_{H^{N-2}} .
}}
\end{proof}

\subsection{Second Commutator  Estimates}\label{sec:Comm2}
We are now in a position to compute various commutator estimates which are required in the later energy estimates. We let $V$ be a symmetric $(0,2)$-tensor on $M$ unless otherwise specified. The first lemma looks at the commutator between the dust derivative $\delu$ with other first-order differential operators. 

\begin{lem}
\label{lem:Commutators-with-Delu}
For $0\leq k \leq N-2$, we have
\eq{\alg{\notag
\| [\delu, \p_T] V\|_{H^k}&\lesssim
\varepsilon e^{(-1+\mu) T} \| V\|_{H^{k+1}}
 + \varepsilon e^{(-1+\mu)T}\|\p_T V\|_{H^k},
\\
\| [\delu, \bnab] V \|_{H^k} &\lesssim \varepsilon e^{(-1+\mu) T}  \| V\|_{H^{k+1}}  +\Lambda(T)\|\p_T V \|_{H^k} ,
\\
\|[\delu, \nabla] V \|_{H^{k}} &\lesssim \varepsilon e^{\max\{-1+\mu, -\lambda\} T} \| V\|_{H^{k+1}} 
+\Lambda(T)\|\p_T V\|_{H^{k}}.
}}
\end{lem}
\begin{proof}
A computation gives
\eq{\alg{\label{eq:Commutators-with-Delu}
[\delu, \p_T] V_{i j} &= -\tau u^a \bnab_a V_{i j} + \tau \p_T u^a \cdot \bnab_a V_{i j} - \p_T \hu \cdot \p_T V_{i j} ,\\
[\delu, \bnab_b] V_{i j} &= -\tau u^a [\bnab_a, \bnab_b]V_{i j} -\bnab_b \hu \cdot \p_T V_{i j} + \tau\bnab_b u^a \cdot \bnab_a V_{i j},
\\
[\delu, \nabla_b] V_{i j} &= u^0[\p_T, \nabla_b]V_{i j}-\tau u^c[\bnab_c,\nabla_b]V_{i j}-\nabla_b\hu\cdot\p_T V_{i j} + \tau \nabla_b u^c\cdot\bnab_c V_{i j}.
}}
Thus by \eqref{eq:est-u0} and  \eqref{eq:est-pT-fluid-compts}, if $2\leq k \leq N-2$,
\eq{\alg{\notag
\| [\delu, \p_T] V\|_{H^k}
&\lesssim
|\tau| \big(\| u \|_{H^k} +\| \p_T u^a \|_{H^k} \big)\|\nabla V + \Upsilon \ast V\|_{H^{k}}
 + \| \p_T u^0\|_{H^k} \|\p_T V\|_{H^k}
\\&\lesssim
\varepsilon e^{(-1+\mu) T}\big( \|V\|_{H^{k+1}} + \|g-\gamma\|_{H^{k+1}}\| V\|_{H^{k}}\big)
 + \varepsilon e^{(-1+\mu)T}\|\p_T V\|_{H^k}.
}}
The estimates for $k=0,1$ follow in the same way. 

The other two estimates follow in a similar way. Note that for $[\delu, \nabla$] we use Lemma \ref{lem:Commutator-nabla-bnab}, and equations \eqref{eq:pt-nabla}, \eqref{eq:est-pt-Gamma} and \eqref{eq:est-pT-fluid-compts}. 
\end{proof}

The next lemma in this subsection investigates the commutator between the second-order operator  $\mL$ and other first-order  operators. 

\begin{lem}\label{lem:Fluid-Commutators}
The following estimates hold
\eq{\alg{\notag
\| [\delu, \mL] V\|_{H^k}&\lesssim \varepsilon e^{\max\{-1+\mu,-\lambda\} T}\|V \|_{H^{k+2}}+ \Lambda(T)\|\p_T V\|_{H^{k+1}}, 
\qquad &0\leq k \leq N-3,
\\
\| [\bnab_m,\mL]V \|_{H^{k}}&\lesssim \varepsilon e^{-\lambda T} \| V \|_{H^{k+2}}+ \| V\|_{H^{k}}, 
\qquad &0\leq k \leq N-2,
\\
\| [\p_T, \mL] V\|_{H^k} &\lesssim \varepsilon e^{-\lambda T} \|V \|_{H^{k+2}},\qquad &0\leq k \leq N-2.
}}
Also for $k,s\in \mathbb{Z}$ such that $0 \leq k \leq N-2$, $s \geq 1$, and $2(s-1)+k \leq N-1$, we have
\eq{\alg{\notag
\| [\p_T, \mL^s] V\|_{H^k} 
&\lesssim \varepsilon e^{-\lambda T}\| V \|_{H^{k+2+2(s-1)}} .
}}

\end{lem}

\begin{proof}
A calculation yields
\begin{align}
\notag
[\delu, \mL] V_{ij}&=  -\big(\delu g^{ab}\big)\bnab_a \bnab_b V_{ij}+\mD \hu \cdot \p_T V_{ij}+2g^{ab} \bnab_a \hu \bnab_b \p_T V_{ij}
\\&\quad
\label{eq:Comm-delu-mL-k}
+\tau u^c g^{ab} 
\big(\Riem[\gamma]^k{}_{bac}\bnab_k V_{ij}
+4 \Riem[\gamma]^k{}_{(i|ac} \bnab_b V_{k|j)}+ 2\bnab_a \Riem[\gamma]^k{}_{(i|bc} \cdot V_{k|j)}\big)
\\&\quad \notag
-2\tau g^{ab} \bnab_a u^c \bnab_b\bnab_c V_{ij} - \tau \mD u^c \cdot \bnab_c V_{ij} + 2\tau u^c \bnab_c \Riem[\gamma]_{iajb}\cdot V^{ab}.
\end{align}
It is useful to write this schematically, using that $\text{Riem}[\gamma]$ and its derivatives are bounded by constants:
\eq{\alg{\notag
[\delu, \mL] V&= (\delu g^{-1}) \ast \bnab^2 V + \bnab^2 \hu \ast \p_T V+ \bnab u^0 \ast \bnab \p_T V 
\\&\quad
+\tau \big( u^c \ast (V + \bnab V )
+  \bnab u^c \ast \bnab^2 V +  \bnab^2 u^c \ast \bnab V  \big).
}}
If $2\leq k \leq N-3$ then by elliptic regularity of Lemma \ref{lem:elliptic-regularity-mL} and the commutator estimates in Lemma \ref{lem:Commutator-nabla-bnab},
\eq{\alg{\label{eq:ReturningEq}
\| [\delu, \mL] V\|_{H^k}
&\lesssim \| \delu g\|_{H^k}\|V \|_{H^{k+2}}
+ \| \hu\|_{H^{k+2}}\|\p_T V\|_{H^{k+1}}
+ |\tau| \| u\|_{H^{k+2}}\|V\|_{H^{k+2}}.
}}
The conclusion then holds by   \eqref{eq:est-u0} and by estimating $\delu g$ using Corollary \ref{corol:Estimates-delu-h-v} . The estimates when $k=0,1$ follow in a similar way. 

Next we compute the identity
\eq{\alg{\label{eq:Comm-mL-bnab}
[\mL, \bnab_m]V_{ij} &= 
 \bnab_m g^{ab} \cdot \bnab_a\bnab_b V_{ij} - g^{ab} [\bnab_a\bnab_b, \bnab_m]V_{ij},
}}
where we are thinking of the $m$ index as not being free (i.e. contracted with a factor of the shift $X^m$). 
Since the commutator involving only $\bnab$ will just generate background Riemann curvature components the required estimate follows straightforwardly. We note also that \eqref{eq:Comm-mL-bnab} and the elliptic regularity of Lemma \ref{lem:elliptic-regularity-mL} imply, for $s \in \mathbb{Z}$ such that $1\leq s \leq N/2$,
\eq{\alg{\label{eq:Comm-mLk-bnab-L2}
\| [\mL^s, \bnab_m]V \|_{L^2} 
&\lesssim 
\| [\mL, \bnab]\mL^{s-1} V\|_{L^2} + \cdots + \| [\mL,\bnab]V\|_{H^{2(s-1)}}
\\&
\lesssim 
\|g-\gamma\|_{H^{\max\{3, 2s-1\}}}\big(\| V \|_{H^{2s}} + \| g-\gamma\|_{H^{2s}}\| V\|_{H^{2s-1}}\big) + \| V\|_{H^{2(s-1)}}
\\&\lesssim
\varepsilon e^{-\lambda T} \| V \|_{H^{2s}} + \| V\|_{H^{2(s-1)}}.
}}

Finally we compute
\eq{\alg{\notag
[\p_T, \mL]V_{ij} &= - [\p_T, \mD]V_{ij} = -(\p_T g^{ab}) \bnab_a \bnab_b V_{ij}.
}}
Using Lemma \ref{lem:Estimates-pT-h-v} to estimate $\p_T g$ this gives, for $2\leq k \leq N-2$,
\eq{\notag
\| [\p_T, \mL] V\|_{H^k} \lesssim \| \p_T h\|_{H^k} \big(\|\nabla(\nabla V +\Upsilon\ast V)\|_{H^{k}} + \|\Upsilon (\nabla V+\Upsilon \ast V) \|_{H^{k}} \big)
\lesssim \varepsilon e^{-\lambda T} \|V \|_{H^{k+2}}.
}
A similar argument holds for the cases $k=0,1$. At higher order, we obtain the identity
\eq{\alg{\label{eq:Commutator-dT-mLk}
[\p_T, \mL^s] V_{ij} 
&= -\sum_{1\leq i \leq s} \mL^{i-1}  (\p_T g^{a_i b_i})\cdot \bnab_{a_i}\bnab_{b_i} \big( \mL^{s-i} (V_{ij})\big).
}}
This can be estimated, for $2\leq k \leq N-2$, by
\eq{\alg{\label{eq:comm789}
\| [\p_T, \mL^s] V\|_{H^k} 
&\lesssim\| \p_T h\|_{H^{2(s-1)+k}} \big(\| V \|_{H^{k+2+2(s-1)}} + \| g-\gamma\|_{H^{k+2}}\| V\|_{H^{k+1+2(s-1)}}\big).
}}
The cases $k=0,1$ are treated in a similar  way and the conclusion follows from Lemma \ref{lem:Estimates-pT-h-v}. 
\end{proof}

The next corollary extends the commutator estimates of the previous lemma to higher-orders of $\mL^s$, $s \in \mathbb{Z}_{\geq 1}$. Note that the $L^2$ estimate that appears in the statement will be typically applied with $V=h$, while the lower-order $H^1$ estimate will be primarily used later on with $V = \Sigma$. 

\begin{corol}
\label{corol:Com-delu-mL-estimate}
We have,
\eq{\alg{\notag
\|[ \delu, \mL^s] V \|_{L^2} 
&\lesssim 
\varepsilon e^{\max\{-1+\mu,-\lambda\} T}\| V\|_{H^{2s}}
+ \Lambda(T)\|\p_T V\|_{H^{2s-1}} , &\quad 1\leq s \leq \ell,
\\
\|[ \delu, \mL^{s-1}] V \|_{H^1} 
&\lesssim 
\varepsilon e^{\max\{-1+\mu,-\lambda\} T}\| V\|_{H^{2s-1}}
+ \Lambda(T)\|\p_T V\|_{H^{2s-2}} , &\quad 2\leq s \leq \ell.
}} 
\end{corol}

\begin{proof}
By Lemma \ref{lem:Fluid-Commutators},
\eq{\alg{\notag
\|[ \delu, \mL^s] V \|_{L^2} 
&\lesssim
\|[ \delu, \mL] \mL^{s-1}V \|_{L^2} 
+ \|[ \delu, \mL] \mL^{s-2}V \|_{H^2} 
+ \cdots
+ \|[ \delu, \mL] V \|_{H^{2(s-1)}} 
\\
&\lesssim 
\varepsilon e^{\max\{-1+\mu,-\lambda\} T}\| V\|_{H^{2s}}
+\Lambda(T)\|\p_T V\|_{H^{2s-1}} 
+ \Lambda(T) \sum_{p=0}^{s-2}\|[\p_T, \mL^{s-1-p}]V\|_{H^{1+2p}}.
}} 
The conclusion then easily follows and the second estimate follows in the same way. 
\end{proof}

\begin{corol}\label{corol:delu-mL-Sigma}
\eq{\alg{\notag
\| \delu\mL^{\ell-1}(\Sigma)\|_{L^2}
&\lesssim 
(E^g_{N-1})^{1/2} +\Lambda(T),
\\
\| \delu\mL^{\ell-1}(\Sigma)\|_{H^1}
&\lesssim 
\| \Sigma\|_{H^{N-1}}+ \| g-\gamma\|_{H^{N}} + \Lambda(T).
}}
\end{corol}

\begin{proof}
By  the $\delu \Sigma, \p_T \Sigma$ estimates of Lemma \ref{lem:Estimates-pT-h-v}, Corollary \ref{corol:Estimates-delu-h-v},  and the previous commutator estimate of corollary \ref{corol:Com-delu-mL-estimate}
\eq{\alg{\notag
\| \delu\mL^{\ell-1}(\Sigma)\|_{L^2}
&\lesssim
\| \delu \Sigma \|_{H^{N-3}} 
+\| [\delu,\mL^{\ell-1}](\Sigma)\|_{L^2}
\lesssim
\| \Sigma\|_{H^{N-2}}+ \| g-\gamma\|_{H^{N-1}} +\Lambda(T).
}}
The conclusion then follows by the coercive estimate of Lemma \ref{lem:Geom-coercive-est}.
Similarly, 
\eq{\alg{\label{eq:nab-delu-mL-Sigma}
\| \delu\mL^{\ell-1} \Sigma\|_{H^1}
&\lesssim
 \| \Sigma\|_{H^{N-1}}+ \| g-\gamma\|_{H^{N}} + \Lambda(T)
+ \varepsilon e^{\max\{-1+\mu,-\lambda\} T}\|\Sigma\|_{H^{N-2}}
\\&\quad
+ \Lambda(T)\|\p_T \Sigma\|_{H^{N-3}} 
\\
&\lesssim
\| \Sigma\|_{H^{N-1}}+ \| g-\gamma\|_{H^{N}} + \Lambda(T).
}}
\end{proof}

\begin{rem}\label{rem:356}
Frequently in our energy estimates we will need to study the term $\| \delu\mL^{\ell-1}(\Sigma)\|_{H^1}$ appearing in the previous Corollary. 
However, by looking at the estimate derived in Corollary \ref{corol:delu-mL-Sigma}, we see that we cannot apply Lemma \ref{lem:Geom-coercive-est} to the top-order Sobolev norms $\| \Sigma\|_{H^{N-1}}$ and $\| g-\gamma\|_{H^{N}} $. To estimate these terms by the geometric energy $\Etot$ we will instead need to use the auxiliary elliptic estimates established in Section \ref{sec:EllipticEstimate}.
\end{rem}

We conclude this subsection with a commutator estimate that plays an important role in the proof of the lapse estimate appearing in Proposition \ref{prop:delu-commuted-lapse}. 

\begin{lem}
\label{lem:comm-delubnab-Delta}
Let $\Delta \define g^{ab}\nabla_a\nabla_b$. Then for $0\leq k \leq N-3$,
\eq{\alg{\notag
\| [\Delta, \delu]V\|_{H^k}
&\lesssim \Lambda(T) \|\p_T V\|_{H^{k+1}}
+ \varepsilon e^{\max\{-1+\mu,-\lambda\} T}\| V\|_{H^{k+2}}.
}}
\end{lem}

\begin{proof}
We compute
\eq{\alg{\notag
[\Delta, \delu]V_{ij}
&= \Delta \hu \cdot \p_T V_{ij}+2\nabla^a\hu\nabla_a(\p_T V_{ij}) 
- u^0 \p_T g^{ab} \cdot \nabla_a\nabla_v V_{ij}
+ u^0 g^{ab}[\nabla_a\nabla_b, \p_T]V_{ij}
\\&\quad
+ \tau \big( -\Delta u^c \cdot \bnab_c V_{ij}
- 2 \nabla_a u^c \nabla^a \bnab_c V_{ij}+ u^c\bnab_c g^{ab} \cdot \nabla_a \nabla_b V_{ij}
\big)
+ \tau u^c g^{ab} [\bnab_c, \nabla_a \nabla_b]V_{ij}.
}}
Let $2\leq k \leq N-3$. From Lemma  \ref{lem:Commutator-nabla-bnab} and  \eqref{eq:Commutator-pT-nablas} we find
\eq{\alg{\notag
\| [\Delta, \delu]V\|_{H^k}
&\lesssim \| \nabla \hu \|_{H^{k+1}}\|\p_T V\|_{H^{k+1}}
+ \| \p_T h \|_{H^{k}} \| V\|_{H^{k+2}}
+ \| \p_T \Gamma\|_{H^{k+1}}\|V\|_{H^{k+1}}
\\&\quad
+ |\tau| \| u \|_{H^{k+2}}\| V\|_{H^{k+2}}.
}}
The cases $k=0,1$ follow in the same way. The conclusion follows using Lemma \ref{lem:Estimates-pT-h-v}, \eqref{eq:est-u0} and \eqref{eq:est-pt-Gamma}.
\end{proof}

\subsection{Second Geometric Components Estimates}\label{sec:Geom2}
We now reach the final subsection of Section \ref{sec:FurtherPrelims}. 
The first lemma is an analogue of Lemma \ref{lem:Geom-coercive-est} for our top-order \emph{$\delu-$boosted} geometric energy. The proof follows those in \cite[Lemma 19]{AF20} and \cite[Lemma 7.2]{AM11}.

\begin{lem}
\label{lem:Properties-of-Egdelu}
Let $s\in \mathbb{Z}$ such that $1\leq s \leq \ell$. 
There is a $\delta>0$ and a constant $C>0$ such that for $(g,\Sigma)\in B_{\delta}((\gamma,0))$ the inequality
\eq{\alg{\notag
 (\delu \mL^s h, \delu \mL^s h)_{\Lgg}  +  |(\delu\mL^s v, \delu \mL^{s-1}v)_{\Lgg}|& \leq C E^g_{\delu, 2\ell}
}}
holds. Furthermore $E^g_{\delu, 2\ell} \geq 0$. 
\end{lem}
\begin{proof}
Recall $E^g_{\delu, 2\ell}$ from Definition \ref{defn:delu-geom-energy}. We first note that $E^g_{\delu, 2\ell}|_{(h,v)=(\gamma,0)} =0$. Next, we see that $(\gamma, 0)$ is a critical point of $E^g_{\delu, 2\ell}$ since the first  derivative vanishes. Considering then the second derivative of this energy at $(\gamma,0)$, we see that the Hessian takes the form
\eq{\alg{\notag
D^2 &(\mathcal{E}^g_{\delu, 2s}+c_E \Gamma^g_{\delu, 2s})((h,k),(h,k))
\\
&= 9 (\delu\mathcal{L}_{\gamma,\gamma}^s h, \delu\mathcal{L}_{\gamma,\gamma}^s h)_{\Lgg}+ (\delu \mathcal{L}_{\gamma,\gamma}^s k, \delu \mathcal{L}_{\gamma,\gamma}^{s-1} k)_{\Lgg} + c_E(\delu \mathcal{L}_{\gamma,\gamma}^{s-1} k, \delu \mathcal{L}_{\gamma,\gamma}^s h)_{\Lgg}.
}}
We claim that the Hessian is non-negative.  By expanding in terms of the eigentensors of $\mathcal{L}_{\gamma, \gamma}$, we are left with terms of the type
\eq{\alg{\notag
\lambda^{2s-1} \Big( 9 \lambda(\delu P_\lambda h, \delu P_\lambda h)_{\Lgg}+  (\delu P_\lambda k, \delu P_\lambda k)_{\Lgg} + c_E(\delu P_\lambda k, \delu  P_\lambda h)_{\Lgg}\Big),
}}
where $P_\lambda$ denotes the projection operator onto the $\lambda$-eigenspace. 
The choice of $c_E$ ensures that the bracketed term is non-negative for the smallest eigenvalue $\lambda_0$, which in turn implies non-negativity for all eigenvalues. Thus we find 
\eq{\alg{\notag
D^2 (\mathcal{E}^g_{\delu, 2s}+c_E \Gamma^g_{\delu, 2s})((h,k),(h,k))\geq 0.
}}
From this it follows that there is a constant $C = C(\lambda_0, \gamma)>0$ such that 
\eq{\alg{\notag
\sum_{s=1}^\ell (\delu \mL^s h, \delu \mL^s h)_{\Lgg}  +  |(\delu\mL^s k, \delu \mL^{s-1}k)_{\Lgg}|& \leq C E^g_{\delu, 2\ell}.
}}
\end{proof}

The next lemma provides estimates for the Sobolev norms of $\delu h, \delu \Sigma$ in terms of the geometric energy $\Etot$. This is natural given that we have constructed the functional $\Etot$ to precisely control such Sobolev norms. 

\begin{lem}[Dust derivatives of geometric variables at high-regularity]
\label{lem:delu-h-Sigma}
We have,
\eq{\alg{\notag
\|\delu h \|_{H^{N-1}}^2
&\lesssim \Edelu
+\varepsilon^2 e^{2\max\{-1+\mu,-\lambda\} T}E^g_{N-1} + \Lambda(T)^4,
\\
\|\delu\Sigma\|_{H^{N-2}}^2
&\lesssim
\Edelu
+ E^g_{N-1}
+
\varepsilon e^{\max\{-1+\mu,-\lambda\} T}\big(\| \Sigma\|_{H^{N-1}}^2 + \| g-\gamma\|_{H^N}^2\big)
+  \Lambda(T)^2.
}}
\end{lem}

\begin{rem}
Similar to Remark \ref{rem:356}, we cannot apply Lemma \ref{lem:Geom-coercive-est} to the top-order Sobolev norms $\| \Sigma\|_{H^{N-1}}^2$ and $\| g-\gamma\|_{H^{N}}^2$. To estimate these terms by the geometric energy $\Etot$ we will instead need to use the auxiliary elliptic estimates of Section \ref{sec:EllipticEstimate}.
It is also crucial in later analysis in Section \ref{sec:EllipticEstimate} that these top-order norms above appear on the right hand side in Lemma \ref{lem:delu-h-Sigma} with a smallness factor of $\varepsilon$. 
\end{rem}

\begin{proof}[Proof of Lemma \ref{lem:delu-h-Sigma}]
Recall $N \define 2\ell+1$.
By the elliptic regularity of Lemma \ref{lem:elliptic-regularity-mL} and Lemma \ref{lem:Properties-of-Egdelu}
\eq{\alg{\notag
\|\delu h \|_{H^{2\ell}}^2\lesssim 
\| \mL^\ell \delu h \|_{L^2}^2
\lesssim
\Edelu + \|[ \delu, \mL^\ell] h \|_{L^2}^2.
}}
We control the commutator term using  Corollary \ref{corol:Com-delu-mL-estimate}, finding
\eq{\alg{\label{eq:delu-mL-h}
\|[ \delu, \mL^\ell] h \|_{L^2}
&\lesssim
\varepsilon e^{\max\{-1+\mu,-\lambda\} T}\| g-\gamma\|_{H^{N-1}}
+\Lambda(T) \|\p_T h\|_{H^{N-2}} 
\\
&\lesssim
\varepsilon e^{\max\{-1+\mu,-\lambda\} T}(E^g_{N-1})^{1/2} + \Lambda(T)^2,
}}
where in the final line we used Lemma \ref{lem:Estimates-pT-h-v} and the coercive estimate of Lemma \ref{lem:Geom-coercive-est}. 

Next, by elliptic regularity and Lemma \ref{lem:H1-mL}, we have
\eq{\alg{\label{eq:delu102}
\|\delu \Sigma \|_{H^{2\ell-1}}^2
&\lesssim 
\| \mL^{\ell-1} \delu \Sigma \|_{H^1}^2
\\&
\lesssim 
\| \delu \mL^{\ell-1} \Sigma \|_{L^2} ^2
+ 
\|[ \delu, \mL^{\ell-1}] \Sigma \|_{L^2}^2
+ (\mL^{\ell-1} \delu \Sigma , \mL^\ell \delu\Sigma)_{\Lgg}.
}}
The first term on the RHS of \eqref{eq:delu102} is treated by Corollary \ref{corol:delu-mL-Sigma}. For the commutator term in \eqref{eq:delu102}, we use  Lemma \ref{lem:Estimates-pT-h-v} and Corollary \ref{corol:Com-delu-mL-estimate} to find
\eq{\alg{\label{eq:delu-mL-Sigma}
\|[ \delu, \mL^{\ell-1}] \Sigma \|_{L^2} 
&\lesssim 
\varepsilon e^{\max\{-1+\mu,-\lambda\} T}\| \Sigma\|_{H^{N-3}}
+ \Lambda(T) \|\p_T \Sigma\|_{H^{N-4}} 
\\&
\lesssim
\varepsilon e^{\max\{-1+\mu,-\lambda\} T}(E^g_{N-2})^{1/2}
+ \Lambda(T)^2.
}} 
Considering next the final term in \eqref{eq:delu102}, we write it as
\eq{\alg{\notag
(\mL^{\ell-1} \delu \Sigma , \mL^\ell \delu\Sigma)_{\Lgg} 
&= (\delu \mL^{\ell-1} \Sigma , \delu\mL^\ell \Sigma)_{\Lgg} 
+ ([\mL^{\ell-1}, \delu] \Sigma , \delu \mL^\ell \Sigma)_{\Lgg} 
\\&\quad
+ (\delu \mL^{\ell-1} \Sigma , [\delu,\mL^\ell ]\Sigma)_{\Lgg} 
+ ([\mL^{\ell-1}, \delu] \Sigma , [\delu, \mL^\ell]\Sigma)_{\Lgg} 
\\&=: E_1 + E_2 + E_3 + E_4.
}}
From Lemma \ref{lem:Properties-of-Egdelu}, $|E_1| \lesssim \Edelu$. 
To estimate $E_2$ we need to integrate by parts one of the derivatives appearing in $\delu\mL^\ell\Sigma$. Using \eqref{eq:IBP} we find
\eq{\alg{\notag
E_2 
&= (g^{ab} \bnab_a [\mL^{\ell-1}, \delu] \Sigma , \bnab_b \delu \mL^{\ell-1} \Sigma)_{\Lgg} 
- 2 ([\mL^{\ell-1}, \delu] \Sigma , \Riem[\gamma]\circ \delu \mL^{\ell-1} \Sigma)_{\Lgg} 
\\&\quad
+ ([\mL^{\ell-1}, \delu] \Sigma , [\delu, \mL] \mL^{\ell-1} \Sigma)_{\Lgg} .
}}
By the commutator estimates of Lemma \ref{lem:Fluid-Commutators}, and Lemma \ref{lem:Estimates-pT-h-v}, we have
\eq{\alg{\label{eq:delu-mL-mL-Sigma}
\| [\delu,\mL]\mL^{\ell-1}\Sigma\|_{L^2}
&\lesssim 
\varepsilon e^{\max\{-1+\mu,-\lambda\} T}\|\Sigma\|_{H^{N-1}}+ \Lambda(T)\|\p_T \Sigma\|_{H^{N-2}}
\\&
\lesssim
\varepsilon e^{\max\{-1+\mu,-\lambda\} T}\big(\| \Sigma\|_{H^{N-1}} + \| g-\gamma\|_{H^N}\big) + \Lambda(T)^2. 
}}
Using this, together with \eqref{eq:delu-mL-Sigma} and Corollary \ref{corol:delu-mL-Sigma}, gives
\eq{\alg{\notag
\big|E_2 \big|
&\lesssim
\| [\delu, \mL^{\ell-1}] \Sigma\|_{H^1}  \Big(\|\delu\mL^{\ell-1}\Sigma \|_{H^1} + \| [\delu,\mL]\mL^{\ell-1}\Sigma\|_{L^2}\Big)
\\
&\lesssim
\varepsilon e^{\max\{-1+\mu,-\lambda\} T}\big(\| \Sigma\|_{H^{N-1}}^2 + \| g-\gamma\|_{H^N}^2\big) + \Lambda(T)^2 .
}}
The terms $E_3, E_4$ are similarly estimated and  inserting all these estimates into \eqref{eq:delu102} gives the required result. 
\end{proof}

We end this subsection with two Lemmas concerning the geometric source terms $F_h$ and $F_v$. 

\begin{lem}[Time derivative of geometric source terms]
\label{lem:Est-pT-Fh-Fv}
We have,
\eq{\alg{\notag
\| \p_T F_h \|_{H^{N-2}} 
+ \| \p_T F_v \|_{H^{N-3}} 
 &\lesssim 
E^g_{N-1}+ \Lambda(T).
}}
\end{lem}

\begin{proof}
Let $2\leq k \leq N-2$. 
Taking a time derivative of $F_h$ as given in Definition \ref{defn:Fh-Fv} we see, 
\eq{\alg{\notag
\| \p_T F_h\|_{H^k}
&\lesssim 
\| \p_T N \|_{H^k}
+ \big(\| \hN \|_{H^k} + \| X\|_{H^{k+1}} + \| g-\gamma \|_{H^{k+1}} \| X \|_{H^k} \big) \| \p_T h \|_{H^k}
\\& 
+\| g-\gamma\|_{H^k} \big( \| \p_T X \|_{H^{k+1}} + \| g-\gamma \|_{H^{k+1}}\| \p_T X \|_{H^{k}}\big).
}}
The first estimate then follows by applying Lemma \ref{lem:Estimates-pT-h-v}.

Next, let $2\leq k \leq N-3$. We take a time derivative of $F_v$ which gives
\eq{\alg{\notag
\p_T F_v
&= \p_T \nabla_i \nabla_j N 
+ \p_T N \cdot \big( 2 \Sigma_{ic} \Sigma^c_j
-\tfrac{1}{3} g_{ij}- \Sigma_{ij} + \tau S_{ij}\big)
+ 2 N \p_T(\Sigma_{ic} \Sigma^c_j)
-\tfrac{1}{3}\hN   \p_T g_{ij} 
\\& \quad 
- \hN \p_T \Sigma_{ij}
- \tau N S_{ij} + \tau N \p_T S_{ij}
- \p_T\big( v_{im}\bnab_j X^m + v_{mj}\bnab_i X^m\big).
}}
For the first term we use the commutator identity from \eqref{eq:Commutator-pT-nablas} and recall that the lapse is a scalar. We obtain 
\eq{\alg{\notag
\| \p_T F_v \|_{H^k} 
&\lesssim
\| \p_T N\|_{H^{k+2}} + \| [\p_T, \nabla]\nabla N\|_{H^k}
+ |\tau| \big( \| S\|_{H^k} + \| \p_T S\|_{H^k} \big)
\\&\quad
+ \| \Sigma\|_{H^k}^2 + \| \p_T \Sigma\|_{H^k}^2 + \| \p_T g\|_{H^k}^2
+ \Lambda(T)^2.
}}
The conclusion then follows by Lemma \ref{lem:Estimates-pT-h-v} and the matter estimates in Lemma \ref{lem:matter-first-estimates} and Lemma \ref{lem:Est-pT-eta-S}. 
\end{proof}

\begin{corol}[Dust derivative of geometric source terms]
\label{corol:Est-delu-Fh-Fv}
We have,
\eq{\alg{\notag
\| \delu F_h \|_{H^{N-2}}  + \|\delu F_v\|_{H^{N-3}}
 &\lesssim 
E^g_{N-1} + 
\varepsilon e^{(-1+\mu)T}\big( \| g-\gamma\|_{H^N}^2 + \| \Sigma\|_{H^{N-1}}^2 \big)+
\Lambda(T),
\\
\| \delu F_h \|_{H^{2}}
 &\lesssim 
E^g_{N-1} + 
\Lambda(T).
}}
\end{corol}

\begin{proof}
The estimates follow by expanding out $\delu = u^0 \p_T - \tau u^c \bnab_c$ and using the $F_h, F_v$ estimates in Lemma \ref{lem:Est-Fh-Fv}  and the $\p_T F_h, \p_T F_v$ estimates from Lemma \ref{lem:Est-pT-Fh-Fv}. 
\end{proof}

\begin{rem}
Since we are dealing with geometric variables in the above corollary, and not matter variables, we roughly estimated the dust derivatives as $\delu \sim \p_T + \tau u^c \ast \nabla$. Note that doing so introduced the top-order Sobolev norms $\| \Sigma\|_{H^{N-1}}^2$ and $\| g-\gamma\|_{H^{N}}^2$. Crucially for later analysis in Corollary \ref{corol:Coercive-top-order-geom}, however, is that these top-order norms appear with a coefficient of $\varepsilon$.
\end{rem}

\section{Elliptic Estimate}\label{sec:EllipticEstimate}
In this section we prove an auxiliary elliptic estimate which allows us to control the top-order Sobolev norms of $g$ and $\Sigma$ in terms of the geometric energy functional $\Etot$. Recall only the lower-order Sobolev norms are controlled using Lemma \ref{lem:Geom-coercive-est}, and so a new idea is indeed needed to cover the top-order of regularity. We also remind the reader that the geometric energy $\Etot$ will eventually fulfill a strong decay estimate enabled by the inclusion of certain correction terms. 

The main result of the section, Corollary \ref{corol:Coercive-top-order-geom}, achieves the goal of the previous paragraph. To prove this corollary, we take the first-order equations of motion for $\p_T h, \p_T \Sigma$ appearing in \eqref{eq:EoM-hv} and convert them into second-order equations involving a perturbed wave operator $\mathcal{W}$ and the $\delu$ derivatives. 
We find, very schematically, that 
\eq{\label{eq:schematic-Elliptic}
\mL h \sim \mathcal{W}(h) + \delu \p_T (h) + \tau u^a \bnab_a\delu h.
}
By elliptic regularity for $\mL$, we can then prove an estimate on $\| g-\gamma\|_{H^{N}} \sim \| \mL h \|_{H^{N-2}}$ by estimating the RHS of \eqref{eq:schematic-Elliptic}, see Proposition \ref{prop:Elliptic-Est}. A similar idea holds also for $\Sigma$. We note that this idea, albeit for a different gauge, was first introduced by  Had\v{z}i\'{c} and Speck in \cite{HaSp15}. 

\begin{defn}[Operators $\mathcal{W},\mathcal{H}$]
Define the operators
\eq{\alg{\notag
\mathcal{W}&\define \p_T\p_T-X^aX^b\bnab_a \bnab_b + N^2 \mL,  
\qquad
\mathcal{H} \define N^2 \mL  - X^aX^b\bnab_a \bnab_b - \tau^2 \frac{u^a u^b}{(u^0)^2} \bnab_b \bnab_a ,
}}
which act on symmetric $(0,2)$-tensors.
\end{defn}
Due to the sign convention on $\mL$, see Definition \ref{defn:mL}, one can think of $\mathcal{W}$ as being a kind of perturbed wave operator. 

\begin{lem}[Wave equations for $h, v$]
\label{lem:mathcalW-EoM}
The differential equations \eqref{eq:EoM-hv} for $h=g-\gamma$ and $v = 6 \Sigma$ imply 
\eq{\notag
\mathcal{W}(h)= F_1, \qquad
\mathcal{W}(v)= F_2,
}
where, 
\eq{\alg{\notag
\| F_1 \|_{H^{N-2}}^2 + \| F_2 \|_{H^{N-3}}^2 &\lesssim 
E^g_{N-1}
+ \varepsilon e^{-2\lambda T}\big(\| g-\gamma\|_{H^N}^2 
+ \| \Sigma\|_{H^{N-1}}^2\big)
+ \Lambda(T)^2.
}}
\end{lem}

\begin{proof}
Rearranging \eqref{eq:EoM-hv}  as $v = w^{-1}(\p_T h+X^m\bnab_m h-F_h)$ and substituting this into \eqref{eq:EoM-hv} gives
\eq{\alg{\notag
-w^{-1}&(\p_T w )v
+ w^{-1} (\p_T^2 h +\p_T(X^m\bnab_m h)-\p_T F_h)
= -2v - 9 w\mL h  - X^m \bnab_m v +6F_v.
}}
Using again \eqref{eq:EoM-hv}  we note that
\eq{\alg{\notag
\p_T (X^a \bnab_a h) 
&= -X^a X^b\bnab_a\bnab_b h -X^a\bnab_a X^b \cdot \bnab_b h+\p_T X^a\cdot \bnab_a h +X^m\bnab_m (wv+F_h).
}}
Rearranging terms (recall $9w^2 = N^2$) we find
\eq{\alg{\notag
\mathcal{W}(h) = F_1 &\define 
X^a\bnab_a X^b \cdot \bnab_b h-\p_T X^m\cdot \bnab_m h -X^m\bnab_m (wv+F_h)
 +\p_T F_h 
+v \p_T w\\
& \qquad -2wv - wX^m\bnab_m v +6wF_v.
}}
Using Lemma \ref{lem:Est-Fh-Fv} and Lemma \ref{lem:Est-pT-Fh-Fv} we see
\eq{\alg{\notag
\| F_1 \|_{H^{N-2}}
&\lesssim
\| \Sigma\|_{H^{N-2}}
+ \| \p_T F_h\|_{H^{N-2}}
+ \| F_v \|_{H^{N-2}} + \Lambda(T)
\\
&\lesssim
E^g_{N-1}
+ \| g-\gamma\|_{H^N}^2 
+ \| \Sigma\|_{H^{N-1}}^2
+ \Lambda(T).
}}
Although the top-order terms involving $g-\gamma$ and $\Sigma$ here look worrying, when squaring the estimate we can then apply the bootstraps to gain the crucial factor of $\varepsilon$. 

To derive the equation for $v$ we take the $\p_T$ derivative of \eqref{eq:EoM-hv}:
\eq{\alg{\notag
\p_T^2 v &= -2 \p_T v - 9\p_T w \cdot \mL h 
- 9w \mL(\p_T h) - 9w[\p_T, \mL]h 
- \p_T X^m \cdot \bnab_m v 
- X^m \bnab_m (\p_T v) + 6 \p_T F_v.
}}
Substituting in the equation of motion \eqref{eq:EoM-hv} where needed, and expanding as $\mL (wv) = w\mL v-v\mD w-2g^{ab}\bnab_a w \bnab_b v$, we obtain
\eq{\alg{\notag
\mathcal{W}(v)= F_2 &\define 4v+18 w \mL h + 2X^m \bnab_m v-12 F_v-9\p_T w \cdot \mL h +9w v\mD w + 18 w g^{ab} \bnab_a w \bnab_b v
\\&\quad
 +9w\mL(X^m\bnab_m h) 
- 9w\mL F_h -9w[\p_T, \mL] h 
+6\p_T F_v 
- \p_T X^m \cdot \bnab_m v+2X^m \bnab_m v
\\&\quad
+9X^m\bnab_m(w\mL h) 
+ X^a \bnab_a X^m\cdot \bnab_m v - 6 X^m \bnab_m F_v.
}}
Using the $F_h, F_v$ estimates in Lemma \ref{lem:Est-Fh-Fv} and Lemma \ref{lem:Est-pT-Fh-Fv}, together with the commutator estimate in Lemma \ref{lem:Fluid-Commutators}, we find
\eq{\alg{\notag
\| F_2 \|_{H^{N-3}}
&\lesssim
\| \Sigma\|_{H^{N-3}} + \|g-\gamma\|_{H^{N-1}}
+ \| [\p_T, \mL] h\|_{H^{N-3}}
+ \| F_v \|_{H^{N-2}}
\\&\quad
 + \| \p_T F_v \|_{H^{N-3}}
+ \| F_h \|_{H^{N-1}}
+ \Lambda(T)
\\
&\lesssim
(E^g_{N-1})^{1/2}
+\varepsilon e^{-\lambda T} \|g-\gamma\|_{H^{N-1}} + \|g-\gamma\|_{H^{N}}^2 + \|\Sigma\|_{H^{N-1}}^2 + E^g_{N-1} + \Lambda(T).
}}
Note that in the above we also used \eqref{eq:Estimate-Upsilon} to estimate a term of the form:
\eq{\notag
\| \bnab h\|_{H^{N-1}} \lesssim 
\| h \|_{H^{N}} + \| \Upsilon\|_{H^{N-1}}\| h \|_{H^{N-1}} \lesssim \|g-\gamma\|_{H^{N}} .
}
\end{proof}

\begin{prop}[Elliptic estimate using $\mathcal{W}, \delu$ operators]
\label{prop:Elliptic-Est}
Let $V$ be an arbitrary $(0,2)$-tensor and $0\leq k\leq N-2$. Then,
\eq{\alg{\notag
\|V\|_{H^{k+2}} 
&\lesssim  \| \mathcal{W}(V) \|_{H^{k}}+ \|\delu \p_T (V) \|_{H^{k}}
+ |\tau| \|u^a \delu \bnab_a(V) \|_{H^{k}}.
}}
\end{prop}

\begin{proof}
From \eqref{defn:delu} we have,
\eq{\notag
\p_TV = (u^0)^{-1}\Big(\delu (V) + \tau u^a \bnab_a V \Big)
}
and thus
\eq{\notag
\p_T\p_TV = (u^0)^{-1}\Big(\delu ( \p_T V) + \tau u^a \bnab_a \p_T V\Big), 
\qquad
\p_T \bnab_b V = (u^0)^{-1}\Big(\delu (\bnab_b V) + \tau u^a \bnab_a \bnab_b V \Big).
}
Recalling $\hu = u^0 - 1/3$, we find
\eq{\alg{\label{eq:elliptic-344}
\mathcal{W}(V) &=  \p_T\p_T V 
- X^aX^b\bnab_a \bnab_b V
+ N^2 \mL V
\\
& = (u^0)^{-1}\Big(\delu (\p_T V) + \tau u^a  \p_T \bnab_a V\Big) 
- X^aX^b\bnab_a \bnab_b  V
+ N^2 \mL V
\\
&=  (u^0)^{-1}\delu (\p_T V) + \frac{\tau u^a}{u^0} \frac{1}{u^0}\left( \delu (\bnab_a V) + \tau u^b \bnab_b \bnab_a V \right)
- X^aX^b\bnab_a \bnab_b V
+ N^2 \mL V
 \\
&= (u^0)^{-1}\delu (\p_T V) + \tau u^a (u^0)^{-2}\delu (\bnab_a V) + \mathcal{H}(V).
}}

We view $N^{-2}\mathcal{H}$ as a perturbation off the elliptic operator $\mL$. 
For $0\leq k\leq N-2$,
\eq{\alg{\notag
\|(N^{-2}\mathcal{H} - \mL )V\|_{H^k} 
&\lesssim 
\|X^aX^b\bnab_a \bnab_b V\|_{H^k}  + \tau^2 \big\| \frac{u^a u^b}{(u^0)^2} \bnab_b \bnab_a V\big\|_{H^k}
\\&\lesssim
\varepsilon^2 e^{(-2+2\mu)T}\big( \|V\|_{H^{k+2}}+\|\Upsilon\|_{H^{k+1}}\|V\|_{H^{k+1}}\big)
\\&\lesssim 
\varepsilon^2 e^{(-2+2\mu)T}\|\mL V\|_{H^{k}}.
}}
Suppose now $\mathcal{H}(V) =0$. By definition of $\mathcal{H}$, and elliptic regularity of $\mL$, this implies
\eq{\alg{\notag
\| N^2 \mL(V) \|_{L^2} &= \big\|  X^aX^b\bnab_a \bnab_b V+\tau^2 (u^0)^{-2}u^a u^b \bnab_b \bnab_a V\big\|_{L^2}
\\&\lesssim 
\varepsilon^2 e^{2(-1+\mu)T} \| V\|_{H^2}
\\&\leq C (C_1)^{-1}
\varepsilon^2 e^{2(-1+\mu)T} \|\mL V\|_{L^2},
}}
for $C>0$ some constant and $C_1>0$ as in Lemma \ref{lem:elliptic-regularity-mL}. We also have a lower bound
\eq{\notag
C' \| \mL(V) \|_{L^2} \leq \| N^2 \mL(V) \|_{L^2}.
}
for another constant $C'>0$. Choosing $\varepsilon$ sufficiently small so that $C (C_1)^{-1} \varepsilon^2 e^{2(-1+\mu)T} < \varepsilon$, we see these two inequalities imply
\eq{\notag
C' \| \mL(V) \|_{L^2} < \varepsilon \| \mL(V) \|_{L^2} .
} 
For $\varepsilon$ sufficiently small this implies $\| \mL(V)\|_{L^2} = 0$ and so $V \in \ker\mL$. However, $\ker{\mL}=0$, and so for small data $\mathcal{H}$ also has trivial kernel and thus we obtain 
\eq{\notag
\|V\|_{H^{k+2}} 
\lesssim  \|N^{-2} \mathcal{H}(V)\|_{H^k}
\lesssim
\|V\|_{H^{k+2}} .
}
Putting this together with \eqref{eq:elliptic-344} we find
\eq{\alg{\notag
\|V\|_{H^{k+2}} 
&\lesssim  \| N^{-2} \mathcal{W}(V) \|_{H^{k}}+ \|(u^0)^{-1}N^{-2} \delu \p_T (V) \|_{H^{k}}
+ |\tau| \|(u^0)^{-2}u^a N^{-2} \delu \bnab_a (V) \|_{H^{k}}.
}}
\end{proof}

We can now bring together the previous results and estimate the top-order Sobolev norms of $g, \Sigma$ in terms of our geometric energy functionals $\Edelu$ and $E^g_{N-1}$. 
\begin{corol}
\label{corol:Coercive-top-order-geom}
We have,
\eq{\alg{\notag
\| g-\gamma\|_{H^N}^2 +\| \Sigma\|_{H^{N-1}}^2 &\lesssim 
\Edelu + E^g_{N-1} + \Lambda(T)^2.
}}
 \end{corol}

\begin{proof}
From Proposition \ref{prop:Elliptic-Est}, 
\eq{\alg{\notag
\|g-\gamma\|_{H^N}^2
&\lesssim  \| \mathcal{W}(h) \|_{H^{N-2}}^2+ \|\delu \p_T (h) \|_{H^{N-2}}^2
+ |\tau|^2 \|u^a \delu \bnab_a (h) \|_{H^{N-2}}^2.
}}
The first term here is treated using Lemma \ref{lem:mathcalW-EoM}. 
For the second term, we begin by using the commutator estimates in Lemma \ref{lem:Commutators-with-Delu}, and the $\p_T h$ and $\delu h$ estimates in Lemma \ref{lem:Estimates-pT-h-v} and Lemma \ref{lem:delu-h-Sigma} respectively, to find
\eq{\alg{\label{eq:8313-a}
\|\delu \bnab (h) \|_{H^{N-2}}^2
&\lesssim \| \delu h\|_{H^{N-1}}^2 
+ \varepsilon^2 e^{(-2+2\mu) T}  \| h\|_{H^{N-1}}^2  
+\Lambda(T)^2\|\p_T h \|_{H^{N-2}} ^2
\\
&\lesssim \Edelu
+\varepsilon^2 e^{2\max\{-1+\mu,-\lambda\} T}E^g_{N-1} + \Lambda(T)^3,
}}
Next, using the expression for $\p_T h$ given in \eqref{eq:EoM-hv}, together with Lemma \ref{lem:delu-N-X},  Corollary \ref{corol:delu-nab-estimates},  Lemma \ref{lem:delu-h-Sigma} and Corollary \ref{corol:Est-delu-Fh-Fv}, we find
\eq{\alg{\notag
\| \delu \p_T (h) \|_{H^{N-2}}^2
&\lesssim
\| \delu \Sigma\|_{H^{N-2}}^2 + \| \Sigma \|_{H^{N-2}}^2 \| \delu N \|_{H^{N-2}}^2 + \| \delu X \|_{H^{N-2}}^2\| \bnab h \|_{H^{N-2}}^2
\\&\quad
+ \| X\|_{H^{N-2}}^2\| \delu \bnab h \|_{H^{N-2}}^2 + \| \delu F_h \|_{H^{N-2}}^2.
\\
&\lesssim
\Edelu
+ E^g_{N-1}
+ \varepsilon e^{\max\{-1+\mu,-\lambda\} T}\big(\| \Sigma\|_{H^{N-1}}^2 + \| g-\gamma\|_{H^N}^2\big)
+ \Lambda(T)^2.
}}
Using again \eqref{eq:8313-a} we find 
\eq{\alg{\notag
|\tau|^2 \| u^a \delu \bnab(h) \|_{H^{N-2}}^2
&\lesssim |\tau|^2 \| u \|_{H^{N-2}}^2 \big(\Edelu
+\varepsilon^2 e^{\max\{-2+2\mu,-2\lambda\} T}E^g_{N-1} + \Lambda(T)^3\big)
\\&
\lesssim
\Edelu + E^g_{N-1} + \Lambda(T)^4.
}}
Putting this all together
\eq{\alg{\label{eq:h-toporder-101}
\|g-\gamma\|_{H^N}^2 
&\lesssim
\Edelu + E^g_{N-1}
+ \varepsilon \big(\| \Sigma\|_{H^{N-1}}^2 + \| g-\gamma\|_{H^N}^2\big) 
+ \Lambda(T)^2. 
}}

We follow the same steps for the $\Sigma$ estimate. 
From Proposition \ref{prop:Elliptic-Est}, 
\eq{\alg{\notag
\|\Sigma\|_{H^{N-1}}^2
&\lesssim  \| \mathcal{W}(v) \|_{H^{N-3}}^2+ \|\delu \p_T (v) \|_{H^{N-3}}^2
+ |\tau|^2 \|u^a \delu \bnab_a (\Sigma) \|_{H^{N-3}}^2.
}}
The first term here is treated using Lemma \ref{lem:mathcalW-EoM}. For the second term, we begin by using Lemma \ref{lem:Estimates-pT-h-v}, Lemma \ref{lem:delu-h-Sigma} and the commutator estimates in Lemma \ref{lem:Commutators-with-Delu} to show
\eq{\alg{\label{eq:8313-b}
\|\delu \bnab (\Sigma) \|_{H^{N-3}}^2
&\lesssim \| \delu \Sigma\|_{H^{N-2}}^2 + 
\varepsilon^2 e^{(-2+2\mu) T}  \| \Sigma\|_{H^{N-2}}^2  +\Lambda(T)^2 \|\p_T \Sigma \|_{H^{N-3}}^2
\\
&\lesssim
\Edelu
+ E^g_{N-1}
+ \varepsilon e^{\max\{-1+\mu,-\lambda\} T}\big(\| \Sigma\|_{H^{N-1}}^2 + \| g-\gamma\|_{H^N}^2\big)
+ \Lambda(T)^2.
}}
Next, by Lemma \ref{lem:Estimates-pT-h-v}, Lemma \ref{lem:delu-h-Sigma} and the commutator estimate of Lemma \ref{lem:Fluid-Commutators}, we note
\eq{\alg{\notag
\| \delu \mL h\|_{H^{N-3}}^2
&\lesssim
\| \delu h\|_{H^{N-1}}^2 + \|[ \delu,\mL] h\|_{H^{N-3}}^2
\\&
\lesssim
\Edelu
+\varepsilon^2 e^{\max\{-2+2\mu,-2\lambda\} T}E^g_{N-1} 
+ \Lambda(T)^2 \big(\| \Sigma\|_{H^{N-1}}^2 + \| g-\gamma\|_{H^N}^2\big)
+ \Lambda(T)^4.
}}
We can now use the expression for $\p_T v$ given in \eqref{eq:EoM-hv} and bring together these previous estimates:
\eq{\alg{\notag
\| \delu \p_T (v) \|_{H^{N-3}}^2
&\lesssim
\| \delu \Sigma\|_{H^{N-3}}^2 + \| \mL h \|_{H^{N-3}}^2 \| \delu N \|_{H^{N-3}}^2 + \| \delu \mL h\|_{H^{N-3}}^2
\\&\quad
+ \| \delu X \|_{H^{N-3}}^2\| \bnab \Sigma \|_{H^{N-3}}^2
+ \| X\|_{H^{N-3}}^2\| \delu \bnab \Sigma \|_{H^{N-3}}^2 + \| \delu F_v \|_{H^{N-3}}^2
\\
&\lesssim
\| \delu \Sigma\|_{H^{N-3}}^2 + \| \delu \mL h\|_{H^{N-3}}^2
 + \| \delu F_v \|_{H^{N-3}}^2 + \Lambda(T)^2
\\
&\lesssim
\Edelu
+\varepsilon e^{\max\{-1+\mu,-\lambda\} T}\big(\| \Sigma\|_{H^{N-1}}^2 + \| g-\gamma\|_{H^N}^2\big) + \Lambda(T)^2. 
}}
In the above we used Lemma \ref{lem:delu-N-X},  Corollary \ref{corol:delu-nab-estimates}, Lemma \ref{corol:Est-delu-Fh-Fv} and Lemma \ref{lem:delu-h-Sigma}.
Using again \eqref{eq:8313-b} and Lemma \ref{lem:delu-h-Sigma}, we find
\eq{\alg{\notag
|\tau|^2 \| u^a \delu \bnab(\Sigma) \|_{H^{N-3}}^2
\lesssim
\Edelu + E^g_{N-1}
+ \varepsilon e^{\max\{-1+\mu,-\lambda\} T}\Lambda(T) \big(\| \Sigma\|_{H^{N-1}}^2 + \| g-\gamma\|_{H^N}^2\big) 
+\Lambda(T)^3.
}}
Finally, putting this all together gives
\eq{\alg{\label{eq:Sigma-toporder-101}
\|\Sigma\|_{H^{N-1}}^2
&\lesssim
\Edelu + E^g_{N-1}
+ \varepsilon \big(\| \Sigma\|_{H^{N-1}}^2 + \| g-\gamma\|_{H^N}^2\big) + \Lambda(T)^2.
}}
The conclusion then follows by adding the estimates \eqref{eq:h-toporder-101} and \eqref{eq:Sigma-toporder-101} together and taking $\varepsilon$ sufficiently small so that we can absorb the $\varepsilon (\| \Sigma\|_{H^{N-1}}^2 + \| g-\gamma\|_{H^N}^2)$ term onto the left hand side. 
\end{proof}

We now provide new estimates on dust derivatives acting on our variables $N, X, g, \Sigma$ and, for the latter two variables, apply Corollary \ref{corol:Coercive-top-order-geom}.
\begin{corol}
\label{corol:delu-nab-estimates}
For $I$ a multi-index,
\eq{\alg{\notag
\sum_{|I|\leq N-1}\| \delu \nabla^I X\|_{L^2}+ \| \delu \nabla^I \hN\|_{L^2}
&\lesssim 
\Lambda(T), 
\\
\sum_{|I|\leq N-1}\| \delu \nabla^I (g-\gamma)\|_{L^2}
+ \sum_{|I|\leq N-2}\| \delu \nabla^I \Sigma\|_{L^2}
&\lesssim 
(\Edelu)^{1/2} + (E^g_{N-1})^{1/2}+ \Lambda(T).
}}
The same estimates hold with $\nabla$ replaced by $\bnab$. 
\end{corol}
\begin{proof}
By Lemma \ref{lem:delu-h-Sigma} and Corollary \ref{corol:Coercive-top-order-geom}
\eq{\alg{\notag
\|\delu h \|_{H^{N-1}} + \|\delu\Sigma\|_{H^{N-2}}
&\lesssim 
(\Edelu)^{1/2} + (E^g_{N-1})^{1/2}
+  \Lambda(T).
}}
Next, let $\nabla^{I} = \nabla_{a_1} \ldots \nabla_{a_{k}}$ where $k\define|I|\leq N-1$ and let $V$ be a $(0,2)$-tensor on $M$. Then, by \eqref{eq:pt-nabla}, \eqref{eq:est-pt-Gamma} and Lemma \ref{lem:Commutators-with-Delu},
\eq{\alg{\notag
\|[ \delu, \nabla^I] V \|_{L^2} 
&\lesssim
\|[ \delu, \nabla_{a_1}] \nabla_{a_2} \ldots \nabla_{a_k} V \|_{L^2}  
+ \cdots
+ \|[ \delu, \nabla] V \|_{H^{k-1}} 
\\&\lesssim
\varepsilon e^{\max\{-1+\mu,-\lambda \}T} \| V\|_{H^{k}} 
+ \Lambda(T)\| \p_T V\|_{H^{k-1}} 
+ \Lambda(T) \| \p_T \Gamma(g) \|_{H^{N-3}} \| V \|_{H^{k-2}}.
}}
This implies
\eq{\alg{\notag
\| \delu \nabla^I V \|_{L^2}
\lesssim
\| \delu V \|_{H^k}  
+ \varepsilon e^{\max\{-1+\mu,-\lambda\} T}\| V\|_{H^{k}} 
+ \Lambda(T)\| \p_T V\|_{H^{k-1}} .
}}
A simple check shows that when considering the appropriate commutator expression (see \eqref{eq:Commutators-with-Delu}) on the scalar lapse we will pick up $\| \hN \|_{H^k}$ terms and not $\| N \|_{H^k}$. 
The result then follow, using Lemma \ref{lem:delu-N-X} and Lemma \ref{lem:Estimates-pT-h-v}. 
\end{proof}

\begin{lem}
\label{lem:deluV-H2-H1}
Let $V$ be a $(0,2)$-tensor on $M$ and let $p\in\{1,2\}$. Then 
\eq{\alg{\notag
\| \delu V\|_{H^p} &\lesssim
\sum_{|I|\leq p} \| \delu \nabla^I V \|_{L^2}
+ \varepsilon e^{\max\{-1+\mu,-\lambda\} T}\| V\|_{H^{1}}
+\Lambda(T)\|\p_T V\|_{H^{p-1}}.
}}
\end{lem}

\begin{proof}
A straight forward application of commutator identities. 
\end{proof}

\section{Lapse and Shift Estimates}\label{sec:LapseShift}
In this section we first establish the basic estimates on the lapse, shift and their time derivatives using elliptic estimates and general formula presented in \cite{AF20}. The most exciting results lie in the novel top-order lapse and shift estimates given in Section \ref{sec:ExtraLapseShiftEst}. 

\begin{lem}
\label{lem:lapse-shift}
For $3\leq k \leq N$,
\eq{\alg{\notag
\| \hN \|_{H^k} & \lesssim	|\tau| \| \rho \|_{H^{k-2}} + \varepsilon^2 e^{-2\lambda T},
\\
\| X \|_{H^k} & \lesssim |\tau| \| \rho \|_{H^{k-3}}
	+ \varepsilon^2 e^{\max\{-2\lambda, -2+\mu\}T}.
}}
\end{lem}

\begin{proof}
Recall from \eqref{eq:EoM-lapse-shift} the lapse equation of motion
\eq{\notag
(\Delta - \tfrac13)N = N(|\Sigma|_g^2 + \tau \eta) - 1.
}
By elliptic regularity for the standard Laplacian $\Delta$ this implies
\eq{\alg{\notag
\| \hN \|_{H^k} & \leq	
	C \big( \| \Sigma \|^2_{H^{k-2}} 
	+ |\tau| \| \eta \|_{H^{k-2}} \big),
}}
and the desired result follows by Lemma \ref{lem:matter-first-estimates}. 

Similarly from the shift equation of motion \eqref{eq:EoM-lapse-shift}, we find 
\eq{\alg{\notag
\| X \|_{H^k} & \leq 
	C \big( 	\| \Sigma \|_{H^{k-2}}^2
	+ \| g-\gamma\|_{H^{k-1}}^2
	+ |\tau| \| \eta \|_{H^{k-3}}
	+ \tau^2 \| N \jmath \|_{H^{k-2}}
	\big),
}}
and the desired result follows by Lemma \ref{lem:matter-first-estimates}. 
\end{proof}

\begin{lem}\label{lem:delT-lapse-shift}
For $4\leq k \leq N-1$,
\eq{\alg{\notag
\| \p_T N \|_{H^k}  + \| \p_T X \|_{H^k} 
&\lesssim 
	\| \hN\|_{H^k}+\| X\|_{H^k} 
+ |\tau|\| \rho \|_{H^{k-2}} + \varepsilon^2 e^{\max\{-2\lambda, -2+\mu\}T}.
}}
\end{lem}

\begin{proof}
Let $4\leq k \leq N-1$. Following \cite{AF20},  we have
\eq{\alg{\notag
\| \p_T N \|_{H^k} 
&\lesssim 
	\| \hN\|_{H^k}+\| X\|_{H^k} + \|\Sigma \|_{H^{k-1}}^2 + \| g-\gamma\|_{H^k}^2
	+ |\tau| \big(\| S \|_{H^{k-2}} + \| \eta \|_{H^{k-2}}+\| \p_T \eta \|_{H^{k-2}}\big)
\\&\lesssim
	\| \hN\|_{H^k}+\| X\|_{H^k} 
+ |\tau|\| \rho \|_{H^{k-2}}+ \varepsilon^2 e^{\max\{-2\lambda, -2+\mu\}T},
}}
where we used Lemma \ref{lem:matter-first-estimates} and Lemma \ref{lem:Est-pT-eta-S}. 

Again, using a general expression given in  \cite[\textsection 7]{AF20}, we have
\eq{\alg{\notag
\| \p_T X \|_{H^k} 
&\lesssim 
	 \| \hN\|_{H^k}+ \| X\|_{H^{k}}+ \|\Sigma \|_{H^{k}}^2 +\|g-\gamma\|_{H^k} ^2
\\
&\quad	+ \tau \big(\| S \|_{H^{k-2}} + \| \eta \|_{H^{k-2}}+\| \p_T \eta \|_{H^{k-2}}\big)+ \tau^2 \| \jmath \|_{H^{k-1}} + \tau^3 \| \underline{T} \|_{H^{k-1}}
\\
&\lesssim
\| \hN\|_{H^k}+ \| X\|_{H^{k}}
+ |\tau| \| \rho\|_{H^{k-2}} 
+\varepsilon^2 e^{\max\{-2\lambda, -2+\mu\}T},
}}
where we again used Lemma \ref{lem:matter-first-estimates} and Lemma \ref{lem:Est-pT-eta-S}.
\end{proof}

\subsection{Additional Lapse and Shift Estimates}
\label{sec:ExtraLapseShiftEst}
In this section we establish important auxiliary estimates in two propositions for the lapse and shift variables. The first proposition is used to  estimate $\delu \nabla_i \nabla_j N$. This somewhat unusual expression comes from a dust derivative acting on the first term in $F_v$ (see Definition \ref{defn:Fh-Fv}), and appears later on in the energy estimates for $\Etot$. 

\begin{prop}[Top-order auxiliary lapse estimate]
\label{prop:delu-commuted-lapse}
We have
\eq{\notag
\big(\Delta-\tfrac13\big)\delu \nabla_i \nabla_j N 
=\mathcal{F}_{N,\mathbf{u}}}
where, 
\eq{\notag
\|\mathcal{F}_{N,\mathbf{u}}\|_{H^{N-4}} \lesssim 
\Lambda(T) + \Edelu+ E^g_{N-1}.
}
\end{prop}

\begin{proof}
Start by commuting the lapse equation \eqref{eq:EoM-lapse-shift} with $\delu \nabla_i \nabla_j$:
\eq{\alg{\notag
\big(\Delta-\tfrac13\big)\delu \nabla_i \nabla_j N 
&= [\Delta,\delu]\nabla_i\nabla_j N
+ \delu [\Delta, \nabla_i\nabla_j]N + \delu \nabla_i \nabla_j \Big( N(|\Sigma|_g^2 - \tau \eta)\Big)
\\
&=: L_1 + L_2 +L_3=: \mathcal{F}_{N,\mathbf{u}}.
}}
We investigate each of the terms in $\mathcal{F}_{N,\mathbf{u}}$ separately. The first term $L_1 $ we estimate using \eqref{eq:est-pt-Gamma}, \eqref{eq:Commutator-pT-nablas} and Lemma \ref{lem:comm-delubnab-Delta} 
\eq{\alg{\notag
\|L_1\|_{H^{N-4}}
&\lesssim
\Lambda(T)\|\p_T \nabla_i\nabla_j N\|_{H^{N-3}}
+ \varepsilon e^{\max\{-1+\mu,-\lambda\} T}\|\nabla_i\nabla_j N\|_{H^{N-2}}
\\&\lesssim
\Lambda(T)\Big( \|\p_T N\|_{H^{N-1}} +  \|[\p_T, \nabla_i]\nabla_j N\|_{H^{N-3}} \Big)
+ \varepsilon e^{\max\{-1+\mu,-\lambda\} T}\|\hN\|_{H^{N}}
\\&\lesssim
\Lambda(T)^2
+ \varepsilon e^{\max\{-1+\mu,-\lambda\} T}\Lambda(T).
}}

For the next term, $L_2$, we first compute, for $\phi$ a scalar,
\eq{\alg{\notag
[\Delta, \nabla_i] \phi &= 2g^{ab}\nabla_{[a}\nabla_{i]}\nabla_b \phi = -\Ric^c{}_i \nabla_c \phi,
\\
[\Delta,\nabla_i]\nabla_j\phi&=
\Ric^c{}_i\nabla_c \nabla_j\phi - (\nabla^a \Riem^c{}_{jai})\nabla_c\phi - 2\Riem^c{}_{jai}\nabla^a\nabla_c\phi.
}}
Thus
\eq{\alg{\notag
L_2 &=\delu([\Delta,\nabla_i]\nabla_j N) + \delu \nabla_i([\Delta,\nabla_j N)
\\&=
-\Big( (\delu \nabla^a \Riem^c{}_{jai})+(\delu \nabla_i \Ric^c{}_j) \Big)\nabla_c N 
- \Big(\nabla^a \Riem^c{}_{jai}+\nabla_i \Ric^c{}_j \Big) \delu \nabla_c N
\\&\quad
+\Big(
(\delu \Ric^c{}_i)\nabla_c\nabla_j N 
-(\delu \Ric^c{}_{j} )\nabla_i\nabla_c N 
- 2 (\delu \Riem^c{}_{jai})\nabla^a\nabla_c N
\Big)
\\&\quad
+\Big( \Ric^c{}_i(\delu \nabla_c\nabla_j N )
- \Ric^c{}_{j} (\delu\nabla_i\nabla_c N )
- 2\Riem^c{}_{jai}(\delu \nabla^a\nabla_c N)
\Big)
\\
&=: -(L_{21})-(L_{22})+(L_{23})+(L_{24})
}}
The second and last terms here are easy to estimate using Lemma \ref{lem:delu-N-X} and Lemma \ref{lem:Commutators-with-Delu} (note also that the expressions \eqref{eq:pt-nabla} and \eqref{eq:Commutators-with-Delu} simplify when calculated for a scalar)
\eq{\alg{\notag
\| &|L_{22}|+| L_{24}|\|_{H^{N-4}}
\\
&\lesssim
\| \Riem\|_{H^{N-3}}\Big( \|\delu N\|_{H^{N-4}}
+  \|[\delu, \nabla]N\|_{H^{N-3}} +  \|[\delu, \nabla]\nabla N\|_{H^{N-4}}\Big)
\\
&\lesssim
\| \Riem\|_{H^{N-3}} \Big( \Lambda(T)
+  \Lambda(T) \|\p_T N\|_{H^{N-3}} 
+ \varepsilon e^{\max\{-1+\mu,-\lambda\} T}\|N\|_{H^{N-2}}  \Big)
 \\
&\lesssim
\Lambda(T).
}}

For the other two terms, $L_{21}$ and $L_{23}$, we schematically write
\eq{\alg{\notag
\delu  \Riem 
&= u^0 \p_T (\p \Gamma + \Gamma \Gamma) - \tau u^c \bnab_c \Riem
= u^0 (\p \p_T \Gamma + \Gamma \p_T \Gamma) - \tau u^c \nabla_c \Riem+ \tau u^c \ast \Upsilon\ast  \Riem,
\\
\delu \nabla \Riem 
&= \nabla \delu \Riem + [\delu, \nabla]\Riem,
}}
and thus 
\eq{\alg{\notag
\| |L_{21}|+ |L_{23}|\|_{H^{N-4}}
&\lesssim
\|\hN \|_{H^{N-2}} \Big(\| \p_T \Gamma\|_{H^{N-2}}
+ \| \p_T \Gamma\|_{H^{N-3}}\| \Gamma\|_{H^{N-3}}
+ \| \p_T \Gamma\|_{H^{N-4}}\| \Riem\|_{H^{N-4}}
\\&\quad
+  \varepsilon e^{\max\{-1+\mu,-\lambda\} T}\| \Riem\|_{H^{N-2}}
+ \varepsilon \| \p_T \Riem \|_{H^{N-4}}
\Big)
\\&\lesssim
\varepsilon e^{-\lambda T} \|\hN \|_{H^{N-2}}
+ \varepsilon e^{\max\{-1+\mu,-\lambda\} T}\|\hN \|_{H^{N-2}} \| \Riem\|_{H^{N-2}}
\lesssim \Lambda(T).
}}

Finally we compute
\eq{\alg{\notag
L_3 &= \nabla_i\nabla_j\delu \big( N(|\Sigma|_g^2 - \tau \eta)\big) + \nabla_i [\delu,\nabla_j]\big( N(|\Sigma|_g^2 - \tau \eta)\big)+ [\delu,\nabla_i]\nabla_j \big( N(|\Sigma|_g^2 - \tau \eta)\big)
\\&=: L_{31}+L_{32}+L_{33}
}}
By the matter estimates in Lemma \ref{lem:matter-first-estimates} and Lemma \ref{lem:Est-delu-eta-S}, as well as the geometry estimates in Lemma \ref{lem:delu-N-X} and  Corollary \ref{corol:delu-nab-estimates}, we find
\eq{\alg{\notag
\| L_{31}\|_{H^{N-4}}
&\lesssim
\| \delu N\|_{H^{N-2}} \big( \|\Sigma \|_{H^{N-2}}^2+|\tau| \| \eta\|_{H^{N-2}}\big)
+ \| \Sigma\|_{H^{N-2}}\|\delu \Sigma\|_{H^{N-2}}
\\&\quad + |\tau| \| \eta\|_{H^{N-2}}+ |\tau|  \|\delu \eta\|_{H^{N-2}}
\\&\lesssim
\Lambda(T) + \| \Sigma\|_{H^{N-2}}\big( (E^g_{\delu,N-1})^{1/2} + (E^g_{N-1})^{1/2}\big).
}}
Similarly by  Lemma \ref{lem:matter-first-estimates},  Lemma \ref{lem:Est-pT-eta-S}, Corollary \ref{corol:Estimates-delu-h-v} and the commutator estimate of Lemma \ref{lem:Commutators-with-Delu},
\eq{\alg{\notag
\| L_{32}\|_{H^{N-4}}
&\lesssim
\varepsilon e^{\max\{-1+\mu,-\lambda\} T} \big( \|\Sigma\|_{H^{N-2}} ^2 + |\tau|\|\eta\|_{H^{N-2}} \big)
+\Lambda(T) \|\p_T N\|_{H^{N-3}} \big( \|\Sigma\|_{H^{N-3}}^2+ |\tau| \|\eta\|_{H^{N-3}}\big) 
\\&\quad
+ \Lambda(T) \| \Sigma\|_{H^{N-3}}\|\p_T \Sigma \|_{H^{N-3}} 
+ \Lambda(T) |\tau| \|\eta\|_{H^{N-3}} +\Lambda(T) |\tau|\| \p_T \eta \|_{H^{N-3}}
\\&\lesssim
\varepsilon e^{\max\{-1+\mu,-\lambda\} T} E^g_{N-1} + \varepsilon e^{\max\{-1+\mu,-\lambda\} T} \Lambda(T).
}}
The final estimate follows in the same way:
\eq{\alg{\notag
\| L_{33}\|_{H^{N-4}}&\lesssim
\varepsilon e^{\max\{-1+\mu,-\lambda\} T} E^g_{N-1} + \varepsilon e^{\max\{-1+\mu,-\lambda\} T} \Lambda(T).
}}

Putting this all together we have found
\eq{\notag
\|\mathcal{F}_{N,\mathbf{u}}\|_{H^{N-4}} \lesssim 
\Lambda(T) + \| \Sigma\|_{H^{N-2}}\big( (E^g_{\delu,N-1})^{1/2} + (E^g_{N-1})^{1/2}\big),
}
and so the conclusion follows by Lemma \ref{lem:Geom-coercive-est}.
\end{proof}

\begin{rem}\label{rem:566}
 In our energy estimates later on we need to estimate a term of the type $\delu \bnab^{N-2} F_v$ where $\bnab^{N-2}$ indicates $N-2$ covariant derivatives and $F_v$ is given in Definition \ref{defn:Fh-Fv}. We would like to commute the $\delu$ operator past these covariant derivatives. However, $F_v$ contains a $\nabla\nabla N$ term, and we cannot commute $\delu$ past $N-2$ copies of $\nabla$ as well as the extra two derivatives in $\nabla \nabla N$ since we only control $u^a$ in $H^{N-1}$.
 The previous proposition crucially allows us to avoid this issue.  
Note also that by commuting in the $\delu$ operator, instead of doing the rough expansion $\delu \sim \p_T + \tau u^c \ast \nabla$, we gain an additional derivative in Corollary \ref{corol:delu-Fv} compared to Corollary \ref{corol:Est-delu-Fh-Fv}. 
\end{rem}

\begin{corol}
\label{corol:delu-Fv}
For $I$ a multi-index of order $|I|\leq N-2$, 
\eq{\alg{\notag
\| \delu \bnab^I F_v\|_{L^2}
&\lesssim 
\Lambda(T) + \Edelu
+E^g_{N-1}.
}}
\end{corol}

\begin{proof}
Write $\bnab^{I} = \bnab_{a_1} \ldots \bnab_{a_{k}}$ where $k\define|I|\leq N-2$.
Using the definition of $F_v$ in Definition \ref{defn:Fh-Fv} and the estimates in Corollary \ref{corol:delu-nab-estimates}, we see the most subtle terms (from the point of view of regularity) are $\nabla_i\nabla_j N$ and $N\tau S_{ij}$. For these terms we look at what happens when we commute in the $\delu$ operator using Lemma \ref{lem:Commutators-with-Delu}:
\eq{\alg{\notag
\| [\delu, \bnab^I] \nabla_i\nabla_j N \|_{L^2}
&\lesssim 
\|[ \delu, \bnab_{a_1}] \bnab_{a_2} \ldots \bnab_{a_k} \nabla_i\nabla_j N \|_{L^2}  
+ \cdots
+ \|[ \delu, \bnab] \nabla_i\nabla_j N \|_{H^{k-1}} 
\\&
\lesssim
\varepsilon e^{(-1+\mu) T}  \|\hN\|_{H^{k+2}}  +\Lambda(T)\|\p_T  N \|_{H^{k+1}} 
+\Lambda(T)\|[\p_T ,\nabla_i\nabla_j] N \|_{H^{k-1}} .
}}
So by the commutator estimates \eqref{eq:pt-nabla} and \eqref{eq:est-pt-Gamma}, as well as Proposition \ref{prop:delu-commuted-lapse},
\eq{\alg{\notag
\| \delu \bnab^I \nabla_i\nabla_j N \|_{L^2}
&\lesssim 
\| \delu \nabla_i\nabla_j N \|_{H^k} 
+ \| [\delu, \bnab^I] \nabla_i\nabla_j N \|_{L^2}
\\&
\lesssim
\| \mathcal{F}_{N, \mathbf{u}} \|_{H^{k-2}} + \Lambda(T)
\\&\lesssim
\Lambda(T) + \| \Sigma\|_{H^{N-2}}\big( (E^g_{\delu,N-1})^{1/2} + (E^g_{N-1})^{1/2}\big).
}}

Similarly, using the matter estimates in Lemma \ref{lem:matter-first-estimates}, Lemma \ref{lem:Est-delu-eta-S} and Lemma \ref{lem:Est-pT-eta-S}, as well as $\delu N$ estimates in Lemma \ref{lem:delu-N-X}, we find
 \eq{\alg{\notag
\|\delu\bnab^I (\tau N S_{ij})\|_{\Lgg}
&\lesssim 
|\tau| \big(\|S\|_{H^{k}}  +  \| \delu S\|_{H^{k}} + \| \delu N \|_{H^k} \| S\|_{H^k} + \varepsilon e^{(-1+\mu) T}  \|S\|_{H^{k}} 
\\&\quad \qquad 
+\Lambda(T)\|\p_TS\|_{H^{k-1}} 
+ \Lambda(T)\|S\|_{H^{k-1}}
+\Lambda(T)\|\p_T N\|_{H^{k-1}}\|S\|_{H^{k-1}}
\big)
\\&
\lesssim 
\Lambda(T).
}}
\end{proof}

In our energy estimates later on we need to estimate a term of the type $\delu \bnab^{N-1} F_h$. Similar to the issue discussed in Remark \ref{rem:566}, the problematic term here is the $\mcr{L}_X g$ term appearing in $F_h$ (see Definition \ref{defn:Fh-Fv}). 
The second proposition of this section estimates this problematic term. Since, however, the shift is not scalar-valued like the lapse, a replication of the ideas used in Proposition \ref{prop:delu-commuted-lapse} ends up failing. Instead, our proof involves a remarkable combination of commutator estimates, the Bianchi identity and the Einstein equations in the CMCSH gauge. 

\begin{prop}[Top-order estimate for $\mcr{L}_X g$]\label{prop:Lie-deriv-Shift}
We have
\eq{\alg{\notag
\Big| \big(\delu\mL^{\ell-1}(\Delta \mcr{L}_X g),\delu\mL^{\ell}(h)\big)_{L^2_{g,\gamma}} \Big| &\lesssim 
\Lambda(T) \Edelu + \Lambda(T) E^g_{N-1}
+ \Lambda(T) (\Edelu)^{1/2}.
}}
\end{prop}

\begin{proof}
Recall $2(\ell-1) = N-3$. 
We first define $ \mathcal{F}^X_a$ by rewriting the shift equation \eqref{eq:EoM-lapse-shift} as 
\eq{\notag
\Delta X_a = - \Ric^c{}_a X_c + \mathcal{F}^X_a.
}
By contracting the Bianchi identity, one finds
\eq{\notag
\nabla^a\Riem_{abcd} = \nabla_c \Ric_{bd} - \nabla_d \Ric_{bc}.
}
Using this, and the fact that $\nabla$ is a torsion-free connection for the metric $g$, we can show that
\eq{\alg{\notag
\Delta (\mcr{L}_X g)_{ab}
&= \nabla_a \Delta X_b + \nabla_b \Delta X_a + [\Delta,\nabla_a]X_b + [\Delta,\nabla_b]X_a
\\
&= \nabla_a(-\Ric^c{}_b) X_c+ \nabla_a\mathcal{F}^X_b + \nabla_b (-\Ric^c{}_a) X_c+ \nabla_b \mathcal{F}^X_a
\\&\quad
-g^{ij} \nabla_i(\Riem^k{}_{bja})X_k
-g^{ij} \nabla_i(\Riem^k{}_{ajb})X_k
+ 2 E_{(ab)}
\\
&= -2 X^k \nabla_k \Ric_{ab}
 + \nabla_a \mathcal{F}^X_b + \nabla_b \mathcal{F}^X_a + 2 E_{(ab)},
}}
where we have introduced the error terms
\eq{\alg{\notag
E_{ij}
\define &
- \Ric^c{}_j\nabla_i X_c 
- 2 \Riem^k{}_{jci}\nabla^c X_k + \Ric_{ki}\nabla^k X_j .
}}
As first noted in \cite{AM03}, in the CMCSH gauge, we have
\eq{\notag
\Ric_{ab} = -\frac29 g_{ab} -\frac12 \mL h_{ab} + J_{ab},
}
and thus
\eq{\alg{\label{eq:mcrLX-101}
\Delta \mcr{L}_X g_{ab}
= X^k \nabla_k \mL h_{ab}
 -2 X^k \nabla_k J_{ab}
 + 2 \nabla_{(a} \mathcal{F}^X_{b)} + 2 E_{(ab)}.
}}

We now need to estimate each term in the RHS of \eqref{eq:mcrLX-101}. The first term requires the most care. We begin with an identity, valid for $V$ an arbitrary $(0,2)$-tensor and $k \in \mathbb{Z}_{\geq 0}$, \eq{\alg{\label{eq:Gi-101}
\delu\mL^{k}(X^m \bnab_m V)
&= 
X^m\bnab_m \delu\mL^{k} (V)
+ X^m\delu[\mL^k,\bnab_m]V
+ X^m[\delu, \bnab_m]\mL^k V 
\\&\quad 
+ \delu X^m\cdot \mL^{k} (\bnab_m V)
+ \delu [\mL^k, X^m] \bnab_m V.
}}
Thus
\eq{\label{eq:Gi-111}
\delu\mL^{\ell-1}(X^k \nabla_k \mL h_{ij})
= X^m\bnab_m \delu\mL^{\ell}( h_{ij} ) + \mathcal{R}_{ij}^1 + \mathcal{R}_{ij}^2,
}
where
\eq{\alg{\notag
\mathcal{R}_{ij}^1
&\define
\delu \mL^{\ell-1}(X^m \Upsilon^k_{mi} \mL h_{kj})
+ \delu \mL^{\ell-1}(X^m \Upsilon^k_{mj} \mL h_{ki})
+ X^m\delu[\mL^{\ell-1},\bnab_m]\mL h_{ij}
\\&\quad
+ \delu([\mL^{\ell-1}, X^m] \bnab_m \mL h_{ij} ) ,
\\
\mathcal{R}_{ij}^2
&\define 
X^m[\delu, \bnab_m]\mL^{\ell} h_{ij}
+ \delu X^m\cdot \mL^{\ell-1} (\bnab_m \mL h_{ij}).
}}
We integrate by parts on the first term in \eqref{eq:Gi-111} using \eqref{eq:thirdIBP-2} and Lemma \ref{lem:Properties-of-Egdelu} to find 
\eq{\alg{\notag
\left|\big(X^m\bnab_m \delu\mL^{\ell}( h),\delu\mL^{\ell}(h)\big)_{L^2_{g,\gamma}} \right|
&\lesssim
\| X \|_{H^3} \| \delu \mL^\ell h \|_{\Lgg}^2
\lesssim
\Lambda(T) \Edelu.
}}
The remaining terms $\mathcal{R}_{ij}^{1,2}$ appearing in \eqref{eq:Gi-111}  are errors terms. Since we work at high regularity, we can control such error terms using the basic idea of taking low-derivative terms out in $L^\infty$.  We briefly present this argument once  for the last term in $\mathcal{R}_{ij}^{1}$:
\eq{\alg{\label{eq:BasicArgument}
\| \delu([\mL^{\ell-1}, X^m] \bnab_m \mL h_{ij})\|_{\Lgg}
&\lesssim
\sum_{\substack{|I|+|J|\leq 2(\ell-1) \\ |I|\geq 1}} \|\delu (\bnab^I X \ast \bnab^J \bnab \mL h )\|_{\Lgg}  
\\
&\lesssim
\sum_{\substack{|I|+|J|\leq N-3, \\ |I|\geq 1, |J| \geq 3}} \|\delu (\bnab^I X) \bnab^J  h \|_{\Lgg} + \|\bnab^I X \delu ( \bnab^J  h )\|_{\Lgg}    
}}
We are now faced with four terms depending on where the derivatives sit. Two of these have high derivatives on the shift $X$, and so by Sobolev embedding these are controlled by:
\eq{\alg{\notag
\sum_{\substack{1 \leq |I|\leq N-3, \\ 3 \leq |J| \leq (N-3)/2}} 
\big(\|\delu \bnab^I X  \|_{L^2} \| \bnab^J  h \|_{H^2}
+ \|\bnab^I X \|_{L^2} \| \delu \bnab^J  h \|_{H^2} \big).
}}
We use Lemma \ref{lem:deluV-H2-H1} to exchange the $H^2$ norm with the $\delu$ derivative, and then all terms are then controlled using Corollary \ref{corol:delu-nab-estimates} and Lemma \ref{lem:Estimates-pT-h-v} provided $\frac{N-3}{2}+2\leq N-1$. The other two terms, where the high derivatives hit the metric, are similarly controlled by:
\eq{\alg{\notag
\sum_{\substack{3 \leq |J|\leq N-1, \\ 1 \leq |I| \leq (N-3)/2}} 
\big(\|\delu \bnab^I X  \|_{H^2} \| \bnab^J  h \|_{L^2}
+ \|\bnab^I X \|_{H^2} \| \delu \bnab^J  h \|_{L^2}    \big)
}}
and once again all these terms are  controlled using Lemma \ref{lem:Estimates-pT-h-v}, Corollary \ref{corol:delu-nab-estimates} and  Lemma \ref{lem:deluV-H2-H1}  provided $\frac{N-3}{2}+2\leq N-1$.

Carrying on this way, and using the commutator estimates contained in Lemma \ref{lem:Fluid-Commutators}, together with Corollary \ref{corol:delu-nab-estimates}, we find
\eq{\alg{\notag
\Big|\big(\mathcal{R}^1,\delu\mL^{\ell}(h)\big)_{L^2_{g,\gamma}} \Big|
\lesssim 
\Lambda(T) \Edelu + \Lambda(T) E^g_{N-1}
+ \Lambda(T)^2 (\Edelu)^{1/2}.
}}
For the other error term, $\mathcal{R}^2$, we use Lemma \ref{lem:Estimates-pT-h-v}, Lemma \ref{lem:Fluid-Commutators}, Corollary \ref{corol:Coercive-top-order-geom} and Corollary \ref{corol:delu-nab-estimates}, to find
\eq{\alg{\notag
\Big|\big(\mathcal{R}^2,\delu\mL^{\ell}(h)\big)_{L^2_{g,\gamma}} \Big|
&\lesssim
\| X\|_{H^2} \|\delu\mL^{\ell}(h)\|_{\Lgg} \big( \varepsilon e^{(-1+\mu) T}  \|g-\gamma\|_{H^{N}}  +\Lambda(T) \|\p_T h \|_{H^{N-1}} \big)
\\&\quad
+
\| \delu X\|_{H^2} \| g-\gamma\|_{H^{N}} \| \delu \mL^\ell h \|_{\Lgg}
\\&
\lesssim  \Lambda(T)\Big( (\Edelu)^{1/2}
+ (E^g_{N-1})^{1/2} + \Lambda(T)\Big) (\Edelu)^{1/2}.
}}

The second term in the RHS of  \eqref{eq:mcrLX-101} is estimated by the same expression as for $\mathcal{R}^1$. 
For the third term in the RHS of  \eqref{eq:mcrLX-101}, we use the definition of $\mathcal{F}^X$ given in \eqref{eq:EoM-lapse-shift}:
\eq{\alg{ \notag
\delu\mL^{\ell-1}(\nabla_j \mathcal{F}^X_i)
&= 
\delu\mL^{\ell-1}\nabla_j \left(2 \bnab_c N \Sigma^c_i - \bnab_i \hN+ 2 N \tau^2 g_{ai}\jmath^a
- g_{ai}(2N \Sigma^{bc} - \nabla^b X^c)\Upsilon^a_{bc}  \right).
}}
The second term in the large brackets here ($\bnab_i \hN$) decays the slowest, while from a regularity point of view the most subtle term is the matter term ($\jmath^a$). For this latter term, we commute in the $\delu$ operator and use the matter estimates of Lemma \ref{lem:matter-first-estimates}, Lemma  \ref{lem:Est-delu-eta-S} and Lemma \ref{lem:Est-pT-eta-S} together with the commutator estimates of Lemma \ref{lem:Commutators-with-Delu} and Corollary \ref{corol:Com-delu-mL-estimate}:
\eq{\alg{\notag
\|\delu\mL^{\ell-1}\nabla_j (\tau^2 g_{ai}\jmath^a)\|_{L^2}
&\lesssim 
\tau^2 \big(\| \delu \jmath^a\|_{H^{N-2}}  
+ \| [\delu, \mL^{\ell-1}]\bnab_j \jmath^a \|_{L^2}
+\| [\delu,\nabla_j] \jmath^a\|_{H^{N-3}}\big)
\\
&\lesssim
\tau^2 \big(\| \delu \jmath^a\|_{H^{N-2}}
+
\varepsilon e^{\max\{-1+\mu,-\lambda\} T} \| \jmath\|_{H^{N-2}} 
+\Lambda(T)\|\p_T \jmath\|_{H^{N-3}}\big)
\\
&\lesssim
\varepsilon^2 e^{(-\lambda+\mu)T} \Lambda(T) + \Lambda(T)^2 + \varepsilon^4 e^{-(3+\delta)T}. 
}}
All together, and using Lemma \ref{lem:Properties-of-Egdelu}, we find
\eq{\alg{\notag
\Big|\big(2 \delu\mL^{\ell-1}\nabla_{(a} \mathcal{F}^X_{b)},\delu\mL^{\ell}(h_{ab})\big)_{L^2_{g,\gamma}} \Big|
&
\lesssim 
\Lambda(T) \Edelu + \Lambda(T) E^g_{N-1}
\\&\quad
+ \Big(\Lambda(T)^2 + \varepsilon^2 e^{(-\lambda+\mu)T} \Lambda(T)  + \varepsilon^4 e^{-(3+\delta)T} \Big) (\Edelu)^{1/2}.
}}

Finally, we have
\eq{\alg{\notag 
\| \delu\mL^{\ell-1} E\|_{L^2} 
&\lesssim  \sum_{\substack{|I|+|J|\leq N-2 \\ |J|\geq 1}} \|\delu (\bnab^I \ri \ast \bnab^J X )\|_{\Lgg} 
\lesssim \Lambda(T),
}}
where the Riemann terms can be estimated using arguments as in the proof of Proposition \ref{prop:delu-commuted-lapse}.
\end{proof}

\section{Geometric Energy Estimates}
\label{sec:GeomEnergyEst}
In this section we establish energy estimates for the geometric energy functionals. We first prove estimates for the time-derivatives of the lower-order energy functional $\Eg_{N-1}$, and then for the top-order energy functional $\mathcal{E}^g_{\delu,N-1}$. Recall $\alpha, c_E$ and $\Etot$ are given in Definitions \ref{defn:GeometricEnergy} and \ref{rem:ell-N}.

\begin{prop}
\label{prop-lower-order-geometry-est}
There exists a constant $C>0$ such that
\eq{ \notag
\p_T \Eg_{N-1}
\leq	
- 2 \alpha \Eg_{N-1} 
+ C \Lambda(T) (\Eg_{N-1})^{1/2}  
+ C(\Eg_{N-1})^{3/2}.
} 
\end{prop}

\begin{proof}
Let $0 \leq k \leq N-1$. 
Using an estimate from \cite[Lemma 20]{AF20} together with Lemma \ref{lem:matter-first-estimates}, we find
\eq{\alg{ \notag
\p_T \Eg_k 
&\leq	
	- 2 \alpha \Eg_k + 6 (\Eg_k)^{1/2} |\tau| \| N S \|_{H^{k-1}} + C(\Eg_k)^{3/2}
	+ C(\Eg_k)^{1/2} \big( 
	|\tau| \| \eta \|_{H^{k-1}} 
	+ \tau^2 \| N j \|_{H^{k-2}}
	\big)
\\&\leq
	- 2 \alpha \Eg_k +  C(\Eg_k)^{1/2} |\tau| \| \rho \|_{H^{k-1}} 
+ C(\Eg_k)^{3/2}
+ C \Lambda(T) (\Eg_k)^{1/2} .
}}
\end{proof}

We now turn to the time evolution of the  top-order \emph{$\delu-$boosted} geometric energy.  The main result is stated in Theorem \ref{thm-top-order-geom}. For ease of presentation however, the proof relies on the subsequent estimates given in Propositions \ref{prop:G1}, \ref{prop:G2}, \ref{prop:G3}, \ref{prop:G4}, and \ref{prop:Gc}. 

\begin{rem}
The key auxiliary estimates of the previous section, Corollary \ref{corol:delu-Fv} and Proposition \ref{prop:Lie-deriv-Shift}, are applied in Proposition \ref{prop:G2} and Proposition \ref{prop:G1} respectively. 
The important integration by parts identity \eqref{eq:symIBP} in Lemma \ref{lem:IBP}, where one term is brought back onto the LHS, is applied in Proposition \ref{prop:G4} (see \eqref{eq:IBP-later}).
\end{rem} 

\begin{thm}
\label{thm-top-order-geom}
There exists a constant $C>0$ such that
\eq{\alg{\notag
\p_T \mathcal{E}^g_{\delu,N-1}
&\leq 
-2\alpha \mathcal{E}^g_{\delu,N-1}
+ C\big( |\tau| \| u\|_{H^{N-1}}+ \Lambda(T) \big)\Etot
+ C\Lambda(T) (\Etot)^{1/2}
\\&\quad
+ C(\Etot)^{3/2}
+ C\Lambda(T)^2.
}}
\end{thm}

\begin{proof}
Recalling Definition \ref{defn:delu-geom-energy}, the energy $\Edelu$ consists of a sum over lower-order energies. The top-order is the most subtle, so we focus on this and merely remark that the estimates for the lower-orders follow in the same (or possibly easier) way. The top-order energy consists of two parts, an $\mathcal{E}^g_{\delu,N-1}$ term and a $c_E \Gamma^g_{\delu,N-1}$ term. We treat these separately for the moment. 

Using \eqref{defn:delu}-\eqref{eq:hdelIntegral} and integration by parts, we find 
\eq{\alg{\label{eq:3897}
\p_T \mathcal{E}^g_{\delu,2\ell}(T)
&\leq \big( \| \hN\|_{L^\infty}+ \| \nabla X\|_{L^\infty}\big)\mathcal{E}^g_{\delu,2\ell}(T)
+\Big[9\big(\delu\mL^{\ell}(\p_T h),\delu\mL^{\ell}(h)\big)_{L^2_{g,\gamma}}
\\&\quad 
+ \tfrac12\big( \delu \mL^{\ell} \p_T v ,\delu\mL^{\ell-1} v \big)_{L^2_{g,\gamma}}
+ \tfrac12\big( \delu \mL^{\ell} v ,\delu\mL^{\ell-1} \p_T v \big)_{L^2_{g,\gamma}} \Big]
+ G_c,
}}
where we define
\eq{\alg{\notag
G_{c} &\define 9\big([\delu\mL^{\ell}, \p_T](h),\delu\mL^{\ell}(h)\big)_{L^2_{g,\gamma}}
\\&\quad
+\tfrac12\big( \delu \mL^{\ell} (v) ,[\delu\mL^{\ell-1},\p_T] (v)\big)_{L^2_{g,\gamma}}
+ \tfrac12\big([\delu \mL^{\ell},\p_T]( v) ,\delu\mL^{\ell-1} (v) \big)_{L^2_{g,\gamma}}.
}}
The terms $G_c$ are error terms, and so our main focus is on the terms appearing \eqref{eq:3897} in the square bracket. 
Using the equations of motion \eqref{eq:EoM-hv} and self-adjointness of $\mL$, these become
\eq{\alg{\notag
9\big(\delu\mL^{\ell}(\p_T h),\, &\delu\mL^{\ell}(h)\big)_{L^2_{g,\gamma}}
+ \tfrac12\big(\delu \mL^{\ell} (\p_T v) ,\delu\mL^{\ell-1} (v) \big)_{L^2_{g,\gamma}}
+ \tfrac12\big(\delu \mL^{\ell} (v),\delu\mL^{\ell-1} (\p_T v) \big)_{L^2_{g,\gamma}}
\\ &= \notag
-2\big( \delu\mL^{\ell}(v),\delu\mL^{\ell-1}(v)\big)_{L^2_{g,\gamma}}
+ G_1 + G_2 + G_3 + G_4,
}}
where
\begin{align*}
G_1 &\define 9\big(\delu\mL^{\ell}(2\hN g - \mcr{L}_X g),\delu\mL^{\ell}(h)\big)_{L^2_{g,\gamma}},
\\
G_2 &\define 3\big(\delu \mL^\ell (F_v),\delu\mL^{\ell-1}(v)\big)_{L^2_{g,\gamma}}
+3 \big(\delu \mL^\ell (v),\delu\mL^{\ell-1}(F_v)\big)_{L^2_{g,\gamma}},
\\
G_3 &\define 9\Big[ \big(\delu\mL^{\ell}(wv),\delu\mL^{\ell}(h)\big)_{L^2_{g,\gamma}}
\\&\quad
- \tfrac92 \Big[\big(\delu\mL^{\ell}(w\mL h),\delu\mL^{\ell-1}(v)\big)_{L^2_{g,\gamma}}
+\big(\delu\mL^{\ell}(v),\delu\mL^{\ell-1}(w\mL h)\big)_{L^2_{g,\gamma}}\Big],
\\
G_4 &\define \tfrac12 \big(\delu\mL^{\ell}(X^m \bnab_m v),\delu\mL^{\ell-1}(v)\big)_{L^2_{g,\gamma}}
+ \tfrac12 \big(\delu\mL^{\ell}(v),\delu\mL^{\ell-1}(X^m \bnab_m v)\big)_{L^2_{g,\gamma}}.
\end{align*}

Similarly using \eqref{defn:delu}-\eqref{eq:hdelIntegral} and integration by parts, we find
\eq{\alg{\label{eq:3890}
\p_T \Gamma^g_{\delu,2\ell}(T)
&\leq \varepsilon e^{-T} \Gamma^g_{\delu,2\ell}(T)
+\Big[\big(\delu\mL^{\ell-1}(\p_T v),\delu\mL^{\ell}(h)\big)_{L^2_{g,\gamma}}
+ \big( \delu \mL^{\ell-1} v ,\delu\mL^{\ell} \p_T h\big)_{L^2_{g,\gamma}}
\Big]
\\&\quad
+ G_{cc} \,,
}}
where we define
\eq{\alg{\notag
G_{cc} &\define 
\big([\p_T, \delu\mL^{\ell-1}](v),\delu\mL^{\ell}(h)\big)_{L^2_{g,\gamma}}
+\big( \delu \mL^{\ell-1} (v) ,[\p_T,\delu\mL^{\ell}] (h)\big)_{L^2_{g,\gamma}}.
}}
Using the equations of motion \eqref{eq:EoM-hv} 
\eq{\alg{\notag
\Big[\big(&\delu\mL^{\ell-1}(\p_T v),\delu\mL^{\ell}(h)\big)_{L^2_{g,\gamma}}
+ \big( \delu \mL^{\ell-1} v ,\delu\mL^{\ell} \p_T h\big)_{L^2_{g,\gamma}}
\Big]
\\&
= -2\big( \delu\mL^{\ell-1}(v),\delu\mL^{\ell}(h)\big)_{L^2_{g,\gamma}}
- 9\big( \delu\mL^{\ell}(h),\delu\mL^{\ell}(h)\big)_{L^2_{g,\gamma}}
+ \big( \delu\mL^{\ell-1}(v),\delu\mL^{\ell}(v)\big)_{L^2_{g,\gamma}}
\\&\quad
+ G_5 + G_6 + G_7 + G_8 + G_9,
}}
where we have defined
\begin{align*}
G_5 &\define -9\big(\delu\mL^{\ell-1}(\hN \mL h),\delu\mL^{\ell}(h)\big)_{L^2_{g,\gamma}},
\qquad
&G_6 &\define \big( \delu\mL^{\ell-1}(v),\delu\mL^{\ell}(\hN v)\big)_{L^2_{g,\gamma}},
\\
G_7 &\define\big(\delu\mL^{\ell-1}(v),\delu\mL^{\ell}(2\hN g - \mcr{L}_X g)\big)_{L^2_{g,\gamma}},
\qquad
&G_8 &\define - \big(\delu\mL^{\ell-1}(X^m \bnab_m v),\delu\mL^{\ell}(h)\big)_{L^2_{g,\gamma}},
\\
G_9 &\define  6\big(\delu\mL^{\ell-1}(F_v),\delu\mL^{\ell}(h),\big)_{L^2_{g,\gamma}}.
\end{align*}

The errors terms are estimated in the following subsection \ref{sec:est-Gi} and subsection \ref{sec:est-Gc}: $G_1$ and $G_7$ (Proposition \ref{prop:G1}),  $G_2$ and $G_9$ (Proposition \ref{prop:G2}), $G_3$, $G_5$ and $G_6$ (Proposition \ref{prop:G3}), $G_4$ and $G_8$ (Proposition \ref{prop:G4}), $G_c$ and $G_{cc}$ (Proposition \ref{prop:Gc}). 
Using these estimates, and adding \eqref{eq:3897} and $c_E \times$\eqref{eq:3890} together, yields the required inequality. 
\end{proof}

\subsection{Estimates on $\mathbf{G_1}$ to $\mathbf{G_9}$}
\label{sec:est-Gi}

In this section we prove estimates on the terms $G_1$ to $G_9$. 
We begin with a useful identity. Let $\tilde V_{ij},\tilde P_{ij}$ be symmetric $(0,2)$-tensors on $M$ and $k \in\mathbb{Z}_{\geq 1}$. Then, the integration by parts rule \eqref{eq:IBP} implies
\eq{\alg{\label{eq:Gi-102}
\big(\delu&\mL^{k}(\tilde V ),\delu\mL^{k-1}(\tilde P)\big)_{\Lgg}
\\&= 
\big(g^{ab}\bnab_a\delu \mL^{k-1}(\tilde V),\bnab_b\delu\mL^{k-1}(\tilde P)\big)_{\Lgg}
- 2\big(\Riem[\gamma]\circ \delu \mL^{k-1}(\tilde V),\delu\mL^{k-1}(\tilde P)\big)_{\Lgg}
\\&\quad
+\big([\delu,\mL] \mL^{k-1}(\tilde V),\delu\mL^{k-1}(\tilde P)\big)_{\Lgg}.
}}

\begin{prop}\label{prop:G1}
We have,
\eq{\notag
|G_1|+|G_7|\lesssim
\Lambda(T)\Edelu
+ \Lambda(T) E^g_{N-1} 
+ \Lambda(T) (\Edelu)^{1/2}
+ \Lambda(T)^2. 
}
\end{prop}

\begin{proof}
The first term in $G_1$  is easily controlled using Lemma \ref{lem:Properties-of-Egdelu} and Corollary \ref{corol:delu-nab-estimates}:
\eq{\alg{\notag
\big|\big(\delu\mL^{\ell}(2\hN g),\delu\mL^{\ell}(h)\big)_{L^2_{g,\gamma}} \big|
&\lesssim
\sum_{|I|\leq N-1} \|\delu \bnab^I \hN \|_{\Lgg}  \cdot \| \delu \mL^\ell h \|_{\Lgg}
\lesssim \Lambda(T) (\Edelu)^{1/2}.
}}

For the Lie derivative term in $G_1$, we rewrite it as
\eq{\label{eq:700}
\delu\mL^{\ell}(\mcr{L}_X g_{ij})
= \delu\mL^{\ell-1}(\Delta \mcr{L}_X g_{ij} 
+ \delu\mL^{\ell-1}\big( \mL - \Delta\big) \mcr{L}_X g_{ij}).
}
The first term on the RHS of \eqref{eq:700} is precisely what is controlled using Proposition \ref{prop:Lie-deriv-Shift}. 
For the second term, we schematically have
\eq{\alg{\notag
(\mL - \Delta) \mcr{L}_X g
&= g^{-1}(\bnab\bnab - \nabla\nabla)\mcr{L}_X g 
+ \Riem[\gamma] \ast \mcr{L}_X g
\\&
= \nabla \Upsilon\ast \nabla X + \Upsilon \ast \nabla \nabla X + \Upsilon \ast \Upsilon \ast \nabla X + \nabla X.
}}
Using the $\Upsilon$ estimate \eqref{eq:Estimate-Upsilon}, we have
\eq{\alg{\notag
\big|\big(\delu\mL^{\ell-1}((\mL-\Delta)\mcr{L}_X g),\delu\mL^{\ell}(h)\big)_{L^2_{g,\gamma}}\big|
&\lesssim
\sum_{|I|+|J| \leq N-1} \|\delu( \bnab^I X  \bnab^J h) \|_{\Lgg} \cdot \| \delu \mL^\ell h \|_{\Lgg}
\\
&\lesssim
\Lambda(T) \Big( 
(\Etot)^{1/2}+ \Lambda(T)\Big) (\Edelu)^{1/2}
+ \Lambda(T) (\Edelu)^{1/2}.
}}
where in the final estimate we used the same ideas as in \eqref{eq:BasicArgument}, namely an application of Corollary \ref{corol:delu-nab-estimates}, Lemma \ref{lem:deluV-H2-H1} and the fact that $2(\ell-1)=N-3$. 
The same arguments clearly apply to $G_7$ also. 
\end{proof}

\begin{prop}
\label{prop:G2}
We have,
\eq{\alg{\notag
|G_2|+|G_9|\lesssim 
(\Etot)^{3/2}
+\Lambda(T) (\Edelu)^{1/2}
+\Lambda(T)^2.
}}
\end{prop}

\begin{proof}
The two terms in $G_2$ are estimated in virtually the same way. We discuss how to estimate the second term, since it is  slightly more difficult. We first use the identity \eqref{eq:Gi-102} with $\tilde V = \Sigma$ and $\tilde P = F_v$, in order to transfer derivatives between terms:
\eq{\alg{\notag
\big(\delu &\mL^\ell (\Sigma),\delu\mL^{\ell-1}(F_v)\big)_{L^2_{g,\gamma}}
X= \big( g^{ab} \bnab_a \delu\mL^{\ell-1} \Sigma, \bnab_b \delu \mL^{\ell-1}(F_v)\big)_{L^2_{g,\gamma}}
\\&\quad
-2\big( \Riem[\gamma]\circ \delu \mL^{\ell-1}\Sigma, \delu \mL^{\ell-1} F_v \big)_{L^2_{g,\gamma}}
+ \big( [\delu,\mL]\mL^{\ell-1} \Sigma,  \delu \mL^{\ell-1}(F_v)\big)_{L^2_{g,\gamma}}.
}}
We estimate each of these three expressions in turn. For the first expression, we use the various $F_v$ estimates from Lemma \ref{lem:Est-Fh-Fv}, Lemma \ref{lem:Est-pT-Fh-Fv} and Corollary \ref{corol:delu-Fv}, together with the commutator estimates from Lemma \ref{lem:Fluid-Commutators} and Lemma \ref{lem:deluV-H2-H1}, to find
\eq{\alg{\notag
\| \delu \mL^{\ell-1}(F_v) \|_{H^1}
&\lesssim 
\sum_{|I|\leq 1} \| \delu \nabla^I \mL^{\ell-1}(F_v) \|_{L^2}
+ \varepsilon e^{\max\{-1+\mu,-\lambda\} T}\|F_v\|_{H^{N-2}}
\\&\quad
+\Lambda(T) \| \p_T F_v \|_{H^{N-3}}
+ \Lambda(T)\|[\p_T, \mL^{\ell-1}](F_v)\|_{L^2}
\\&\lesssim
\Lambda(T) + \Edelu
+E^g_{N-1}
+ \varepsilon e^{\max\{-1+\mu,-\lambda\} T} \Big( \Lambda(T) + \| g-\gamma\|_{H^N}^2 + \| \Sigma\|_{H^{N-1}}^2\Big)
\\&\lesssim
\Lambda(T) + \Edelu
+E^g_{N-1}.
}}
Note in the final line above we used Corollary \ref{corol:Coercive-top-order-geom}. 

In addition, combining Corollary \ref{corol:delu-mL-Sigma} with  Corollary \ref{corol:Coercive-top-order-geom}, we have
\eq{\alg{\notag
\| \delu\mL^{\ell-1}(\Sigma)\|_{H^1}
&\lesssim 
\| \Sigma\|_{H^{N-1}}+ \| g-\gamma\|_{H^{N}} + \Lambda(T)
\lesssim 
(\Edelu)^{1/2}
+ (E^g_{N-1} )^{1/2}
+\Lambda(T). 
}}
So we find 
\eq{\alg{\notag
\big|\big( g^{ab} \bnab_a \delu\mL^{\ell-1} \Sigma, \bnab_b \delu \mL^{\ell-1}(F_v)\big)_{L^2_{g,\gamma}}\big| 
\lesssim
(\Etot)^{3/2}+\Lambda(T)^2. 
}} 
In a similar way, using Corollary \ref{corol:delu-mL-Sigma} and the coercive lower-order estimate from Lemma \ref{lem:Geom-coercive-est},
\eq{\alg{\notag
\big|\big( \Riem[\gamma]\circ \delu \mL^{\ell-1}\Sigma, \delu \mL^{\ell-1} F_v \big)_{L^2_{g,\gamma}}\big| 
&\lesssim \| \delu \mL^{\ell-1}\Sigma \|_{L^2} \| \delu \mL^{\ell-1}(F_v) \|_{L^2}
\\&
\lesssim
(\Etot)^{3/2}+\Lambda(T)^2. 
}}
Combining \eqref{eq:delu-mL-mL-Sigma} with Corollary \ref{corol:Coercive-top-order-geom} gives
\eq{\alg{\label{eq:delu-mL-mL-Sigma2}
\| [\delu,\mL]\mL^{\ell-1}\Sigma\|_{L^2}
&
\lesssim
\varepsilon e^{\max\{-1+\mu,-\lambda\} T}\big(\| \Sigma\|_{H^{N-1}} + \| g-\gamma\|_{H^N}\big) + \Lambda(T)^2
\\&\lesssim 
\varepsilon e^{\max\{-1+\mu,-\lambda\} T}\Big((\Edelu)^{1/2}
+ (E^g_{N-1} )^{1/2}+ \Lambda(T)\Big). 
}}
Then this gives control over the final term in  $G_2$. 

Finally we use the above estimates and Lemma \ref{lem:Properties-of-Egdelu} to estimate $G_9$ by
\eq{\notag
\Big|\big(\delu\mL^{\ell-1}(F_v),\delu\mL^{\ell}(h),\big)_{L^2_{g,\gamma}}\Big|
\lesssim 
\|\delu \mL^{\ell-1}(F_v) \|_{L^2}\| \delu\mL^{\ell}(h)\|_{L^2}
\lesssim \Lambda(T) (\Edelu)^{1/2}.
}
\end{proof}

\begin{rem}
In the next proposition, it is crucial that the three terms of $G_3$ are treated together, since there is an important cancellation that appears. 
\end{rem}

\begin{prop}\label{prop:G3}
We have,
\eq{\alg{\notag
|G_3|+|G_5|+|G_6|&\lesssim
\big( |\tau| \| u\|_{H^{N-1}}+ \Lambda(T) \big)\Etot
+ \Lambda(T) (\Edelu)^{1/2}
\\&\quad
+(\Etot)^{3/2}+ \Lambda(T)^2.
}}
\end{prop}

\begin{proof}
 First note from \eqref{eq:mL-self-adjoint} that
\eq{\alg{\notag
\big(\delu\mL^{\ell}(w \mL h),\delu\mL^{\ell-1}(v)\big)_{L^2_{g,\gamma}}
&= \big(\delu\mL^{\ell-1}(w \mL h),\delu\mL^{\ell}(v)\big)_{L^2_{g,\gamma}}
+ G_3^e ,
}}
where we have defined 
\eq{\notag
G_3^e \define \big([\delu,\mL]\mL^{\ell-1}(w \mL h),\delu\mL^{\ell-1}(v)\big)_{L^2_{g,\gamma}}
+\big(\delu\mL^{\ell-1}(w \mL h),[\delu,\mL]\mL^{\ell-1}(v)\big)_{L^2_{g,\gamma}}.
} 
A computation then gives,
\eq{\alg{\notag
\tfrac19 G_3 &= 
\big(\delu\mL^{\ell}(wv),\delu\mL^{\ell}(h)\big)_{L^2_{g,\gamma}}
-\big(\delu\mL^{\ell-1}(w \mL h),\delu\mL^{\ell}(v)\big)_{L^2_{g,\gamma}}
-\tfrac12 G_3^e.
}}
The worst term above occurs when all the derivatives hit the geometric variables $v, h$ instead of the lapse variable $w$. However, crucially, such terms cancel
\eq{\notag
\big(w \delu\mL^{\ell}(v),\delu\mL^{\ell}(h)\big)_{L^2_{g,\gamma}}
-\big(w\delu\mL^{\ell-1}( \mL h),\delu\mL^{\ell}(v)\big)_{L^2_{g,\gamma}}
=0.
}
So we are left to consider
\eq{\alg{\label{eq:g3-204}
&\big(\delu\mL^{\ell}(N \Sigma),\delu\mL^{\ell}(h)\big)_{\Lgg}
-\big(\delu\mL^{\ell-1}(N \mL h),\delu\mL^{\ell}(\Sigma)\big)_{\Lgg}
\\&\quad
= \big(\delu\mL^{\ell-1}([\mL,N] \Sigma),\delu\mL^{\ell}(h)\big)_{\Lgg}
-\big(\delu\mL^{\ell}(\Sigma),\delu\mL^{\ell-2}([\mL,N] \mL h)\big)_{\Lgg}.
}}
For the first term on the RHS of \eqref{eq:g3-204},
\eq{\alg{\notag
\Big| \big(\delu\mL^{\ell-1}([\mL,N] \Sigma),\delu\mL^{\ell}(h)\big)_{\Lgg} \Big|
&\lesssim
\sum_{\substack{|I|+|J|\leq 2\ell \\ |I|\geq 1}}\| \delu( \bnab^{I} \hN \cdot \bnab^J \Sigma)\|_{L^2} \|\delu\mL^{\ell}(h)\|_{L^2}
\\&\lesssim
\Lambda(T) \Big( (\Edelu)^{1/2}
+ (E^g_{N-1})^{1/2} + \Lambda(T)\Big) (\Edelu)^{1/2}.
}}

For the second term on the RHS of \eqref{eq:g3-204} we need to apply the integration by parts identity \eqref{eq:IBP}. This gives
\eq{\alg{\notag
\big(&\delu\mL^{\ell}(\Sigma),\delu\mL^{\ell-2}([\mL,N] \mL h)\big)_{\Lgg}
= \big(g^{ab}\bnab_a\delu\mL^{\ell-1}(\Sigma),\bnab_b \delu\mL^{\ell-2}([\mL,N] \mL h)\big)_{\Lgg}
\\&
- 2 \big(\Riem[\gamma]\circ \delu\mL^{\ell-1}(\Sigma),\delu\mL^{\ell-2}([\mL,N] \mL h)\big)_{\Lgg}
+  \big([\delu,\mL]\mL^{\ell-1}(\Sigma),\delu\mL^{\ell-2}([\mL,N] \mL h)\big)_{\Lgg},
}}
and thus
\eq{\alg{\notag
\Big|&\big(\delu\mL^{\ell}(\Sigma),\delu\mL^{\ell-2}([\mL,N] \mL h)\big)_{\Lgg}
\Big|\\
&\lesssim
\Big(\| \delu \mL^{\ell-1} \Sigma\|_{H^1} +  \| [\delu,\mL]\mL^{\ell-1}(\Sigma)\|_{L^2} \Big) \|\delu\mL^{\ell-2}([\mL,N] \mL h)\|_{H^1}
\\
&\lesssim
\Big(\| \delu \mL^{\ell-1} \Sigma\|_{H^1} +  \| [\delu,\mL]\mL^{\ell-1}(\Sigma)\|_{L^2} \Big) \sum_{\substack{|I|+|J|\leq N-1\\|I| \geq 1, |J| \geq 2}} \|\delu\big( \bnab^I \hN \bnab^J h\big)\|_{H^1}.
}}
All these terms can be controlled by estimates in Corollary \ref{corol:delu-mL-Sigma}, Corollary \ref{corol:Coercive-top-order-geom} and \eqref{eq:delu-mL-mL-Sigma2}, together with distributing derivatives and applying Lemma \ref{lem:deluV-H2-H1} and Corollary \ref{corol:Estimates-delu-h-v} as needed. We refer to the example given in \eqref{eq:BasicArgument}.
All together, we find
\eq{\alg{\notag
\Big|\big(&\delu\mL^{\ell}(N \Sigma),\delu\mL^{\ell}(h)\big)_{L^2_{g,\gamma}}
-\big(\delu\mL^{\ell-1}(N \mL h),\delu\mL^{\ell}(\Sigma)\big)_{L^2_{g,\gamma}} \Big|
\\
&\lesssim 
\Lambda(T) \Big( \Edelu
+ E^g_{N-1}+ \Lambda(T)^2\Big).
}}

Finally, we need to estimate the terms  in $G_3^e$. These in fact require the most care. For the first term in $G_3^e$, when using the commutator identity \eqref{eq:Comm-delu-mL-k} on the first factor we see that this has potentially too-many derivatives hitting the metric $h$
\eq{\notag [\delu,\mL]\mL^{\ell-1}(N \mL h)
}
Using the schematic identity given below \eqref{eq:Comm-delu-mL-k}, we see these problematic terms come when all the derivatives hit the metric $h$:
\eq{\alg{\notag
(\delu g^{-1}) \ast \bnab^2 \bnab^{N-1} h +  \bnab u^0 \ast \bnab \bnab^{N-1}\p_T  h 
+\tau  \bnab u^c \ast \bnab^2 \bnab^{N-1} h .
}}
Thus for these terms we integrate by parts using \eqref{eq:thirdIBP-2}. For example, we find
\eq{\alg{\notag
\Big|\Big((\delu g^{ab})\bnab_a \bnab_b \mL^{\ell-1}(N \mL h),\delu\mL^{\ell-1}(\Sigma)\Big)_{L^2_{g,\gamma}}\Big|
&\lesssim
\| \delu g\|_{H^3}  \| \mL^\ell h \|_{H^{1}}\|\delu \mL^{\ell-1}\Sigma\|_{H^1},
\\
\Big|\Big((\tau (\bnab^a u^c) \bnab_a \bnab_c \mL^{\ell-1}(N \mL h),\delu\mL^{\ell-1}(\Sigma)\Big)_{L^2_{g,\gamma}}\Big|
&\lesssim
|\tau| \|u\|_{H^4}  \| \mL^\ell h \|_{H^{1}}\|\delu \mL^{\ell-1}\Sigma\|_{H^1}.
}}
These terms can be controlled using Corollary \ref{corol:Estimates-delu-h-v}, Corollary \ref{corol:delu-mL-Sigma} and Corollary \ref{corol:Coercive-top-order-geom}. 
We eventually obtain the following estimate for the first term in $G_3^e$ 
\eq{\alg{\label{eq:G3-comm-est-455}
\Big| \big([\delu,\mL]\mL^{\ell-1}(N \mL h),\delu\mL^{\ell-1}(\Sigma)\big)_{L^2_{g,\gamma}} \Big|
&\lesssim 
|\tau| \| u\|_{H^{N-1}}\Edelu
+ |\tau| \| u\|_{H^{N-1}}E^g_{N-1}
\\&\quad
+ \Lambda(T) \Edelu  + \Lambda(T) E^g_{N-1}+ \Lambda(T)^2.
}}
For the second term of $G_3^e$, we use the original commutator estimate \eqref{eq:ReturningEq} combined with Corollary \ref{corol:Coercive-top-order-geom} to find 
\eq{\alg{\notag
\Big| &\big(\delu\mL^{\ell-1}(w \mL h),[\delu,\mL]\mL^{\ell-1}(v)\big)_{L^2_{g,\gamma}} \Big|
\lesssim 
\| \delu\mL^{\ell-1}(N \mL h)\|_{L^2}
\| [\delu,\mL]\mL^{\ell-1}(\Sigma)\|_{L^2}
\\&
\lesssim 
(\Edelu)^{1/2}
\Big(\| g-\gamma\|_{H^N}^2 + \| \Sigma\|_{H^{N-1}}^2 + \Lambda(T)^2 + |\tau| \|u\|_{H^2} \|\Sigma\|_{H^{N-1}}^2\Big).
\\&
\lesssim 
(\Edelu)^{1/2}
\Big(\Edelu
+ E^g_{N-1} + \Lambda(T)\Big).
}}

Finally we turn to the estimates for $G_5$ and $G_6$. The first of these is straightforward in light of previous estimates:
\eq{\alg{\notag
|G_5| = \Big|\big(\delu\mL^{\ell-1}(\hN \mL h),\delu\mL^{\ell}(h)\big)_{L^2_{g,\gamma}}\Big|
&\lesssim 
\sum_{\substack{|I|+|J|\leq N-1\\ |J| \geq 2}} \|  \delu(\bnab^{I}\hN  \bnab^J h) \|_{L^2}\| \delu \mL^{\ell}h\|_{L^2} 
\\&\lesssim
\Lambda(T)\Big((\Edelu)^{1/2} + (E^g_{N-1})^{1/2} + \Lambda(T)\Big) 
(\Edelu)^{1/2} .
}}
For $G_6$ we need to integrate by parts once just as for the first term in $G_3^e$, however the estimate follows in the same way and so we omit the details. 
\end{proof}

\begin{prop}\label{prop:G4}
We have,
\eq{\alg{\notag
|G_4|+|G_8| &
\lesssim
 \Lambda(T) \Edelu  + \Lambda(T) E^g_{N-1}+ \Lambda(T)^2.
}}
\end{prop}
\begin{proof}
The two terms appearing in $G_4$ are
\eq{\alg{\notag
\big(\delu\mL^{\ell}(X^m \bnab_m v),\delu\mL^{\ell-1}(v)\big)_{\Lgg}
+\big(\delu\mL^{\ell}(v),\delu\mL^{\ell-1}(X^m \bnab_m v)\big)_{L^2_{g,\gamma}}.
}}
We explain the estimate for the first term only since the second one follows in the same way. 
Using \eqref{eq:Gi-102}, we obtain
\eq{\alg{\label{eq:G4-834}
\big(\delu\mL^{\ell}(X^m\bnab_m \Sigma),\delu\mL^{\ell-1}(\Sigma)\big)_{\Lgg}
&= 
\big(g^{ab}\bnab_a\delu \mL^{\ell-1}(X^m\bnab_m \Sigma),\bnab_b\delu\mL^{\ell-1}(\Sigma)\big)_{\Lgg}
\\&\quad
- 2\big(\Riem[\gamma]\circ \delu \mL^{\ell-1}(X^m\bnab_m \Sigma),\delu\mL^{\ell-1}(\Sigma)\big)_{\Lgg}
\\&\quad
+\big([\delu,\mL] \mL^{\ell-1}(X^a\bnab_a \Sigma),\delu\mL^{\ell-1}(\Sigma)\big)_{\Lgg}.
}}

We focus on the first term on the RHS of \eqref{eq:G4-834} and use  \eqref{eq:Gi-101} to write it as
\eq{\alg{\label{eq:G4-344}
\bnab_a &\delu\mL^{\ell-1}(X^m \bnab_m \Sigma)
= 
X^m\bnab_m \bnab_a \delu\mL^{\ell-1} (\Sigma)
+ G_4^e,
}}
where we have introduced
\eq{\alg{\notag
G_4^e &\define \bnab_a X^m \cdot \bnab_m \delu \mL^{\ell-1} \Sigma
+ X^m [\bnab_a, \bnab_m] \delu \mL^{\ell-1}  \Sigma
\\&\quad
+ \bnab_a \Big[ \delu [\mL^{\ell-1} , X^m] \bnab_m \Sigma + \delu\big( X^m [\mL^{\ell-1} , \bnab_m] \Sigma \big)
+ \delu X^m \cdot \bnab_m \mL^{\ell-1} \Sigma 
+ X^m [\delu, \bnab_m] \mL^{\ell-1} \Sigma \Big].
}}
By a slight adaption of the integration by parts estimate \eqref{eq:symIBP}, we see
\eq{\alg{\label{eq:IBP-later}
\big| (g^{ab}X^m \bnab_m \bnab_a V, \bnab_b V)_{\Lgg} \big|
&\lesssim
\| X \|_{H^3} \|V\|_{H^1}^2.
}}
Using this, Corollary \ref{corol:delu-mL-Sigma} and Corollary \ref{corol:Coercive-top-order-geom}, the first term of \eqref{eq:G4-344} is thus estimated as
\eq{\alg{\notag
\big|\big(g^{ab}X^m \bnab_m \bnab_a\delu \mL^{\ell-1}(\Sigma),\bnab_b\delu\mL^{\ell-1}(\Sigma)\big)_{\Lgg}\big|
&\lesssim
\|X\|_{H^3}\|\delu \mL^{\ell-1} \Sigma \|_{H^1}^2 
\\
&\lesssim
\Lambda(T) \Big(\Edelu
+ E^g_{N-1} + \Lambda(T)^2\Big).
}}
We then turn to the remaining terms in \eqref{eq:G4-344}, namely $G_4^e$. The first, second and fifth terms of $G_4^e$ are fairly straightfowardly estimated by 
\eq{\alg{\notag
\| X\|_{H^3} \| \delu \mL^{\ell-1}\Sigma\|_{H^1} + \| \delu X\|_{H^3} \| \mL^{\ell-1} \Sigma\|_{H^2}
&\lesssim
\Lambda(T) \Big(\| \delu \mL^{\ell-1}\Sigma\|_{H^1} + \|  \Sigma\|_{H^{N-1}}\Big)
\\&
\lesssim
\Lambda(T) \Big((\Edelu)^{1/2}
+ (E^g_{N-1})^{1/2} + \Lambda(T)\Big).
}}
The third and fourth terms of  $G_4^e$ merely require a careful counting of derivatives. As an example, the fourth term of $G_4^e$ is estimated by
\eq{\alg{\notag
\| \delu\big( X^m [\mL^{\ell-1}, \bnab_m]  (\Sigma)\big)\|_{H^1}
&\lesssim \| \delu X\|_{H^3} \| [\mL^{\ell-1},\bnab]\Sigma\|_{H^1} + \|X\|_{H^3} \| \delu[\mL^{\ell-1},\bnab]\Sigma\|_{H^1}.
}}
The first commutator term here can be studied using Lemma \ref{lem:Fluid-Commutators}, in particular the same ideas as in \eqref{eq:Comm-mLk-bnab-L2} yield
\eq{\alg{\notag
\| [\mL^{\ell-1}, \bnab_m]\Sigma \|_{H^1} 
&\lesssim 
\| g-\gamma\|_{H^{N-3}} \| \Sigma\|_{H^{N-2}} + \|\Sigma\|_{H^{N-4}} \lesssim (E^g_{N-1})^{1/2}.
}}
While we use \eqref{eq:Comm-mL-bnab} and the same ideas in the proof of Proposition \ref{prop:delu-commuted-lapse} to estimate the second commutator term  by
\eq{\alg{\notag
\| \delu[\mL^{\ell-1},\bnab]\Sigma\|_{H^1}
&\lesssim
\sum_{\substack{|I|+|J|\leq 2\ell-1\\|I|\geq 1, |J|\geq 2}} \|  \delu(\bnab^{I}g \bnab^J \Sigma) \|_{L^2} + 
\sum_{|I|+|J|\leq 2(\ell-2)} \|  \delu(\bnab^I \Riem[\gamma]\ast \bnab^J \Sigma) \|_{L^2}
\\&
\lesssim (\Edelu)^{1/2}
+ (E^g_{N-1})^{1/2} + \Lambda(T).
}}
Finally, for the last term of $G_4^e$ we use the original commutator identity \eqref{eq:Commutators-with-Delu} to estimate it by
\eq{\alg{\notag
\| X\|_{H^3} \|[\delu, \bnab]\mL^{\ell-1}\Sigma\|_{L^2}
&\lesssim
\Lambda(T) \Big( |\tau| \| u\|_{H^3} \| \mL^{\ell-1}\Sigma\|_{H^1} + \| \bnab u^0 \|_{H^2} \| \p_T \Sigma\|_{H^{2(\ell-1)}}\Big)
\\&\lesssim
\Lambda(T) \Big(
|\tau| \| u\|_{H^{N-1}} (E^g_{N-1})^{1/2} + \Lambda(T)(E^g_{N-1})^{1/2} + \Lambda(T)^2\Big).
}}
In summary, we obtain
\eq{\notag
| G_4^3| \lesssim
\Lambda(T)\Edelu
+ \Lambda(T) E^g_{N-1} +  \Lambda(T)^2.
}

We now turn to the second term on the RHS of \eqref{eq:G4-834}, and estimate it by
\eq{\alg{\notag
\| \delu \mL^{\ell-1}\Sigma\|_{L^2} \sum_{\substack{|I|+|J|\leq N-2\\ |J| \geq 1}} \|  \delu(\bnab^{I}X^m \bnab^J \Sigma) \|_{L^2}
\lesssim \Lambda(T)\Edelu
+ \Lambda(T) E^g_{N-1} +  \Lambda(T)^2.
}}

We finally turn to the last term on the RHS of \eqref{eq:G4-834} and use the identity \eqref{eq:Comm-delu-mL-k} on the first factor
\eq{\notag [\delu,\mL]\mL^{\ell-1}(X^a \bnab_a \Sigma).
}
Using the schematic identity given below \eqref{eq:Comm-delu-mL-k}, we see certain problematic terms arise when all the derivatives miss the shift and hit $\Sigma$:
\eq{\alg{\notag
(\delu g^{-1}) \ast X \ast \bnab^2 \bnab^{N-2} \Sigma +  \bnab u^0 \ast X\ast \bnab \bnab^{N-2}\p_T  \Sigma 
+\tau  \bnab u^c \ast X \ast \bnab^2 \bnab^{N-2} \Sigma .
}}
So, just as for the first term in $G_3^e$ given in the proof of Proposition \ref{prop:G3}, we need to integrate by parts using \eqref{eq:thirdIBP-2} on these terms. Given the close similarities, we omit the details and just state the result:
\eq{\notag
\Big|\big([\delu,\mL] \mL^{\ell-1}(X^a\bnab_a \Sigma),\delu\mL^{\ell-1}(\Sigma)\big)_{\Lgg}\Big|
\lesssim
 \Lambda(T) \Edelu  + \Lambda(T) E^g_{N-1}+ \Lambda(T)^2.
 }
 Note, however, that the above estimate is better than in \eqref{eq:G3-comm-est-455}. This is because of the additional shift term appearing which always gives us additional decay (unlike the lapse which needs at least one derivative acting on it).  

Finally we turn to $G_8$ and remark that 
\eq{\notag
\Big|\big(\delu\mL^{\ell-1}(X^m \bnab_m v),\delu\mL^{\ell}(h)\big)_{L^2_{g,\gamma}}\Big|
\lesssim  \sum_{\substack{|I|+|J|\leq N-2\\ |J| \geq 1}} \|  \delu(\bnab^{I}X^m \bnab^J \Sigma) \|_{L^2}
\| \delu \mL^{\ell}h\|_{L^2}
}
and this is easily estimated using previous ideas. 
\end{proof} 

\subsection{Estimates on $\mathbf{G_c}$ and $\mathbf{G_{cc}}$}\label{sec:est-Gc}
In this final part of the section we prove estimates on the commutator error terms $G_c$ and $G_{cc}$. We first prove a preliminary lemma concerning the commutator between the operators $\p_T$ and $\delu \mL$, and then use this lemma in the subsequent proposition. 

\begin{lem}
\label{lem:Gc-prelim-est}
We have,
\eq{\alg{\notag
\| [\delu\mL^{\ell},\p_T]h \|_{L^2} + \| [\delu\mL^{\ell-1},\p_T]\Sigma \|_{H^1}  
&\lesssim 
\Edelu + E^g_{N-1} + \Lambda(T).
}}
\end{lem}

\begin{proof}
Let $V$ be a $(0,2)$-tensor on $M$.
A computation yields
\eq{\alg{\label{eq:Comm-pt-delu-mL}
[\p_T, \delu \mL]V_{ij} &= - \p_T u^0 \cdot \mL \p_T V_{ij}
+ \p_T u^0 \p_T g^{ab} \bnab_a \bnab_b V_{ij}
- \tau u^c \bnab_c (\mL V_{ij})
\\&\quad 
+ \tau \p_T u^c \bnab_c (\mL V_{ij})
- \delu (\p_T g^{ab})\bnab_a\bnab_b V_{ij}
- u^0 \p_T g^{ab} \cdot \bnab_a\bnab_b \p_T V_{ij}
\\&\quad
+\tau (\p_T g^{ab})  u^c \bnab_c\bnab_a\bnab_b V_{ij} .
}} 
Using Sobolev embedding we find 
\eq{\alg{\notag
\| [\p_T, \delu \mL]V \|_{L^2}
&\lesssim 
\Big(\| \p_T \hu \|_{H^2} \| \p_T h\|_{H^{2}}+ \| \delu\p_T (g^{-1})\|_{H^2} \Big) \| V\|_{H^{2}} 
\\&\quad
+
|\tau|\Big(\| u\|_{H^2}  + \| \p_T u^c\|_{H^2} \Big)\| V\|_{H^{3}} 
+ \Big( \| \p_T \hu \|_{H^2} + \| \p_T g \|_{H^2} \Big)\| \p_T V\|_{H^2}.
}}
We use \eqref{eq:EoM-hv} to schematically compute 
\eq{\notag
\delu \p_T g^{-1} = \delu ( g^{-1} g^{-1} \p_T h )= \delu \Sigma +\delu \hN  + \delu(F_h) + \text{h.o.t}.
}
Using Lemma \ref{lem:delu-N-X}, Corollary \ref{corol:Estimates-delu-h-v} and  Corollary \ref{corol:Est-delu-Fh-Fv}, we obtain
\eq{\notag 
\| \delu\p_T (g^{-1})\|_{H^2}
\lesssim \| \delu \Sigma \|_{H^2} +\| \delu \hN \|_{H^2} + \|\delu(F_h) \|_{H^2} + \text{h.o.t. } \lesssim (E^g_{N-1})^{1/2} + \Lambda(T).
}
So, by applying the estimates from \eqref{eq:F-uj-est-2}, \eqref{eq:est-pT-fluid-compts} and Lemma \ref{lem:Estimates-pT-h-v},
\eq{\alg{\notag
\| [\p_T, \delu \mL]V \|_{L^2}
&\lesssim 
\Big( |\tau|\|u\|_{H^{2}} + |\tau|\| F_{u^a} \|_{H^{2}}+ \Lambda(T)
+ (E^g_{N-1})^{1/2} \Big)\| V\|_{H^{3}} 
\\&\quad
+ \big(E^g_{N-2} + \Lambda(T)\big) \| \p_T V\|_{H^2}
\\
&\lesssim 
\Big( |\tau|\|u\|_{H^{2}} + (E^g_{N-1})^{1/2} + \Lambda(T)
\Big)\| V\|_{H^{3}} 
+ \big(E^g_{N-2} + \Lambda(T)\big) \| \p_T V\|_{H^2}.
}}
Thus, from Lemma \ref{lem:Estimates-pT-h-v}, Lemma \ref{lem:Fluid-Commutators} and Corollary \ref{corol:Coercive-top-order-geom},
\eq{\alg{\notag
\| [\p_T, \delu \mL]\mL^{\ell-1} h\|_{L^2}
&\lesssim 
\Big( |\tau|\|u\|_{H^{2}} + (E^g_{N-1})^{1/2} + \Lambda(T)
\Big)\| g-\gamma\|_{H^{N}} 
+ \big(E^g_{N-2} + \Lambda(T)\big) \| \p_T h\|_{H^{N-1}}.
\\&
\lesssim
\Edelu + E^g_{N-1} + \Lambda(T).
}}

In addition, we note that at higher order $s \in \mathbb{Z}_{\geq 1}$, we have the identity
\eq{\alg{\label{eq:Comm-pt-delu-mL-k}
[\p_T, \delu \mL^s]V_{ij} &= [\p_T, \delu \mL](\mL^{s-1} V_{ij})
	+ (\delu g^{ab}) \bnab_a \bnab_b ([\p_T, \mL^{s-1}] V_{ij})
\\&\quad
	+ g^{ab} \delu  \bnab_a \bnab_b ([\p_T, \mL^{s-1}] V_{ij}).
}}
Using this and the commutator Lemma \ref{lem:Fluid-Commutators}, we find
\eq{\alg{\notag
\| [\p_T, \delu \mL^\ell] h \|_{L^2}
&\lesssim
\|[\p_T, \delu \mL](\mL^{\ell-1} h)\|_{L^2}
+ \| \delu g \|_{H^2} \|[\p_T, \mL^{\ell-1}] h\|_{H^2}
+ \|  g^{ab} \delu  \bnab_a \bnab_b ([\p_T, \mL^{\ell-1}] h) \|_{L^2}
\\
&\lesssim
\Edelu + E^g_{N-1} + \Lambda(T)
+ \|  g^{ab} \delu  \bnab_a \bnab_b ([\p_T, \mL^{\ell-1}] h) \|_{L^2}.
}}
So it remains to control this last term, which we do so using the commutator identity \eqref{eq:Commutator-dT-mLk}
\eq{\alg{\notag
\| g^{ab}\delu \bnab_a \bnab_b \big( [\p_T, \mL^{\ell-1}]h\big) \|_{L^2}
&\lesssim
\sum_{i=1}^{\ell-1} \| g^{ab}\delu \bnab_a \bnab_b \Big( \mL^{i-1}  (\p_T g^{a_i b_i})\cdot \bnab_{a_i}\bnab_{b_i} \big( \mL^{\ell-1-i} (h)\big)\Big)\|_{L^2}
\\&
\lesssim \sum_{\substack{|I|+|J|\leq 2\ell \\ |J|\geq 2}} \| \delu (\bnab^{I} \p_T g) \bnab^{J}h\|_{L^2} +  \|\bnab^{I} \p_T g \cdot \delu(\bnab^{J}h)\|_{L^2}
\\&\lesssim
\Edelu
+ E^g_{N-1} + \Lambda(T)^2,
}}
where in the final line we used \eqref{eq:EoM-hv} to replace $\p_T g$ and then distributed derivatives and applied   Lemma \ref{lem:Estimates-pT-h-v}, Corollary \ref{corol:delu-nab-estimates} and Lemma \ref{lem:deluV-H2-H1} as needed. 

The $\Sigma$ estimate follows in exactly the same way. 
\end{proof}

\begin{prop}
\label{prop:Gc}
We have,
\eq{\alg{\notag
|G_{c}| + |G_{cc}| \lesssim 
(\Etot)^{3/2}+ \Lambda(T) (\Edelu)^{1/2}.
}}
\end{prop}

\begin{proof}
There are three terms to estimate in $G_c$. 
By Corollary \ref{corol:Estimates-delu-h-v} and Lemma \ref{lem:Gc-prelim-est}, the first term is easily estimated as
\eq{\alg{\notag
\Big| \big([\delu\mL^{\ell}, \p_T](h),\delu\mL^{\ell}(h) \big)_{\Lgg} \Big|
&\lesssim
\| [\delu\mL^{\ell}, \p_T](h)\|_{\Lgg}\| \delu\mL^{\ell}(h)\|_{\Lgg}
\lesssim
(\Etot)^{3/2}+ \Lambda(T) (\Edelu)^{1/2}.
}}
The terms appearing in $G_{cc}$ are similarly easily estimated:
\eq{\alg{\notag
\Big|\big(&[\p_T, \delu\mL^{\ell-1}](v),\delu\mL^{\ell}(h)\big)_{L^2_{g,\gamma}}\Big|
+\Big|\big( \delu \mL^{\ell-1} (v) ,[\p_T,\delu\mL^{\ell}] (h)\big)_{L^2_{g,\gamma}}\Big|
\\
&\lesssim
\| [\delu\mL^{\ell-1}, \p_T](\Sigma)\|_{\Lgg}\| \delu\mL^{\ell}(h)\|_{\Lgg}
+
\| [\delu\mL^{\ell}, \p_T](h)\|_{\Lgg}\| \delu\mL^{\ell-1}(\Sigma)\|_{\Lgg}
\\
&\lesssim
(\Etot)^{3/2}+ \Lambda(T) (\Etot)^{1/2} + \Lambda(T)^2.
}}

For the second term of $G_c$, we integrate it by parts using \eqref{eq:Gi-102} to obtain
\eq{\alg{\notag
 \big(\delu &\mL^{\ell} (v) ,[\delu\mL^{\ell-1},\p_T] (v)\big)_{\Lgg} 
= 
\big(g^{ab}\bnab_a \delu \mL^{\ell-1} (v) ,\bnab_b [\delu\mL^{\ell-1},\p_T] (v)\big)_{\Lgg}
\\&\quad
- 2\big( \Riem[\gamma]\circ \delu \mL^{\ell-1} (v), [\delu\mL^{\ell-1},\p_T] (v)\big)_{\Lgg}
+ \big([\delu,\mL] \mL^{\ell-1} (v) ,[\delu\mL^{\ell-1},\p_T] (v)\big)_{\Lgg}.
}}
So by the estimate \eqref{eq:delu-mL-mL-Sigma}, Corollary \ref{corol:delu-mL-Sigma}, Corollary \ref{corol:Coercive-top-order-geom} and Lemma \ref{lem:Gc-prelim-est},
\eq{\alg{\notag
\Big| \big(\delu \mL^{\ell} (v) ,[\delu\mL^{\ell-1},\p_T] (v)\big)_{\Lgg} \Big|
&\lesssim 
\big( \| \delu \mL^{\ell-1}\Sigma \|_{H^1} + 
\|[\delu, \mL]\mL^{\ell-1}\Sigma\|_{L^2}\big) \| [\delu\mL^{\ell}, \p_T](\Sigma)\|_{H^1}
\\
&\lesssim
(\Etot)^{3/2}+ \Lambda(T) (\Etot)^{1/2} + \Lambda(T)^2.
}}

For the final term of $G_c$, namely 
\eq{\notag
\big( [\delu \mL^{\ell},\p_T](\Sigma) ,\delu\mL^{\ell-1} (\Sigma) \big)_{\Lgg},
}
we need to integrate certain terms by parts since some of the factors contain too many derivatives on $\Sigma$. For example, using \eqref{eq:Comm-pt-delu-mL} and  \eqref{eq:Comm-pt-delu-mL-k}, we see that two such terms are
\eq{\notag
\big( \p_T u^0 \mL \p_T (\mL^{\ell-1}\Sigma) ,\delu\mL^{\ell-1} (\Sigma) \big)_{\Lgg} + \big( \tau u^c \bnab_c (\mL^\ell \Sigma),\delu\mL^{\ell-1} (\Sigma) \big)_{\Lgg}.
}
Using the integration by parts estimate \eqref{eq:thirdIBP-2} and \eqref{eq:est-pT-fluid-compts} we find
\eq{\alg{\notag
\Big|\big(\p_T u^0 \mL \p_T (\mL^{\ell-1}\Sigma),\delu\mL^{\ell-1} (\Sigma) \big)_{\Lgg}\Big|
&\lesssim
\| \p_T u^0 \|_{H^3} \| \p_T \Sigma\|_{H^{N-2}} \| \delu\mL^{\ell-1}\Sigma\|_{H^1}
\\&
\lesssim
\Lambda \Etot + (\Etot)^2 + \Lambda(T)^3.
}}
Similarly,
\eq{\alg{\notag
\Big|\big( \tau u^c \bnab_c (\mL^\ell \Sigma),\delu\mL^{\ell-1} (\Sigma) \big)_{\Lgg}\Big|
&\lesssim
|\tau| \| u \|_{H^3} \|  \Sigma\|_{H^{N-1}} \| \delu\mL^{\ell-1}\Sigma\|_{H^1}
\lesssim
|\tau| \| u \|_{H^{N-1}}  \Etot + \Lambda(T)^2.
}}
All the remaining terms can be controlled using the ideas from before. We eventually obtain 
\eq{\alg{\notag
\Big| \big( [\delu \mL^{\ell},\p_T]( v) ,\delu\mL^{\ell-1} (v) \big)_{\Lgg}\Big|
&\lesssim 
\Big( |\tau| \| u\|_{H^{N-1}}+ \Lambda(T) \Big)\Etot
+ \Lambda(T) (\Edelu)^{1/2}
\\&\quad
+(\Etot)^{3/2}+ \Lambda(T)^2.
}}

\end{proof}

\section{Energy estimates for the dust variables}\label{sec:EnergyEstMatter}
In this section we establish energy estimates for the dust variables $\rho$ and $u^a$ in two propositions. An integration by parts identity,  where due to symmetry a high-derivative term can be brought back onto the LHS, plays a crucial role in both propositions.  We write the argument out explicitly for the first proposition, although note in principle the argument is just as in  Lemma \ref{lem:IBP}, \eqref{eq:symIBP}.

\begin{defn}[$\mathcal{E}_k{[}\rho{]}$]
Define the following functionals
\eq{\alg{\notag
E_k[\rho](T) &\define \frac12 \int_M |\nabla^k \rho(T, \cdot)|^2 \mu_g,
\qquad
\mathcal{E}_k[\rho](T) &\define \sum_{0\leq \ell \leq k} E_\ell[\rho](T).
}}
\end{defn}

\begin{prop}[Time evolution of $\mathcal{E}_{N-2}{[}\rho{]}$ for fluid energy density]
\label{prop:del-T-E-rho}
We have,
\eq{\notag
|\p_T \mathcal{E}_{N-2}[\rho](T)| \lesssim \varepsilon e^{\max\{-1+\mu,-\lambda\} T} \mathcal{E}_{N-2}[\rho](T).
}
\end{prop}

\begin{proof}
Recall from \eqref{eq:EoM-fluid} that the equation of motion for $\rho$ is
\eq{\alg{\notag
\p_T \rho
&= (u^0)^{-1}\rho F_\rho+ \tau (u^0)^{-1} u^a \nabla_a \rho.
}}
Let $3 \leq k \leq N-2$.
Using \eqref{eq:hdelIntegral} and the integration by parts estimate \eqref{eq:thirdIBP-2}, we have
\eq{\alg{\notag
\p_T E_k[\rho](T) &\lesssim \| \hN \|_\infty E_k[\rho](T)  + \int_M \p_T (|\nabla^k \rho|^2)\mu_g + \int_M X^c \bnab_c (|\nabla^k \rho|^2) \mu_g
\\
&\lesssim
\Big( \| \hN \|_{H^2} + \| X\|_{H^3} \Big) E_k[\rho](T)
+\int_M  \p_T \Big( g^{a_1 b_1} \cdots g^{a_k b_k}\Big) (\nabla^{a_1}\cdots \nabla^{a_k}\rho)(\nabla_{a_1}\cdots \nabla_{a_k}\rho )\mu_g
\\
& \quad+ \int_M  g^{a_1 b_1} \cdots g^{a_k b_k}\Big( \p_T \nabla_{a_1}\cdots \nabla_{a_k}\rho\Big) (\nabla_{b_1}\cdots \nabla_{b_k}\rho )\mu_g
\\
&\lesssim
\Big( \| \hN \|_{H^2} + \| X\|_{H^3} + \| \Sigma\|_{H^2} \Big) E_k[\rho](T)
+ I_k
}}
where we have defined
\eq{\notag
I_k \define \int_M  g^{a_1 b_1} \cdots g^{a_k b_k}\Big( \p_T \nabla_{a_1}\cdots \nabla_{a_k}\rho\Big) (\nabla_{b_1}\cdots \nabla_{b_k}\rho )\mu_g.
}

We focus on the integrand of $I_k$, commuting the derivatives to find
\eq{\alg{\notag
\p_T \nabla_{a_1}\cdots \nabla_{a_k}\rho
&= 
	\nabla_{a_1}\cdots \nabla_{a_k}\Big( \frac{\rho}{u^0} F_\rho+ \tau \frac{u^i}{u^0}  \nabla_i\rho \Big)+ [\p_T, \nabla_{a_1}\cdots \nabla_{a_k}]\rho.
}}
We first look at the term $\tau u^i \nabla_i\rho$  which contains the most number of derivatives on $\rho$. 
By commuting and integration by parts, this term becomes:
\eq{\alg{\notag
\int_M  &\tau \frac{u^i}{u^0} \big( \nabla_{a_1}\cdots \nabla_{a_k}\nabla_i\rho \big) (\nabla^{a_1}\cdots \nabla^{a_k}\rho )\mu_g
\\
&=  
\int_M  \tau \frac{u^i}{u^0} \big( [\nabla_{a_1}\cdots \nabla_{a_k},\nabla_i]\rho \big) (\nabla^{a_1}\cdots \nabla^{a_k}\rho )\mu_g
- \int_M  \tau \nabla_i \big(\frac{u^i}{u^0} \big) \big( \nabla_{a_1}\cdots \nabla_{a_k}\rho \big) (\nabla^{a_1}\cdots \nabla^{a_k}\rho )\mu_g
\\
& \quad - \int_M \tau \frac{u^i}{u^0}  \big( \nabla^{a_1}\cdots \nabla^{a_k}\rho \big) ([\nabla_i ,\nabla_{a_1}\cdots \nabla_{a_k}]\rho )\mu_g
- \int_M \tau\frac{u^i}{u^0}  \big( \nabla^{a_1}\cdots \nabla^{a_k}\rho \big) (\nabla_{a_1}\cdots \nabla_{a_k} \nabla_i \rho )\mu_g.
}}
Rearranging, we find
\eq{\alg{\notag
\int_M  &\tau \frac{u^i}{u^0}  \big( \nabla_{a_1}\cdots \nabla_{a_k}\nabla_i\rho \big) (\nabla^{a_1}\cdots \nabla^{a_k}\rho )\mu_g
\\
&=  
\int_M  \tau \frac{u^i}{u^0} \big( [\nabla_{a_1}\cdots \nabla_{a_k},\nabla_i]\rho \big) (\nabla^{a_1}\cdots \nabla^{a_k}\rho )\mu_g
- \frac12 \int_M  \tau \nabla_i \big(\frac{u^i}{u^0} \big)|\nabla^k \rho|^2 \mu_g.
}}

We see that we need to study two commutator error terms. Using \eqref{eq:Commutator-pT-nablas}, and noting that $\rho$ is a scalar, we see the first error term takes the form
\eq{\notag
|[\p_T,\nabla_{a_1}\cdots \nabla_{a_k}]\rho|
\lesssim 
\sum_{\substack{|I|+|J|= k-1 \\ |J|\geq 1}} |\nabla^{I} (\p_T \Gamma[g] ) ||\nabla^{J} \rho|.
}
Thus, using \eqref{eq:est-pt-Gamma},
\eq{\notag
\Big| \int_M \big( [\p_T, \nabla_{a_1}\cdots \nabla_{a_k}]\rho \big) (\nabla^{a_1}\cdots \nabla^{a_k}\rho )\mu_g\Big|
\lesssim \| \rho\|_{H^{k}}^2 \| \p_T \Gamma[g]\|_{H^{N-2}} 
\lesssim \varepsilon e^{-\lambda T}\mathcal{E}_k[\rho](T).
}

For the second commutator error term, we use \eqref{eq:Commutator-nabla-nablas} to find
\eq{\notag
|[\nabla_i,\nabla_{a_1}\cdots \nabla_{a_k}]\rho |\lesssim 
\sum_{\substack{|I|+|J| = k-1\\ |J| \geq 1}}| \nabla^{I} \text{Riem} || \nabla^{J} \rho|,
}
and so we have
\eq{\notag
\Big| \int_M  \tau u^i \big( [\nabla_{a_1}\cdots \nabla_{a_k},\nabla_i]\rho \big) (\nabla^{a_1}\cdots \nabla^{a_k}\rho )\mu_g\Big|
\lesssim |\tau| \| u \|_{H^2} \| \text{Riem} \|_{H^{k-2}}\| \rho\|_{H^{k}}^2
\lesssim 
\varepsilon e^{(-1+\mu)T}
\mathcal{E}_k[\rho](T).
}

Putting this all together, and using Lemma \ref{lem:Fluid-Source-Estimates}, we obtain
\eq{\alg{\notag
|I_k| & \lesssim
\| (u^0)^{-1} \rho \|_{H^k} \|\rho\|_{H^k} \| F_\rho \|_{H^k} + |\tau| \| \rho \|_{H^k} \| \nabla (\tfrac{u^i}{u^0}) \cdot \nabla_i \rho \|_{H^{k-1}} + |\tau| \| \tfrac{u^i}{u^0}\|_{H^3} E_k[\rho](T)
\\
& \quad +\varepsilon e^{\max\{-1+\mu,-\lambda\} T}
\mathcal{E}_k[\rho](T).
\\
&\lesssim
	\varepsilon e^{\max\{-1+\mu,-\lambda\} T} \mathcal{E}_k[\rho](T).
}}
By summing over $k$ we can conclude
\eq{\notag
|\p_T \mathcal{E}_{N-2}[\rho](T)| \lesssim \varepsilon e^{\max\{-1+\mu,-\lambda\} T}\mathcal{E}_{N-2}[\rho](T).
}
\end{proof}

\begin{defn}[$\mathcal{E}_k{[}u{]}$]
Define the functionals
\eq{\alg{\notag
E_k[u](T) &\define \frac12 \int_M |\nabla^k u(T, \cdot)|^2 \mu_g,
\qquad
\mathcal{E}_k[u](T) &\define \sum_{0\leq \ell \leq k} E_\ell[u](T).
}}
\end{defn}

\begin{prop}[Evolution of $\mathcal{E}_{N-1}{[}u{]}$ fors spatial fluid velocity components]
\label{prop:del-T-E-u} 
The following estimate holds,
\eq{\alg{\notag
|\p_T \mathcal{E}_{N-1}[u](T)| 
&\lesssim \varepsilon e^{(-1+\mu)T} \mathcal{E}_{N-1}[u](T)  + |\tau|^{-1} \Lambda (T)  (\mathcal{E}_{N-1}[u](T))^{1/2}.
}}
\end{prop}

\begin{proof}
Let $3 \leq k \leq N-1$. 
The proof is similar to the estimate for $\mathcal{E}_k[\rho](T)$. As in the proof of Proposition \ref{prop:del-T-E-rho}, we obtain
\eq{\alg{\notag
\p_T E_k[u](T) 
&\lesssim
\Big( \| \hN \|_{H^2} + \| X\|_{H^3} + \| \Sigma\|_{H^2} \Big) E_k[u](T)
+ I_k',
}}
where we have defined
\eq{\notag
I'_k \define \int_M  g_{ij}g^{a_1 b_1} \cdots g^{a_k b_k}\big( \p_T \nabla_{a_1}\cdots \nabla_{a_k}u^i\big) (\nabla_{b_1}\cdots \nabla_{b_k}u^j)\mu_g.
}
The integration by parts analysis on $I_k'$ follows unchanged to $I_k$. The main difference now arises in the commutator estimates since $u^a$ is a vector while $\rho$ is a scalar. 
For example, we now have an error term of the form
\eq{\notag
|[\p_T,\nabla_{a_1}\cdots \nabla_{a_k}]u^i |
\lesssim 
\sum_{|I|+|J|= k-1} |\nabla^{I} (\p_T \Gamma[g] )|| \nabla^{J} u^i|.
}
Thus, using \eqref{eq:est-pt-Gamma},
\eq{\notag
\Big| \int_M \big([\p_T, \nabla_{a_1}\cdots \nabla_{a_k}]u^i \big) (\nabla^{a_1}\cdots \nabla^{a_k}u_i )\mu_g\Big|
\lesssim \| u^a\|_{H^{k}}^2 \| \p_T \Gamma(g)\|_{H^{N-2}}
\lesssim \varepsilon e^{-\lambda T} \mathcal{E}_k[u^a](T) .
}

Similarly the second commutator error term
looks like
\eq{\notag
|[\nabla_i,\nabla_{a_1}\cdots \nabla_{a_k}]u^i |
\lesssim \sum_{|I|+|J| = k-1} |\nabla^{I} \text{Riem} || \nabla^{J} u^i|,
}
and so we have
\eq{\notag
\Big| \int_M  \tau u^a \big( [\nabla_{a_1}\cdots \nabla_{a_k},\nabla_a]u^i \big) (\nabla^{a_1}\cdots \nabla^{a_k}u_i )\mu_g\Big|
\lesssim |\tau|\| u\|_{H^2} \| \text{Riem} \|_{H^{k-1}}\| u\|_{H^{k}}^2
\lesssim \varepsilon e^{(-1+\mu) T}\mathcal{E}_k[u](T) .
}

Recalling from \eqref{eq:EoM-fluid} the equation of motion
\eq{\notag
\p_Tu^j
= \tau (u^0)^{-1} u^i \nabla_i u^j + 
\tau^{-1}u^0 N\nabla^j N+(u^0)^{-1}F_{u^j},
}
we obtain
\eq{\alg{\notag
|I'_k| & \lesssim
|\tau| \| (u^0)^{-1}u^a \|_{H^3}E_k[u](T) 
+
|\tau| \| u \|_{H^k} \| \nabla (\tfrac{u^i}{u^0}) \cdot \nabla_i u^a \|_{H^{k-1}}
+ |\tau|^{-1} \| \hN\|_{H^{k+1}} \|u\|_{H^k} \\&\quad
+ \| (u^0)^{-1}F_{u^j} \|_{H^k}\|u\|_{H^k} 
 +\varepsilon e^{\max\{-1+\mu,-\lambda\} T} E_k[u](T) 
\\
&\lesssim \varepsilon e^{(-1+\mu)T} E_k[u](T) + \Big( |\tau|^{-1} \| \hN\|_{H^{k+1}}  + \| F_{u^j}\|_{H^k} \Big)E_k[u](T)^{1/2}.
}}
The conclusion then follows by summing over $k$ and using the source estimates for $F_{u^j}$ given in Lemma \ref{lem:Fluid-Source-Estimates}. 
\end{proof}

\section{Proof of Theorem \ref{thm-1}}\label{sec:EndBoot}
In this final section we bring together our main estimates to conclude the bootstrap argument. One standard, yet technical, aspect of the argument is deferred to Appendix \ref{appendix:Bootstrap}. 

\begin{proof}
 Smallness of the initial data guarantees the existence of a constant $C_0>0$ such that
\eq{\alg{\notag
\| &\hN \|_{H^N} + \| X\|_{H^N} + \| \p_T N \|_{H^{N-1}} + \| \p_T X \|_{H^{N-1}} 
\\&+ (\Eg_{N-1}(T))^{1/2}
+ (\mathcal{E}^g_{\delu,N-1}(T_0))^{1/2}
+ \mathcal{E}_{N-2}[\rho](T_0)^{1/2} 
+\mathcal{E}_{N-1}[u^a](T_0)^{1/2}\leq C_0 \varepsilon.
}}

Proposition \ref{prop:del-T-E-rho}, together with Gr\"onwall's inequality, implies 
\eq{\notag
\| \rho\|_{H^{N-2}} \leq \mathcal{E}_{N-2}[\rho](T_0)^{1/2} \exp \Big( \int_{T_0}^T C\varepsilon e^{\max\{-1+\mu,-\lambda\} s}  \di s \Big)^{1/2}
\leq C_0 \varepsilon \exp(C' \varepsilon). 
}
Thus by Lemma \ref{lem:lapse-shift}, and shrinking $\varepsilon$ as needed,
\eq{\alg{\notag
\| \hN \|_{H^{N}} + \| X \|_{H^N}   & \leq	|\tau| C\| \rho \|_{H^{N-2}} + \varepsilon^2 Ce^{\max\{-2\lambda, -2+\mu\}T} 
\leq C C_0 \varepsilon e^{-T}.
}}
Turning next to Lemma \ref{lem:delT-lapse-shift}, we find
\eq{\notag
\| \p_T N \|_{H^{N-1}}  + \| \p_T X \|_{H^{N-1}} 
\leq 
	C \| \hN\|_{H^N}+C \| X\|_{H^N} 
+ |\tau|C \| \rho \|_{H^{N-2}} + \varepsilon^2 C e^{\max\{-2\lambda, -2+\mu\}T}
\leq C_0 \varepsilon e^{-T}. 
}

Next, we divide the estimate in Proposition \ref{prop:del-T-E-u} by the square root of the energy, to obtain
\eq{\notag
|\p_T \mathcal{E}_{N-1}[u](T)^{1/2}| 
\leq  C \varepsilon e^{(-1+\mu)T} \mathcal{E}_{N-1}[u](T)^{1/2} + C C_0 \varepsilon.
}
An application of Gr\"onwall's inequality then produces
\eq{\alg{\notag
\|u\|_{H^{N-1}} 
&\leq \Big(\mathcal{E}_{N-1}[u](T_0)^{1/2} + \int_{T_0}^T C C_0\varepsilon  \di s \Big) \exp  \Big(\int_{T_0}^TC \varepsilon e^{(-1+\mu)s}\di s\Big)
\\&
\leq 
C\big(C_0 \varepsilon + C_0\varepsilon (T-T_0) \big).
}}
Up to possibly redefining $T_0$, see the discussion around \eqref{eq:boot-T} in Appendix \ref{appendix:Bootstrap}, these inequalities now imply 
\eq{\notag
\Lambda (T) \leq C C_0 \varepsilon e^{-T}.
}

Finally adding together the results of Proposition \ref{prop-lower-order-geometry-est} and Theorem \ref{thm-top-order-geom} and using the above improved estimates, we find
\eq{\alg{\notag
\p_T \Etot
&\leq 
-2\alpha \Etot
+ CC_0 \varepsilon e^{-1+\mu} \Etot
+ CC_0 \varepsilon e^{-T} (\Etot)^{1/2}
+ C(\Etot)^{3/2}
+ C\Lambda(T)^2.
}}
Thus, by the bootstrap argument presented in Appendix \ref{appendix:Bootstrap}, 
we obtain
\eq{\notag
\Edelu+ E^g_{N-1} \leq \frac12 C_1^2 \varepsilon^2 e^{-2\lambda T},
}
where $C_1 \gg C_0$. 
Finally, as a consequence of Corollary \ref{corol:Coercive-top-order-geom},
\eq{\alg{\notag
\| g-\gamma\|_{H^N}^2 +\| \Sigma\|_{H^{N-1}}^2 &\lesssim 
\Edelu
+ E^g_{N-1} 
+ \Lambda(T)^2
\leq \frac34 C_1^2 \varepsilon^2 e^{-2\lambda T}.
}}
\end{proof}

\appendix
\section{Background geometry}\label{appendix:Background}
In this appendix we derive the background solutions given in Remark \ref{rem:BackgroundSolution}. 
On the Milne background the ADM variables induced on a $t_c=const$ slice are
\eq{\notag
(\phTmet_{ab}, \tilde{k}_{ab}, \tlapse, \tshift^a)|_B = (\tfrac{t_c^2}{9} \gamma_{ab}, -\tfrac{1}{t_c}\gamma_{ab}, 1, 0) , \quad \tau = \phTmet^{ab}\tilde{k}_{ab} = -\tfrac{3}{t_c}.
}
Condition \eqref{eq:FluidNormalisation} implies
$
(\tu^{t_c})^2 = 1+\tfrac{t_c^2}{9}\gamma(\tu, \tu).
$
The fluid equations of motion in cosmological coordinates $(t_c, x^i)$ read
\eq{\alg{\notag
\tu^{t_c}\p_{t_c} \ln \trho+\tu^i \p_i \ln \trho+3 \p_{t_c} \tu^{0}+ \p_i \tu^i+3t_c^{-1}\tu^{t_c}+\Gamma^i_{ij}[\gamma]\tu^j&=0, \\
 \tu^{0}\p_{t_c}\tu^{0} +\tu^i\p_i \tu^{0}-\tfrac{1}{t_c}\gamma(\tu, \tu) &=0,\\
\tu^{t_c}\p_{t_c}\tu^{i} +\tu^i\p_i \tu^{i}-2\tfrac{t_c}{9}\tu^{t_c}\tu^i+ \Gamma^i_{jk}[\gamma]\tu^j\tu^k&=0.
}}
Picking $\tu^{t_c}=1, \tu^i=0$ we have $
\p_{t_c} \ln \trho + 3t_c^{-1}= 0$. 
Thus on the background the fluid solution is
\eq{\notag
(\tu^{t_c}, \tu^i, \trho)|_B = (1, 0, \trho_0 t_c^{-3})
}
where $\trho_0>0$ is a constant. By rescaling the variables as in  Definition \ref{defn:rescaling} we arrive at Remark \ref{rem:BackgroundSolution}. 

\section{Matter variables}\label{appendix:Matter}
In this appendix we discuss the rescaled components of the energy momentum tensor given in Definition \ref{Def-resc-matter}. To see the origin of the various terms, we follow \cite{AF20} and introduce the following notation for the matter variables
\eq{\alg{\notag
\tilde{\EnergyDensity} &\define \tEMT^{\mu\nu}n_\mu n_\nu=\tEMT^{00}\tlapse^2, \qquad
&\tilde\jmath^a &\define-\tEMT^{\mu\nu}\perp^a{}_\mu n_\nu ,
\\
\tilde{\eta} &\define \tilde{E}+\phTmet^{ab}\tEMT_{ab},
\qquad
&\tilde{S}_{ab} &\define \tEMT_{ab}-\frac12 \text{Tr}_{\phFmet} \tEMT \cdot  \phTmet_{ab},
\qquad
\tEMT^{ab} \define \trho\tu^a\tu^b \phFmet^{ab}.
}}
Using \eqref{eq:Eulerform1}  we find 
\eq{\alg{\notag
\text{Tr}_{\phFmet} \tEMT &= \phFmet^{\mu\nu}\tEMT_{\mu\nu}=-\trho,  \\
\tEMT_{ab}  &= \phFmet_{a\mu} \phFmet_{b\nu} \tEMT^{\mu\nu} = \trho( \tshift_a \tu^0 + \phTmet_{ai}\tu^i)(\tshift_b \tu^0 + \phTmet_{bj}\tu^j ),
\\
\phTmet^{ab}\tEMT_{ab} 
&=  \trho(\tshift^a\tshift_a (\tu^0)^2 + 2 \tshift_b \tu^b \tu^0 + \phTmet_{ab}\tu^a\tu^b ).
}}
A calculation then yields
\eq{\alg{\notag
\tilde{\EnergyDensity} &=|\tau|^3 \rho  (u^0)^2 N^2 ,
\qquad
\tilde\jmath^a =|\tau|^5 \rho u^0 u^a N  ,
\\
\tilde{\eta} &= |\tau|^3 \rho (u^0)^2 N^2
+ |\tau|^3 \rho \Big(X_a X^a (u^0)^2 + 2\tau X_b u^b u^0 + \tau^2 g_{ab} u^a u^b \Big),
\\
\tilde{S}_{ab} &= |\tau|  \rho\Big( X_a X_b (u^0)^2 + \tau u^0 u^c X_b g_{ac} + \tau u^0 u^c X_a g_{bc} + \tau^2 g_{ac}g_{bd}u^cu^d\Big)
 +|\tau|  \frac{\rho}{2}g_{ab} ,
\\ 
\tEMT^{ab} &= |\tau|^7\rho u^au^b.
}}
Finally, we introduce the following rescaled variables 
\eq{\alg{\notag
\EnergyDensity &\define |\tau|^{-3} \tilde{E} ,
\quad
&\jmath^a &\define |\tau|^{-5} \tilde{\jmath}^a, \\
\eta &\define |\tau|^{-3}\tilde{\eta},
\quad
&S_{ab} &\define |\tau|^{-1} \tilde{S}_{ab},
\quad
\underline{T}^{ab} &\define |\tau|^{-7}\tEMT^{ab},
}}
which yield the expressions given in Definition \ref{Def-resc-matter}. 

\section{Bootstrap}
\label{appendix:Bootstrap}
In this appendix we detail how the bootstrap assumption for $\Etot (=f)$ is closed. 
Let $\lambda<1$ and $\mu \ll 1$ be fixed constants. 
Let $f:[T_0,T_*)\to [0,\infty)$  be a continuous function, where $T_0 \leq T_* \leq \infty$. Suppose $f(T_0)\leq C_0^2 \varepsilon^2$. Suppose also, that for each $T_0 \leq T<T_*$ the assumption $f(T) \leq C_1^2 \varepsilon^2 e^{-2\lambda T}$ allows us to show
\eq{\notag
\p_T f \leq - 2\alpha f + C C_0 \varepsilon e^{(-1+\mu)T} f + C C_0 \varepsilon e^{-T} f^{1/2} + C f^{3/2} + C C_{0}^2\varepsilon^2 e^{-2T}.
}
Since $\alpha \in [1-\delta_\alpha, 1]$ we pick a $\zeta$ such that $\lambda < \zeta <1$ and $\lambda < \alpha\zeta<\tfrac12(1+\lambda)$.\footnote{For example if $\lambda = 0.75, \mu = 0.001$ we can choose $\varepsilon$ sufficiently small that $\alpha = 0.98$ and then $\zeta = 0.8$ works.}  Define also $\beta>0$ to be the difference $\beta \define \alpha \zeta - \lambda$. We then have,
\eq{\alg{\notag
\p_T (e^{2\alpha \zeta T} f)
&\leq -2\alpha (1-\zeta) e^{2\alpha \zeta T}f + C C_0 \varepsilon e^{(-1+\mu+ 2\alpha \zeta )T} f + C C_0 \varepsilon e^{(-1+ 2\alpha \zeta) T} f^{1/2} + C f^{3/2} e^{2\alpha \zeta T}
\\&\quad 
+ C C_{0}^2\varepsilon^2 e^{(-2+ 2\alpha \zeta )T}
\\
&\leq -\Big( 2\alpha (1-\zeta) - C C_0 \varepsilon e^{(-1+\mu)T} - C f^{1/2}\Big) e^{2\alpha \zeta T}f + C C_0 \varepsilon e^{(-1+ 2\alpha \zeta) T} f^{1/2} 
+ C C_{0}^2\varepsilon^2.
}}
Substituting in the bootstrap assumption, and choosing $\varepsilon$ sufficiently small:
\eq{\alg{\notag
\p_T (e^{2\alpha \zeta T} f)
&\leq -\Big( 2\alpha (1-\zeta) - C C_0 \varepsilon - C C_1 \varepsilon \Big) e^{2\alpha \zeta T}f + C C_0 \varepsilon e^{(-1+ 2\alpha \zeta) T} f^{1/2} 
+ C C_{0}^2\varepsilon^2
\\
&\leq C C_1 C_0 \varepsilon^2 e^{(-1+ 2\alpha \zeta-\lambda) T} 
+ C C_{0}^2\varepsilon^2
\\
&\leq C C_1 C_0 \varepsilon^2.
}}
Thus, by Gr\"onwall's inequality
\eq{\alg{\notag
e^{2\alpha \zeta T} f &
\leq C_0^2 \varepsilon^2 +  \int_{T_0}^T C C_1 C_0 \varepsilon^2 \di s
\qquad \Rightarrow
f \leq \big(C C_0^2 \varepsilon^2  + C C_0 C_1 \varepsilon^2(T-T_0) \big)e^{-2\beta T}e^{-2\lambda T}.
}}
Let $T_0^*$ be such that for all $T\geq T_0^*$,
\eq{\label{eq:boot-T}
\big(C C_0^2  + C C_0 C_1 (T-T_0) \big)e^{-2\beta T} < \frac12 C_1^2.
}
We then \emph{a fortiori} choose $T_0\geq T_0^*$. So we have shown that for all $T'\in [T_0,T_*)$, if $f(T) \leq C_1^2 \varepsilon^2 e^{-2\lambda T}$ for each $T \in [T_0, T']$ then in fact $ f(T) \leq \frac12 C_1^2 \varepsilon^2 e^{-2\lambda T}$. Thus  $ f(T) \leq \frac12 C_1^2 \varepsilon^2 e^{-2\lambda T}$ for all $T \in [T_0, T_*)$.  


\end{document}